\documentclass[superscriptaddress,amsmath,longbibliography,amssymb,twocolumn]{revtex4-1}

\usepackage{complexity}
\usepackage{graphicx}
\usepackage{amsmath}
\usepackage{fullpage}
\usepackage{amsfonts}
\usepackage{amssymb}
\usepackage{dsfont}
\usepackage{amsthm}
\usepackage{tensor}
\usepackage{amsthm}
\usepackage{xcolor}
\usepackage{comment}
\usepackage{mathtools}
\usepackage{physics}
\usepackage{bbold}

\usepackage{appendix}
\usepackage[ruled,vlined]{algorithm2e}
\usepackage[colorlinks]{hyperref}
\hypersetup{
pdfstartview={FitH},
pdfnewwindow=true,
colorlinks=true,
linkcolor=blue,
citecolor=blue,
filecolor=blue,
urlcolor=blue}
\usepackage{microtype}
\usepackage{tikz}
\newtheorem{theorem}{Theorem}
\newtheorem{problem}{Problem}

\renewcommand{\polylog}{\mathrm{polylog}}
\renewcommand{\poly}{\mathrm{poly}}

\def\sl{{l}}

\def\cO{\mathcal{O}}

\def\cS{\mathcal{S}}

\def\cU{\mathcal{U}}

\def\cV{\mathcal{V}}
\def\cW{\mathcal{W}}

\def\one{{\mathchoice {\rm 1\mskip-4mu l} {\rm 1\mskip-4mu l} {\rm
1\mskip-4.5mu l} {\rm 1\mskip-5mu l}}}
\def\dqc1{\textsc{DQC1}}
\newcommand{\bM}{{\bf M}}
\newcommand{\bA}{{\bf A}}
\newcommand{\bK}{{\bf K}}
 \newcommand{\sket}[1]{| #1 \rangle}
 \newcommand{\sbra}[1]{\langle #1|}

\newcommand{\lam}{\Lambda}
\makeatletter
\newtheorem*{rep@theorem}{\rep@title}
\newcommand{\newreptheorem}[2]{%
\newenvironment{rep#1}[1]{%
 \def\rep@title{#2 \ref{##1}}%
 \begin{rep@theorem}}%
 {\end{rep@theorem}}}
\makeatother
\newreptheorem{theorem}{Theorem}
\newtheorem{lemma}[theorem]{Lemma}
\newreptheorem{lemma}{Lemma}

\newreptheorem{claim}{Claim}

\newreptheorem{definition}{Definition}

\newreptheorem{remark}{Remark}

\newreptheorem{corollary}{Corollary}

\newreptheorem{observation}{Observation}
\newcommand{\nn}{\nonumber \\}

\begin{document}
\title{Exponential quantum speedup in simulating coupled classical oscillators}

\author{
Ryan Babbush}
\affiliation{Google Quantum AI, Venice, CA, United States}
\author{
Dominic W.~Berry}
\affiliation{School of Mathematical and Physical Sciences,
Macquarie University, Sydney, NSW, Australia}
\author{
Robin Kothari}
\affiliation{Google Quantum AI, Venice, CA, United States}
\author{
Rolando D.~Somma}
\affiliation{Google Quantum AI, Venice, CA, United States}
\author{
Nathan Wiebe}
\affiliation{Department of Computer Science, University of Toronto, Toronto, ON, Canada}
\affiliation{Pacific Northwest National Laboratory, Richland, WA, United States}
\affiliation{Canadian Institute for Advanced Research, Toronto, ON, Canada}
\date{\today}

\begin{abstract}
We present a quantum algorithm for simulating the classical dynamics of $2^n$ coupled oscillators (e.g., $2^n$ masses coupled by springs). Our approach leverages a mapping between the Schr\"odinger equation and Newton's equation for harmonic potentials such that the amplitudes of the evolved quantum state encode the momenta and displacements of the classical oscillators. When individual masses and spring constants can be efficiently queried, and when the initial state can be efficiently prepared, the complexity of our quantum algorithm is polynomial in $n$, almost linear in the evolution time, and sublinear in the sparsity. As an example application, we apply our quantum algorithm to efficiently estimate the kinetic energy of an oscillator at any time. We show that any classical algorithm solving this same problem is inefficient and must make $2^{\Omega(n)}$ queries to the oracle and, when the oracles are instantiated by efficient quantum circuits, the problem is \BQP-complete.
Thus, our approach solves a potentially practical application with an exponential speedup over classical computers. Finally, we show that under similar conditions our approach can efficiently simulate more general classical harmonic systems with $2^n$ modes.
\end{abstract}

\maketitle

\section{Introduction}
\label{sec:introduction}

The efficient simulation of quantum dynamics is among the most promising areas of quantum computing~\cite{Fey82}. This is due to having practical applications as well as universality providing strong complexity theoretic evidence for exponential quantum advantage \cite{Lloyd1996,Adleman1997}. Several other applications with similarly favorable qualities are known (e.g., factoring \cite{Shor1994}) but the discovery of more compelling use cases remains critical for understanding and motivating the value proposition of quantum computers.

One approach for leveraging the power of Hamiltonian simulation would be to consider the space of problems that reduces to it. Hamiltonian dynamics is an example of a homogeneous, first order partial differential equation that gives rise to unitary evolutions. Thus, it is natural to explore whether other differential equations can be mapped to Hamiltonian simulation. The goal would be a more natural (albeit perhaps more narrow) alternative to solving differential equations compared with methods \cite{Ber14,Clader2013a,Clader2013b} leveraging quantum linear systems algorithms \cite{HHL,Childs2017,Somma2019}. Past work has sought to develop Hamiltonian simulation approaches for a limited set of other differential equations but thus far, has only succeeded in obtaining polynomial speedups \cite{Costa19,Jin2022b,An2022b}.

In this paper, we discuss a rich set of problems in \emph{classical} dynamics that can be mapped to Hamiltonian simulation and solved with exponential quantum advantage. As a prominent example, our approach can simulate the dynamics of exponentially many coupled classical  oscillators in polynomial time. Such systems describe a variety of physical phenomena including: networks of masses and springs \cite{LandauLifshitzb}, circuits with capacitors and inductors \cite{JDJacksonb}, models of neuron activity \cite{Neven1992}, and vibrations in molecules \cite{WilsonVibrationsb}, materials, and mechanical structures. More generally, the harmonic approximation (defined by a quadratic potential) arises as the first order correction to equilibrium in bound systems.  

That the dynamics of coupled classical oscillators can be studied via Hamiltonian simulation is perhaps not too surprising. Solutions to the Schr\"odinger equation contain oscillatory terms and interference, and we are effectively using these properties to simulate the same phenomena in the classical system. This correspondence was also presaged by the finding that Grover search can be implemented (in the absence of errors) using mechanical wave interference~\cite{grover2002classical}. Similarly,
the problem of simulating the (discrete) classical wave equation is a special case of the oscillator dynamics problem considered here,
and a quantum algorithm
to solve the classical wave equation via Hamiltonian simulation has been studied in Ref.~\cite{Costa19}.
However, despite the seemingly natural connection, our mapping involves a number of subtle technical features and modern quantum algorithms are required to efficiently simulate the resultant Hamiltonians.

Besides providing a classical-to-quantum reduction, another key to obtain exponential quantum advantage for this problem is the way we encode physical quantities in quantum states. In our case, this encoding is motivated by energy conservation (being different to the one in Ref.~\cite{Costa19}). Specifically,
we encode quantities related to the displacements and momenta of the classical system in the amplitudes of a quantum state. Thus, extracting the full configuration of the classical system would scale polynomially in the Hilbert space dimension, precluding a large quantum speedup. Nevertheless, interesting global properties, like estimating the kinetic or potential energies at any time, can still be computed efficiently. Likewise, our methods only provide an exponential speedup for the dynamics of sparsely coupled oscillator networks when masses, spring constants, and initial states can be efficiently computed. Fortunately, many interesting systems meet those conditions.

 Finally, we provide strong evidence that such simulations cannot also be performed efficiently on a classical computer. In particular, we are able to show that estimating the kinetic energy of
 simple instances of these systems at any time that is $\poly(n)$
 requires $2^{\Omega(n)}$ queries to the oracles in the worst case and, when the oracles are instantiated by quantum circuits of $\poly(n)$ size,
 the problem is \BQP-complete.
Thus, all problems efficiently solved by a quantum computer can be reduced to an instance of simulating exponentially many coupled classical oscillators, which can also be solved efficiently by our approach.

The rest of the paper is organized as follows. In Sec.~\ref{sec:harmonicapprox} we describe classical systems of oscillators, define the associated simulation problems of interest, and state our main results. In Sec.~\ref{sec:quantumevolution} we show how these problems reduce to instances of Hamiltonian simulation, giving rise to an efficient quantum algorithm. We 
 discuss some applications of this algorithm in Sec.~\ref{sec:applications}, making emphasis on the problem of estimating the time-dependent kinetic (or potential) energies of the oscillators. In Sec.~\ref{sec:gluedtrees} we show the exponential lower bound for classical algorithms for this problem in the oracle setting, and in Sec.~\ref{sec:BQP} we show that
this problem is \BQP-complete when the oracles can be accessed via efficient quantum circuits. Finally, in Sec.~\ref{sec:generalization} we generalize our approach to efficiently simulating classical systems under the harmonic approximation. Detailed proofs of our claims are provided in the Appendix. We also provide a comparison of
our results with those of related work on quantum algorithms for differential equations in Appendix~\ref{app:relatedwork}.

\section{Simulating coupled oscillators: Main results} 
\label{sec:harmonicapprox}

We consider a classical system of coupled harmonic oscillators, i.e., $N=2^n$ point (positive) masses $m_1,\ldots,m_N$ that are coupled with each other by springs. At any time $t\ge 0$, the displacements (with respect to their rest position) and velocities of the masses are given by
$\vec x(t) =(x_1(t),\ldots,x_{N}(t))^T \in \mathbb R^{N}$
and $\dot{\vec x} (t)=( \dot x_1(t), \ldots,  \dot x_{N}(t))^T \in \mathbb R^N$, respectively, where $\dot a(t)=\frac {d}{dt}a(t)$ and $\ddot a(t)=\frac {d^2}{dt^2}a(t)$.
We let $\kappa_{jk}=\kappa_{kj}\ge 0$ be the spring constants that couple the $j^{\rm th}$ and $k^{\rm th}$ oscillator, and $\kappa_{jj}\ge 0$ is the spring constant that connects the $j^{\rm th}$ oscillator to a ``wall''.
We have described it for dimension $D=1$ for simplicity (i.e., we assigned one spatial coordinate to each oscillator), but the same formulation holds for coupled harmonic oscillators in $D$ spatial dimensions, for arbitrary $D$, by using $D$ coordinates to represent the position of a single oscillator (see Fig.~\ref{fig:2Dsprings} for an example).

\begin{figure}\centering
\includegraphics[scale=.4]{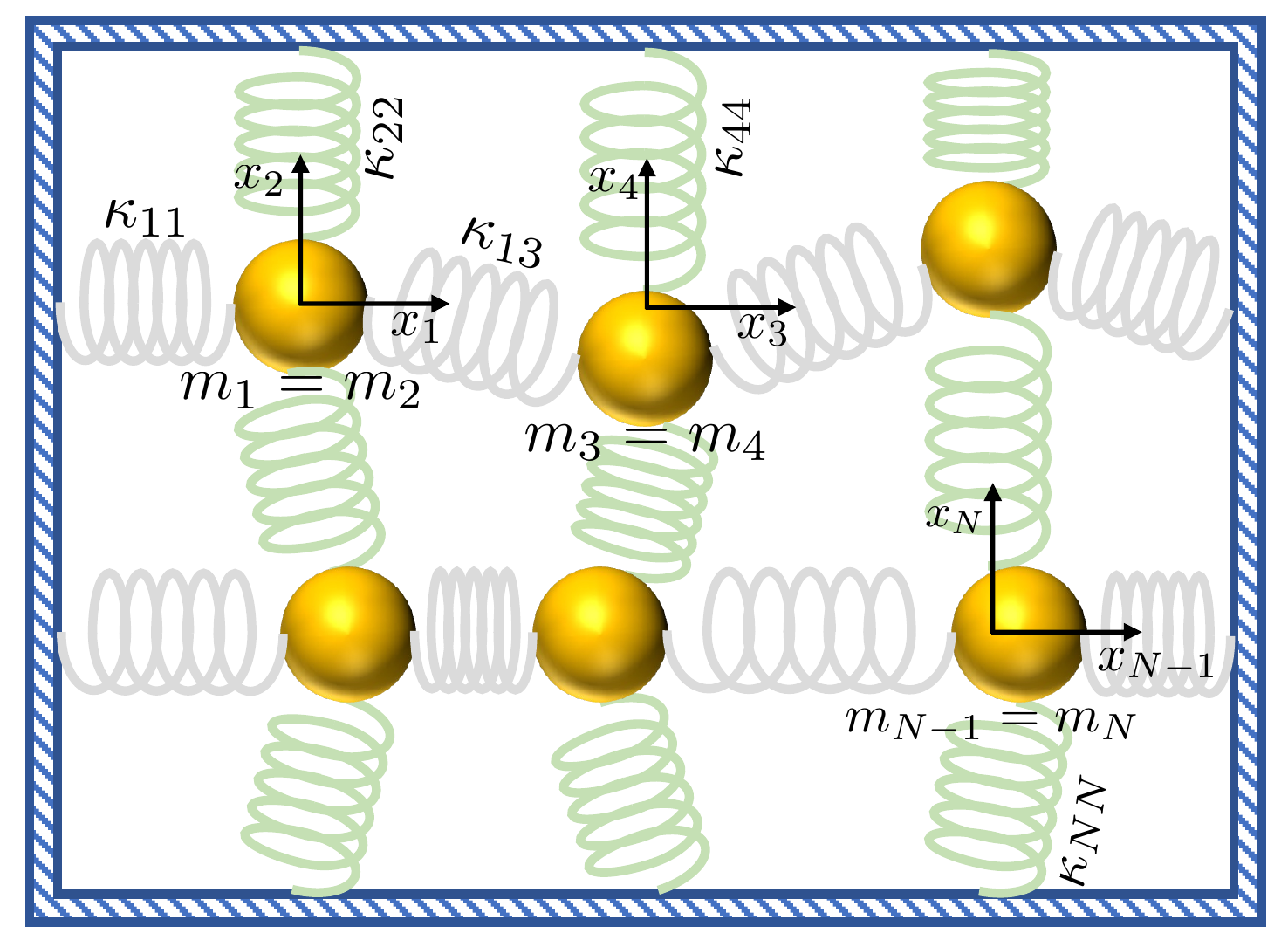}
\caption{
An example system of $N/2$ oscillators in $D=2$
 dimensions. We use $x_1(t)$ for the first coordinate of the first mass and $x_2(t)$ for the second coordinate of the first mass, and so on. Since the first and second mass now correspond to the same original mass, $m_1=m_2$.
\label{fig:2Dsprings}}
\end{figure}

In the harmonic approximation,
the dynamics of the oscillators
can be determined from
the initial values $\vec x(0)$ and $\dot{\vec x}(0)$ and
Newton's equation (for all $j \in [N]:=\{1,\ldots,N\}$): 
\begin{align}
m_j \ddot x_j(t) = \sum_{k\neq j}\kappa_{jk}\bigl(x_k(t)-x_j(t)\bigr) - \kappa_{jj}x_j(t)\;. 
\label{eq:newton}
\end{align}
In matrix form, this is ${\bf M}\ddot{\vec x}(t)=-{\bf F} \vec x(t)$, where ${\bf M}$ is a $N\times N$ diagonal matrix with entries $m_j>0$ and ${\bf F}$
is the $N\times N$
matrix whose diagonal and off-diagonal entries are $f_{jj}=\sum_{k}\kappa_{jk}$ and $f_{jk}=-\kappa_{jk}$, respectively.
(Observe that ${\bf F}$ is the discrete Laplacian of a weighted graph.)
Solutions to Eq.~\eqref{eq:newton} are well understood and can be expressed in terms of normal modes~\cite{LandauLifshitzb}, which is essentially a way of describing the system as $N$ uncoupled harmonic oscillators in a different basis. Classical algorithms that compute $\vec x(t)$ and $\dot{\vec x}(t)$
have complexity $\poly(N)$ or, equivalently, $\exp(n)$. (In this paper we use $\exp(m)$, $\poly(m)$ and $\polylog(m)$ to mean $\cO(k^m)$, $\cO(m^k)$, and $\cO(\log^k m)$ respectively, for some constant $k$. (We use ``log'' to mean logarithm to base 2.)

Our goal is to provide a quantum algorithm for simulating the dynamics of the classical oscillators
efficiently, in time $\poly(n)$. 
Doing so requires a particular notion 
of ``simulation'', since any algorithm that outputs the full vectors $\vec x(t)$ or $\dot{\vec x}(t)$ would necessarily have complexity at least linear in $N$.
Specifically, we consider a problem formulation where the output is a quantum state with some amplitudes proportional to $ \sqrt{m_j} \dot x_j(t)$ and others to $\sqrt{\kappa_{jk}}(x_j(t)-x_k(t))$ or $\sqrt{\kappa_{jj}}x_j(t)$, and the input masses and spring constants are provided through oracles in the usual way; see App.~\ref{app:accessmodel}. (Here and throughout this paper, $\sqrt{x}$ is the principal square root of $x\geq 0$ and $\sqrt{\bf X}$ is the principal square root of a positive semidefinite matrix ${\bf X} \succcurlyeq 0$.)
The formal problem is as follows:
\begin{problem}
\label{problem:oscillators}
Let ${\bf K}$ be the $N \times N$ symmetric matrix of spring constants $\kappa_{jk}\ge 0$ and assume it is $d$-sparse (i.e., there are at most $d$ non-zero entries in each row). Let ${\bf M}$ be the $N \times N$ diagonal matrix of masses $m_j>0$ and define the normalized state
\begin{align}
\label{eq:encodedstate}
\ket{\psi(t)} := \frac 1 {\sqrt {2E}} \begin{pmatrix} \sqrt{\bM} \dot{\vec x}(t) \cr  i\vec \mu(t) \end{pmatrix} \;,
\end{align}
where $E>0$ is a constant, and $\vec \mu(t) \in \mathbb R^{M}$ ($M:=N(N+1)/2$) is a vector
with $N$ entries $\sqrt{\kappa_{jj}}x_j(t)$ and $N(N-1)/2$ entries $\sqrt{\kappa_{jk}}(x_j(t)-x_k(t))$, with $k>j$.
Assume we are given oracle access to ${\bf K}$ and ${\bf M}$, and 
oracle access to a unitary $\cW$ that prepares 
the initial state, i.e., $\cW \ket 0 \mapsto \sket{\psi(0)}$.
Given $t \ge 0$ and $\epsilon > 0$, the goal is to output a state that is $\epsilon$-close to $\ket{\psi(t)}$ in Euclidean norm.
\end{problem}

Our main result is  a quantum algorithm that prepares $\sket{\psi(t)}$ efficiently: 
\begin{theorem}
\label{thm:main}
    Problem~\ref{problem:oscillators} can be solved with a quantum algorithm that makes $Q=\cO(\tau + \log(1/\epsilon))$ queries to the oracles for ${\bf K}$ and ${\bf M}$, uses 
    \begin{equation}
        G= \cO\left(Q \times \log^2\left(\frac{N \tau}{\epsilon} \frac{m_{\max}}{m_{\min}}\right)\right)
    \end{equation}
    two-qubit gates, and uses $\cW$ once, where 
    $\tau := t \sqrt{\aleph d} \geq 1$,
    $\aleph:=\kappa_{\max}/m_{\min}$, 
    $m_{\max}\ge m_j \ge m_{\min}>0$ and $\kappa_{\max} \ge \kappa_{jk}$ for all $j,k \in [N]$ are known quantities. 
\end{theorem}

 In the $\cO$ notation above, the asymptotically large parameters are $N$, $\tau$, $1/\epsilon$, and $m_{\max}/m_{\min}$. This complexity has explicit dependence on $n=\log (N)$ that is polynomial. If $\tau$, $\log(1/\epsilon)$, and $\log(m_{\max}/m_{\min})$ are $\poly(n)$, and all oracles (including $\cW$) can be performed with cost $\poly(n)$, then the entire algorithm has complexity $\poly(n)$.
An interesting feature of this complexity is that it scales as the square root of the sparsity $d$, whereas Hamiltonian simulation complexities are usually linear in $d$. The complexity is also linear in $\tau$, which  bounds the maximum number of oscillations of the normal modes after time $t$.

Our algorithm can be used to determine properties of the system at time $t$. For example, the state $\sket{\psi(t)}$ encodes the velocities and displacements of the oscillators in a way that makes it easy for the estimation of the kinetic or potential energies, as we discuss below.  We discuss in App.~\ref{app:otherencodings} that other encodings are also possible, and the choice of encoding may determine which initial states and properties (observables) are efficient to prepare and measure. Nevertheless, our main goal is to establish results for this particular encoding. 

The constant in Eq.~\eqref{eq:encodedstate}
is $E=K(t)+ U(t)$, where $K(t) = \frac 1 2 \dot{\vec x}(t)^T {\bf M} \dot{\vec x}(t)$ is the kinetic energy, i.e.,
\begin{align}
K(t)
& = \frac 1 2  \sum_{j} m_j \dot x_j(t)^2 \;,
\end{align}
and $U (t) = \frac 1 2 {\vec \mu}(t)^T {\vec \mu}(t)$ is the potential energy, i.e.,
\begin{align}
\label{eq:potential}
U (t) =\frac 1 2 \biggl(\sum_{j} \kappa_{jj}x_j(t)^2 +   \sum_{k>j} \kappa_{jk}\bigl(x_j(t)-x_k(t)\bigr)^2 \biggr).
\end{align}
Hence, $E$ is the total energy of the system, which is time-independent as the system is closed. The support of $\sket{\psi(t)}$
on the subspace spanned by the first $N$ basis vectors
is $K(t)/E$, which can be estimated via a simple measurement on $\ket{\psi(t)}$; see Sec.~\ref{sec:applications}. The support on the other subspace is $U(t)/E$. As an example of the type of problem that can be solved by preparing $\sket{\psi(t)}$, consider the following:
\begin{problem}
\label{problem:kineticestimation}
    Given the same inputs as Problem \ref{problem:oscillators} and an oracle for $V \subseteq [N]$, which is a subset of the oscillators,
    output an estimate $\hat k_V(t) \in \mathbb R$ such that 
    \begin{align}
    \label{eq:kineticestimator}
    |\hat k_V(t)-K_V(t)/E| \le \epsilon \;,
    \end{align}
     where
    $K_V(t):=\frac 1 2 \sum_{j \in V} m_j \dot x_j(t)^2$ is the kinetic energy of $V$ at time $t$.
\end{problem}
\noindent In Sec.~\ref{sec:applications}, we prove that our quantum algorithm  solves this problem with high probability with $\cO(1/\epsilon)$ uses of the quantum algorithm from Thm.~\ref{thm:main}. We also show a related result for the estimation of the potential energy stored on a subset $V \subseteq[N] \times [N]$ of springs at time $t$.

Problems~\ref{problem:oscillators} and~\ref{problem:kineticestimation} are formulated using a specific input and output format, and one might wonder if a classical algorithm could also solve these problems efficiently. The answer is no: we show that any classical algorithm solving Problem~\ref{problem:kineticestimation} must make $\poly(N)$ or $\exp(n)$ queries to the oracles in general.
\begin{theorem}\label{thm:oracle}
    Any classical algorithm that solves Problem~\ref{problem:kineticestimation} with high probability must make $2^{\Omega(n)}$ queries to the oracle for $\bK$. This lower bound holds even if we further require that all $m_j=1$, all $\kappa_{jk}\in \{0,1 \}$, $d\leq 3$, and the initial state is $\ket{\psi(0)}=\ket{0}^{\otimes q}$ for some $q=\poly(n)$, which corresponds to $\vec x(0)=(0,0,\ldots,0)^T$ and $\dot{\vec x}(0)=(1,0,\ldots ,0)^T$.
\end{theorem}

To provide more evidence of the impossibility of efficient classical simulations and prove a stronger hardness result,
we define a non-oracular version of Problem~\ref{problem:kineticestimation} and show that it is \BQP-complete. Thus, no classical algorithm can solve this non-oracular problem efficiently (in time $\poly(n))$ unless it can also solve every problem solved by quantum computers in polynomial time. The same observation essentially applies to Problem~\ref{problem:kineticestimation}
as a result. 
Formally, to be \BQP-complete a problem must be a decision problem (i.e., a problem where there are only two possible answers), so we define a decision version of Problem~\ref{problem:kineticestimation} with $|V|=1$ and a sparse initial state: 

\begin{problem}
    \label{prob:nonoracular}
    The setup is as in Problem \ref{problem:kineticestimation}, but we are given efficient quantum circuits (i.e., circuits with $\poly(n)$ gates) to implement the oracles for $\bf K$ and $\bf M$, an explicit description of $\ket{\psi(0)}$, which is required to have a constant number of nonzero entries, and a label $v \in [N]$ of a single oscillator.
    We additionally require that $\tau$, $1/\epsilon$, and $m_{\max}/m_{\min}$ are bounded by $\poly(n)$. 
    The problem is to decide if $K_v(t)/E$, the kinetic energy of oscillator $v$ as a fraction of total energy, is at least $1/\poly(n)$ or at most $1/\exp(\sqrt n)$, promised that one of these holds.
\end{problem}
Note that this problem has an input of size $\poly(n)$ and our algorithm for Problem~\ref{problem:oscillators} will solve this on a quantum computer in time $\poly(n)$. 
We then show this problem is \BQP-complete in Sec.~\ref{sec:BQP}:

\begin{theorem}
\label{thm:bqpcomplete}
Problem~\ref{prob:nonoracular} is \BQP-complete. The problem remains \BQP-complete even if we further require that all $\kappa_{jk}\leq 4$, $m_j=1$, $d \leq 4$, and the initial state is $\ket{\psi(0)}=\ket{0}^{\otimes q} \otimes \ket -$ for some $q=\poly(n)$, which corresponds to 
$\vec x(0)=(0,0,\ldots,0)^T$ and $\dot{\vec x}(0)=(+1,-1,0,\ldots ,0)^T$.
\end{theorem}

Finally, we show that our results can be extended to address more general classical systems that arise naturally under the harmonic approximation~\cite{LandauLifshitzb}. In these cases, for example,
the matrix ${\bf M}$ resulting from Eq.~\ref{eq:newton} is not necessarily diagonal and the off-diagonal entries of ${\bf F}$ are not necessarily non-positive.
A simple change of variables takes  Eq.~\ref{eq:newton}
to the standard form $\ddot {\vec y}(t)=-{\bf A} \vec y(t)$, where $\vec y(t)=\sqrt{\bM}\vec q(t)$,
and $\vec q(t) \in \mathbb R^N$ are the generalized displacements.
We then consider the following generalization of Problem~\ref{problem:oscillators}:
\begin{problem}
\label{problem:harmonicapprox}
Let ${\bf A} \succcurlyeq 0$ be an $N \times N$ real-symmetric,  PSD, $d$-sparse matrix.
Define the normalized state
\begin{align}
\label{eq:encodedstate2}
\ket{\psi(t)} := \frac 1 {\sqrt {2E}} \begin{pmatrix} \dot{\vec y}(t) \cr  i \vec \mu(t) \end{pmatrix} \;,
\end{align}
where $E>0$ is a constant and $\vec \mu(t):=\sqrt{\bA} {\vec y}(t) \in \mathbb C^{N}$.
Assume we are given oracle access to ${\bf A}$ 
and oracle access to a unitary $\cW$ that prepares 
the initial state $\sket{\psi(0)}$.
Given $t$ and $\epsilon$, the goal is to output a state that is $\epsilon$-close to $\ket{\psi(t)}$ in Euclidean norm.
\end{problem}

The constant $E$ in Eq.~\eqref{eq:encodedstate2}
is also the energy of the system.
The encoding is 
different from the 
one used for Problem~\ref{problem:oscillators}; for example $\sqrt{\bf A}$ is of dimension $N \times N$ and might not be sparse even if ${\bf A}$ is.
However, the support of $\sket{\psi(t)}$ on the subspace spanned by the first $N$ basis states is still $K(t)/E$, where $K(t)$ is the kinetic energy.
We then show the following:
\begin{theorem}
\label{thm:generalized}
    Problem~\ref{problem:harmonicapprox} can be solved with a quantum algorithm that makes
    \begin{equation}
        Q=\cO \left(  \|{\mathbf{A}}\|_{\max} d\log(1/\epsilon) \min \left( \frac{t\sqrt{\|\mathbf{A}^{-1}\|}}{\epsilon} , \frac{t^2}{\epsilon^2} \right)\right)
    \end{equation}
    queries to the oracles, uses
    $\widetilde \cO(Q \times {\rm polylog}(N/\epsilon))$ two-qubit gates, and uses $\cW$ once.
\end{theorem}
The  $\tilde \cO$ notation hides logarithmic factors in several parameters that specify the input.

\section{Reduction to quantum evolution} 
\label{sec:quantumevolution}

We show how to reduce
Problem \ref{problem:oscillators} to time evolution of a quantum system.
A  change of variables where $\vec y(t):=\sqrt{\bM} \vec x(t)$ allows us to write Eq.~\eqref{eq:newton} as 
\begin{equation}\label{eq:matrixnewton}
\ddot{\vec y}(t)=-{\bf A} \vec y(t),    
\end{equation}
where ${\bf A}:=\sqrt{\bM}^{-1} {\bf F} \sqrt{\bM}^{-1} \succcurlyeq 0$ is  PSD and real-symmetric, and ${\bf M}\succ 0$ is the diagonal matrix with entries $m_j$. 
Any solution to Eq.~\eqref{eq:matrixnewton} satisfies
\begin{align}
\label{eq:closedsolution}
 \ddot {\vec y} (t) + i \sqrt{\bA} \dot {\vec y}  (t)  = i \sqrt{\bA}\Bigl(\dot {\vec y}  (t) + i \sqrt{\bA} \vec y  (t) \Bigr)\;.
\end{align}
This is simply Schr\"odinger's equation induced by the Hamiltonian $-\sqrt{\bA}$. Hence its solution is
\begin{equation}
\label{eq:solution}
    \dot {\vec y}  (t) + i \sqrt{\bA} \vec y  (t) = 
    e^{it\sqrt{\bA}} 
    \Bigl(\dot {\vec y}  (0) + i \sqrt{\bA} \vec y  (0)\Bigr).
\end{equation}
Unfortunately, we do not have direct access to $-\sqrt{\bA}$ and casting the problem as quantum evolution with simpler Hamiltonians requires additional steps. It is possible, however, to do Hamiltonian simulation with  $-\sqrt{\bA}$ given oracle access to ${\bf A}$ using quantum phase estimation, as we describe in Sec.~\ref{sec:generalization}, but that is less efficient than what we describe here.
The phase estimation approach also gives rise to a
different encoding, and some properties of the system are not as easy to access. For example, in the encoding used in Problem~\ref{problem:oscillators},
 the support of $\sket{\psi(t)}$ in basis states are either proportional to the kinetic energies or potential energies of the springs, but this might not be the case in other encodings.

Let $\bf B$ be any $N \times M$ matrix
that satisfies $\bf B  \bf{B}^\dagger =  {\bf A}$
and ${\bf H}$ be the block Hamiltonian:
\begin{align}
\label{eq:hamiltonian}
{\bf H}:=- \begin{pmatrix} {\bf 0} &  \bf B \cr  \bf B^\dagger & {\bf 0} \end{pmatrix} \;,
\end{align}
where ${\bf 0}$ are matrices of all zeroes. (The dimension of each ${\bf 0}$ is clear from context.)
Note that   ${\bf H}$
acts on the space $\mathbb C^{N+M}$ and 
functions as the ``square root'' of ${\bf A}$, since
the first block of ${\bf H}^2$ is ${\bf A}$.
Schr\"odinger's equation induced by ${\bf H}$ is
\begin{align}
\label{eq:schrodinger}
 \sket{\dot \psi(t)} = -i {\bf H} \sket{\psi(t)} \;,
\end{align}
where $\ket{\psi(t)} \in \mathbb C^{N+M}$ is the  state of a quantum system at time $t$. 

It can be verified by direct substitution that 
 \begin{align}
\label{eq:quantumstate}
\ket{\psi(t)} \propto \begin{pmatrix} \dot{\vec y}(t) \cr i {\bf B}^\dagger {\vec y}(t) \end{pmatrix}
\, .
\end{align}
is a valid solution to Eq.~\eqref{eq:schrodinger}. Hence we have
\begin{equation}
    \begin{pmatrix} \dot{\vec y}(t) \cr i {\bf B}^\dagger {\vec y}(t) \end{pmatrix} = 
    e^{-it{\bf H}}
    \begin{pmatrix} \dot{\vec y}(0) \cr i {\bf B}^\dagger {\vec y}(0) \end{pmatrix},
\end{equation}
which lets us compute $\dot{\vec{y}}(t)$ and ${\bf B}^\dagger {\vec y}(t)$ using Hamiltonian simulation. This generalizes Eq.~\eqref{eq:solution} since ${\bf B}$ is not necessarily $-\sqrt{\bA}$. (This formulation also stores the positions and velocities in different components of the vector, but this is a minor difference.)

To match the state representation of Problem \ref{problem:oscillators}, we need to choose $\bf B$ such that $\vec \mu(t)={\bf B}^\dagger {\vec y}(t)={\bf B}^\dagger \sqrt{\bM} {\vec x}(t)$, where $\vec \mu(t)  \in \mathbb C^M$ (for $M=N(N+1)/2$). 
We can express $\vec \mu(t)$ in the basis $\{\ket{j,k}:j\le k \in [N]\}$, which is of size $M$. 
Since ${\bf F}$ is a graph Laplacian, we can obtain $\bf B$ from the incidence matrix of ${\bf F}$ given by
\begin{align}
\label{eq:Bdaggeractionmain}
  \sqrt{\bM}  {\bf B} \ket {j,k}= 
    \begin{cases}
        \sqrt{\kappa_{jj}}\ket j, & \text{if}\ j=k\\
        \sqrt{{\kappa_{jk}}}( \ket j -\ket k), & \text{if}\ j<k
    \end{cases}.    
\end{align}
This choice satisfies $\sqrt{\bM}{\bf B} (\sqrt{\bM}{\bf B})^\dagger= {\bf F}$, and hence $\bf B \bf{B}^\dagger={\bf A}$, because we can write 
${\bf F} = \sqrt{\bM}(\sum_{j\leq k} {\bf B} {\ket {j,k}} {\bra {j,k}}{\bf B}^\dagger ){\sqrt{\bM}}$
and each term in this sum describes the interaction between the $j^{\rm th}$ and $k^{\rm th}$ oscillators:
\begin{align}
\sqrt{\bM} {\bf{ B}} {\ket {j,j}} {\bra {j,j}}{\bf {B}}^\dagger \sqrt{\bM} =\kappa_{jj} \ketbra j\;,
\end{align}
and for $j<k$,
\begin{align}
\nonumber
   \sqrt{\bM} {\bf B} {\ket {j,k}} {\bra {j,k}}{\bf B}^\dagger \sqrt{\bM}& = \kappa_{jk}(\ketbra j + \ketbra k \\
   & \quad - \ket j \! \bra k - \ket k \! \bra j).
\end{align}
These reproduce ${\bf F}$ once we perform the sum.
To obtain the action of ${\bf B}$ we apply $\sqrt{\bM}^{-1}$
to Eq.~\eqref{eq:Bdaggeractionmain}.

Note that a similar approach based on the incidence matrix
was taken in Ref.~\cite{Costa19} to provide a quantum algorithm
that simulates the wave equation. 
The setting considered in that paper corresponds to the special case of our problem
where the graph is spatially local, and the masses and spring constants are uniform.

\section{Quantum algorithm}
\label{sec:algorithm}

Our quantum algorithm solves Problem~\ref{problem:oscillators} 
by preparing $\ket{\psi(0)}$ and simulating ${\bf H}$ for time $t$.

{\bf \em Access model:}
The Hamiltonian $\bf H$  is built from the matrix of spring constants ${\bf K}$
(or ${\bf F}$) and ${\bf M}$. As we wish to avoid complexities
that are polynomial in $N$, we require succinct representations
of these matrices. 
For maximum generality, we assume that access to ${\bf M}$ and the $d$-sparse ${\bf K}$ is provided by a black box, similar to that used in prior quantum algorithms~\cite{BAC07,BC12}.
The black box is a unitary $\cS$
that computes the masses $m_j$ on input $j$ and the nonzero entries of ${\bf K}$, i.e.\ $\kappa_{jk}$, on input $(j,k)$,  as well as their locations.
This black box readily gives query access to ${\bf H}$ in Eq.~\eqref{eq:hamiltonian}, which is also $d$-sparse. Its entries are $\pm\sqrt{\kappa_{jk}/m_j}$ and these factors can be applied using the values of $m_j$ and $\kappa_{jk}$; see App.~\ref{app:accessmodel}.

{\bf \em Simulation of ${\bf H}$:}
Time evolution with the sparse ${\bf H}$ can be simulated using one of the methods in Refs.~\cite{BCC+15,BCK15,LC17,LC19}, which apply an approximation of the exponential $e^{-it{\bf H}}$. Such methods are tailored to our access model and achieve almost optimal scaling in the parameters $\epsilon$, $d$, $t$, and $\|{\bf H}\|_{\max}$,  the largest  entry of ${\bf H}$ in absolute value.
Here $\|{\bf H}\|_{\max}$ is at most $\sqrt\aleph$, which is defined  in Thm.~\ref{thm:main}.

{\bf \em Complexity:}
The complexity of our approach is determined by the simulation of ${\bf H}$ for time $t$ from the initial state. 
We consider the query complexity $Q$ of the number of calls to $\cS$,
and we allow inverses and controlled forms of these oracles.
This query complexity can also be taken to include the unitary for preparing the initial state, but only one call to that preparation is needed.
We also consider the gate complexity $G$, which is the number of additional elementary gates, which could, for example, be single-qubit gates and CNOTs.
We describe these as ``two-qubit gates'' in Thm.~\ref{thm:main} for simplicity.

Using Ref.~\cite{LC17},
the Hamiltonian evolution may be simulated with
\begin{align}\label{eq:Qformula}
\cO(t\lam  +\log(1/\epsilon))
\end{align}
calls to a block encoding of ${\bf H}$, where the block encoding gives a factor of $1/\lam$.
In App.~\ref{app:accessmodel} we show how to block encode ${\bf H}$ with $\lam = \sqrt{2\aleph d}$
using $\mathcal{O}(1)$ calls to $\cS$, and so Eq.~\eqref{eq:Qformula} gives the total query complexity $Q$ given in Thm.~\ref{thm:main} since $\tau:=t\sqrt{\aleph d}$.
The reason for the square root dependence on the sparsity is then because the algorithm involves simulation of ${\bf H}$, which is effectively the square root of the $d$-sparse operator ${\bf A}$.
Note that $\aleph d$ is an upper bound on $\|{\bf A}\|$, implying that $\sqrt{\aleph d}$ is an upper bound on the largest frequency of the normal modes.
See App.~\ref{app:accessmodel} for a more technical explanation.

To achieve final error $\epsilon$, each block encoding should be given within error
\begin{align}
\label{eq:epsilon'}
\epsilon' = \cO(\epsilon/[\tau  +\log(1/\epsilon)]) \, ,
\end{align}
which will determine the gate complexity.
In Lemmas~\ref{lem:blockencodingB} and \ref{lem:blockencodingH} of App.~\ref{app:accessmodel} we also show that it suffices
to represent the relevant inputs, i.e., masses and spring constants,
with a number of bits that is $r_m=\cO(\log(m_{\max}/(m_{\min}\epsilon'))$ and $r_\kappa=\cO(\log(1/\epsilon'))$, respectively, to achieve error $\epsilon'$ in the block encoding.
The gate complexity for the arithmetic is essentially $\cO(r_m r_\kappa)$,
and there is also another contribution $\cO(n)$ to the gate complexity for other operations in the block encoding (e.g, for inequality tests and other state preparations).
Hence, we can bound the overall gate complexity as
$G=\cO(Q \log^2(N m_{\max}/(m_{\min}\epsilon'))$. Substituting the value of $\epsilon'$ from Eq.~\eqref{eq:epsilon'}, this gives
\begin{equation}
    G= \cO\left(Q \log^2\left(\frac{N \tau}{\epsilon} \frac{m_{\max}}{m_{\min}}\right)\right),
\end{equation}
which is the expression provided in Thm.~\ref{thm:main}.
To simplify this expression we used the fact that $\log((\tau+\log(1/\epsilon))/\epsilon) \leq \log((\tau+(1/\epsilon))/\epsilon)$.
Because we consider complexity for large $\tau$ and $1/\epsilon$, we can upper bound that by $\log(\tau/\epsilon^2)=\cO(\log(\tau/\epsilon))$.

In Thm.~\ref{thm:main} we have not quantified the state preparation complexity (i.e., the complexity of $\cW$), since Problem~\ref{problem:oscillators} assumes $\cW$ is given as an input.
However, in many situations this preparation can often be performed efficiently. One example is when we have oracles to separately prepare states with amplitudes proportional to $\vec{x}(0)$ and $\dot{\vec{x}}(0)$.
In App.~\ref{app:initialstate}, we show the query complexity to prepare $\sket{\psi(0)}$ is $Q_{\rm ini}=\cO(\sqrt{E_{\max}d/E) })$,
where
\begin{equation}
    E_{\max}=\frac {m_{\max}} 2 \|\dot{\vec{x}}(0)\|^2+\frac {\kappa_{\max}} 2 \| \vec{x}(0)\|^2
\end{equation}
is the energy of $N$ uncoupled oscillators of mass $m_{\max}$ and spring constants $\kappa_{\max}$ with similar initial conditions.
We also establish the gate complexity $G_{\rm ini}=\cO(Q_{\rm ini} \log^2 (N E_{\rm max}/(E\epsilon)))$.

\section{Example Application}
\label{sec:applications}

Our quantum algorithm can simulate
classical oscillators in time polynomial in $n$ under some assumptions. 
As we are concerned with dynamics and seek to preserve the speedup, we focus on computing certain time-dependent properties of the system. In particular, we explain how to use the quantum algorithm to solve Problem~\ref{problem:kineticestimation}, i.e., to estimate the kinetic energy of all or some oscillators $V \subseteq [N]$ at any time.

\begin{theorem}
\label{thm:kineticestimation}
Let $\cV$ be an oracle for $V$, i.e., a unitary that performs the map $\cV \ket j =-\ket j$ if $j \in V$ and $\cV \ket j=\ket j$ otherwise, and $\delta >0$.
Then Problem~\ref{problem:kineticestimation} can be solved with probability at least $1-\delta$
by a quantum algorithm that makes $\cO(\log (1/\delta)/\epsilon)$ uses of $\cV$ and the quantum circuit that prepares $\ket {\psi(t)}$  (along with its inverse and controlled version).
\end{theorem}

When $V=[N]$, the potential energy is $U(t)=E-K_V(t)$,
and a solution to Problem~\ref{problem:kineticestimation}
provides an estimate of $U(t)/E$ as well.
 As usual, we include controlled applications of the oracle $\cV$, which means that the global phase is distinguishable (this way we have that $V=[N]$ is distinct from $V=\{ \}$).
 Although here we focus on estimating kinetic energies, a similar idea can be used to estimate the potential energies stored on a subset of springs at any time, as we discuss below.

To prove Thm.~\ref{thm:kineticestimation}, we
note that ${K_V(t)}/E= \sbra{\psi(t)} P_V \sket{\psi(t)}$,
where $P_V = (\one -\cV)/2$ is the projector onto $V$.
This expectation, which is the support of $\sket{\psi(t)}$ on the subspace spanned by the basis states whose labels correspond to $V$ (i.e., a subspace of that spanned by the first $N$ basis states),
can be obtained by making measurements
on many copies of $\sket{\psi(t)}$, but that approach is not optimal (i.e., the scaling is quadratic in $1/\epsilon$).
The problem reduces to estimating $\bra{\psi(t)} \cV \ket{\psi(t)}$ with additive error at most $2\epsilon$ and error probability $\delta$.
A method based on high-confidence amplitude estimation~\cite{KOS07} then
provides the result in Thm.~\ref{thm:kineticestimation}. 

Hence, when $\sket{\psi(t)}$ can be efficiently prepared and $\cV$ can be efficiently implemented, we can obtain an estimate of $K_V(t)/E$ with error $\epsilon$ and high probability in time polynomial in $n$. If in addition $E$ is known or is efficiently computed, this translates to an efficient estimation of $K_V(t)$ with error $\epsilon E$.

A related result shows that our quantum algorithm can be applied to estimate the potential energy of a subset of oscillators:
\begin{problem}
\label{problem:potentialestimation}
   Given the same inputs as Problem~\ref{problem:oscillators} and an oracle for $V \subset [N] \times [N]$, which is now a subset of springs (edges) labeled as $(j,k)$ with $k \ge j$, output an estimate $\hat u_V(t) \in \mathbb R$ such that
   \begin{align}
       |\hat u_V(t) - U_V(t)/E| \le \epsilon \;,
   \end{align}
   where 
   \begin{align}
   \nonumber
       U_V(t)&:= \frac 1 2  \sum_{j:(j,j) \in V} \kappa_{jj}x_j(t)^2  +\\
      &+\frac 1 2 \sum_{k>j:(j,k) \in V} \kappa_{jk}(x_j(t)-x_k(t))^2
   \end{align} 
   is the potential energy of the springs in $V$ at time $t$.
\end{problem}
By noting that $U_V(t)/E$ is the support of $\ket{\psi(t)}$ on the subspace spanned by the  basis states $\ket{j,k}$ such that $(j,k) \in V$, the problem reduces to estimating the expectation of another unitary $\cV$ on $\ket{\psi(t)}$. Like in the prior example for the kinetic energy, this can be done efficiently using our quantum algorithm:
\begin{theorem}
    Let $\cV$ be an oracle for $V$, i.e., a unitary that performs the map $\cV \ket{j,k}=-\ket{j,k}$ if $(j,k) \in V$ and $\cV \ket{j,k}=\ket{j,k}$ otherwise, and $\delta >0$. Then, Problem~\ref{problem:potentialestimation} can be solved with probability at least $1-\delta$ by a quantum algorithm that makes $\cO(\log(1/\delta)/\epsilon)$ uses of $\cW$ and the quantum circuit that prepares $\ket{\psi(t)}$ (along with its inverse and controlled version).
\end{theorem}

\section{Impossibility of efficient classical simulations}
\label{sec:hardness}

An important question is whether our quantum algorithm
results in an exponential quantum speedup, or whether it can be efficiently simulated by classical algorithms. We address this question in two ways: (i) by showing that 
our approach can solve an oracular problem -- the ``glued trees'' problem of Ref.~\cite{childs2003exponential} -- in $\poly(n)$ time, while Thm.~\ref{thm:oracle} is implied by Ref.~\cite{childs2003exponential}, and (ii) by showing that our approach can simulate any quantum circuit of size $L$ acting on $q$ qubits in $\poly(L,q)$ time.
More precisely we show that Problem~\ref{prob:nonoracular} is \BQP-complete, thereby proving Thm.~\ref{thm:bqpcomplete}.
Although our results use physical systems
that may seem artificial (e.g., resulting interactions between oscillators are 
not spatially local),
they are strong evidence that no polynomial-time (in $n$) classical algorithm for simulating systems of coupled classical oscillators exists.
Our results also add a new problem to the list of ``natural'' \BQP-complete problems~\cite{wocjan2006several}.

\subsection{Oracle lower bound}
\label{sec:gluedtrees}

At a high level, in the glued trees problem of Ref.~\cite{childs2003exponential}, we are given oracle access to the adjacency matrix of a sparse graph with an ENTRANCE vertex, and the goal is to find the EXIT vertex. We map this problem to a system of masses and springs by putting a unit mass at each vertex and a spring of unit spring constant at every edge. Then we show that if the system starts at rest except with the ENTRANCE vertex having unit velocity, after polynomial time, the EXIT vertex will have inverse polynomial velocity, and thus it is a special case of Problem~\ref{problem:kineticestimation}.
Formally, we study the following problem:

\begin{problem}
\label{problem:gluedtrees}
Consider a network of $N=2^{n+1}-2$ masses coupled by springs, where ${\bf M}=\one$ and ${\bf K}$ coincides with the adjacency matrix of a graph
consisting of two balanced binary trees of depth $n$ ``glued''
randomly as in Fig.~\ref{fig:gluedtrees} (i.e., the spring constants are $\kappa_{jk}=1$ if $(j,k)$ is an edge of the graph or $\kappa_{jk}=0$ otherwise).  Each mass (vertex) of the network is labeled randomly with a bit string of size $2n$. 
The network contains two special and verifiable masses, ${\rm ENTRANCE}$ and ${\rm EXIT}$, which correspond to the roots of both trees. 
Given oracle access to ${\bf K}$ and the label of the ${\rm ENTRANCE}$ mass, the problem is to find the label of the ${\rm EXIT}$ mass.
\end{problem}

\begin{figure}\centering
\includegraphics[scale=.3]{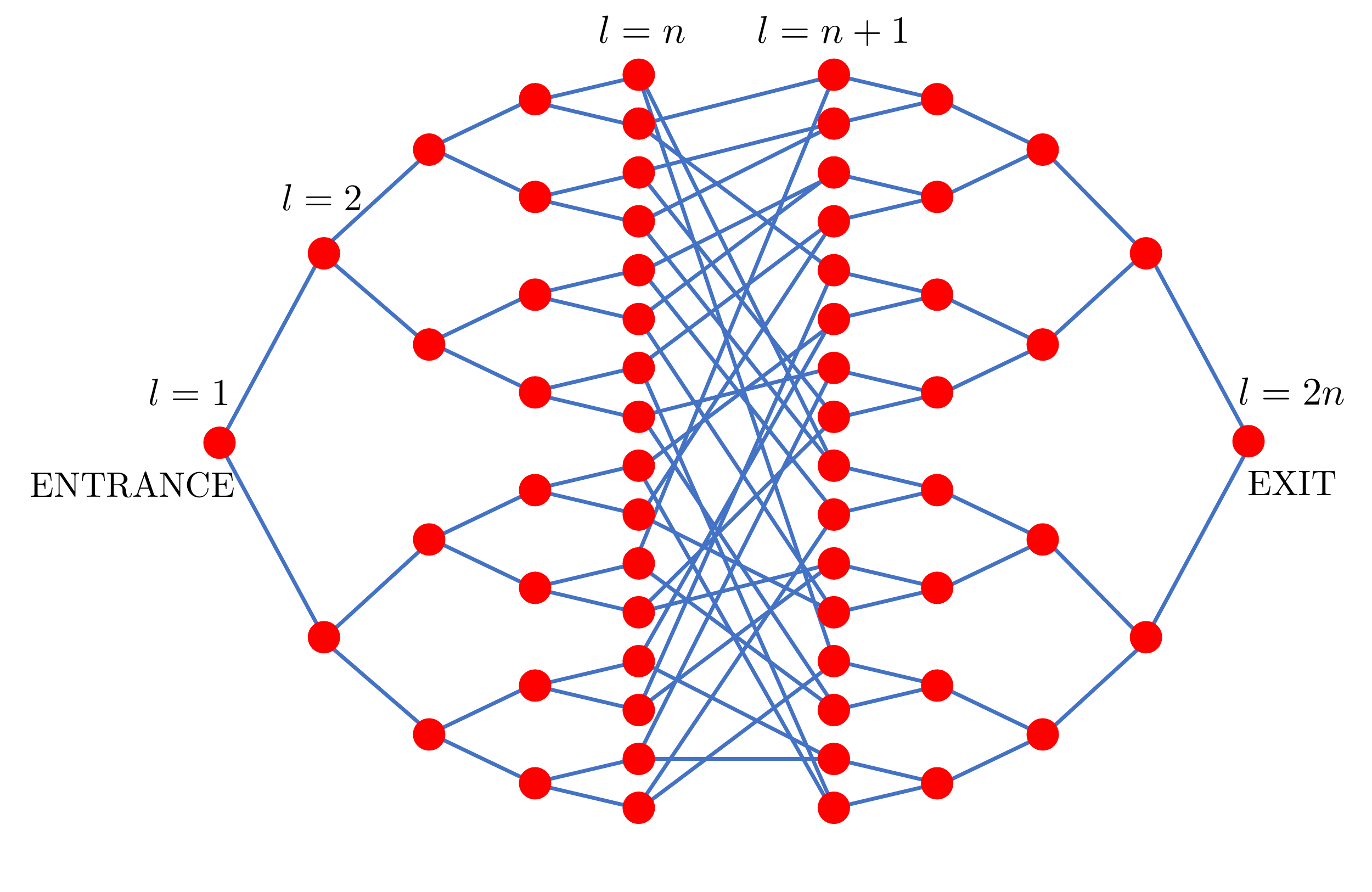}
\label{fig:gluedtrees}
\caption{A network of $N=2^{n+1}-2$ coupled oscillators obtained by randomly gluing two binary trees. Each mass is $m_j=1$ for all $j \in [N]$ and is labeled by a random bit string of size $2n$. Edges denote springs of constant $\kappa_{jk}=1$. The goal is to find the label of the ${\rm EXIT}$ mass given the label of the ${\rm ENTRANCE}$ mass and given oracle access to the network.}
\end{figure}

This problem is equivalent to that studied in Ref.~\cite{childs2003exponential} where the authors
presented an efficient quantum walk based algorithm 
while showing that no classical algorithm can solve the problem efficiently in this oracular setting. 
In that work the root vertices do not have the extra 
edge that connects them to a ``wall'', so they can be easily verified as they are the only vertices of degree two. 
We add these extra edges to simplify the analysis below;
this small change allow us to use some results on the spectral properties of the adjacency matrix already studied in Ref.~\cite{childs2003exponential}. 
With this change, the oracle to access ${\bf K}$ allows us to verify whether a certain vertex is a root or not, since they are the only vertices (or masses) with nonzero diagonal spring constants; i.e., $\kappa_{jj}=1$ only if $j$ corresponds to ${\rm ENTRANCE}$ or ${\rm EXIT}$ and $\kappa_{jj}=0$ otherwise. Since our problem is a minor modification of the glued trees problem, and any oracle query to our problem can be answered using a single oracle query to the glued trees problem, the original classical lower bound of Ref.~\cite{childs2003exponential} implies that our problem requires $2^{\Omega(n)}$ queries to solve classically.

We  then show that our quantum algorithm for simulating coupled classical oscillators can also be used to solve Problem~\ref{problem:gluedtrees} efficiently. This is a different approach from that of Ref.~\cite{childs2003exponential} as we are not simulating the quantum walk but rather the classical dynamics obtained from Newton's equation with a quantum algorithm.
This shows that our quantum algorithm cannot be simulated classically efficiently in the oracular setting in general. Our main claim is the following lemma, which we prove in App.~\ref{app:oracleproof}.

\begin{lemma}
\label{lem:kineticgluedtrees}
Let $\vec x(0)=(0,0,\ldots,0)^T$ and $\dot {\vec x}(0)=(1,0,\ldots,0)^T$, where the first and last components refer to the initial conditions of the ${\rm ENTRANCE}$ and ${\rm EXIT}$ masses in the network of Fig.~\ref{fig:gluedtrees}, respectively. Then, there exists a time $t=\cO({\rm poly}(n))$  at which the kinetic energy of the ${\rm EXIT}$ mass is $K_{\rm EXIT}(t):=\frac 1 2 (\dot x_{\rm EXIT}(t))^2=\Omega(1/{\poly}(n))$.
\end{lemma}

Given this lemma, it is easy to see that Problem~\ref{problem:gluedtrees} is a special case of the problem of estimating the kinetic energy, Problem~\ref{problem:kineticestimation}. The quantum algorithm can prepare $\sket{\psi(t)}$
efficiently under the conditions of Lemma~\ref{lem:kineticgluedtrees}.
Because Lemma~\ref{lem:kineticgluedtrees} proves that the kinetic energy of the ${\rm EXIT}$ mass will be inverse polynomially large after time $t=\cO(\poly(n))$, a projective measurement on $\sket{\psi(t)}$ will return 
the label of the ${\rm EXIT}$ with probability $K_{\rm EXIT}(t)/E=\Omega(1/\poly(n))$, where the energy is $E=\frac 1 2 (\dot x_{\rm ENTRANCE}(0))^2=\frac 1 2$ in this case.
Since the ${\rm EXIT}$ mass can be verified with the oracle by assumption,
the probability of finding the label of ${\rm EXIT}$
can be made a constant close to 1 after $\cO(\poly(n))$ repetitions of the previous procedure, giving an overall query and gate complexity $\cO(\poly(n)))$.
Since Problem~\ref{problem:gluedtrees} requires $2^{\Omega(n)}$ queries to solve classically, this yields the exponential lower bound in Thm.~\ref{thm:oracle}.

\subsection{{\BQP}-completeness}
\label{sec:BQP}
To prove this, we start from the standard \BQP-complete problem of simulating universal quantum circuits, i.e., the problem of determining some property of the output state $\ket{\phi}=U_L\ldots U_1 \ket 0 ^{\otimes q}$, where $L={\rm poly}(q)$. The unitary gates $U_l$ belong to a universal set such as $\{H,X,\mathrm{Toff}\}$, where $H$ and $X$ are single-qubit Hadamard and Pauli $X$ (bit flip) gates, and Toff is the three-qubit Toffoli gate. (Without loss of generality, all Hadamard gates can act on the last qubit, a property that is not necessary but we use it to simplify the presentation of the proof.) The specific problem is that of deciding if $\ket \phi$ is essentially the all zero state or has almost no overlap on it; see Problem~\ref{prob:bqpcomplete}.
In App.~\ref{app:BQPcompletenessproof} we show that this problem reduces to one in which the kinetic energy of a specific oscillator is either exponentially close to zero or at least $1/\poly(q)$ large. This is essentially Problem~\ref{prob:nonoracular} for $n=\cO(q)$.

We use the standard Feynman-Kitaev construction~\cite{Feynman1986a,Kitaev2002b} to encode the $U_l$'s into a Hamiltonian with an extra ``clock'' register. One property of this construction is that evolution under the Hamiltonian is known to apply the sequence of gates $U_L \ldots U_1$ on a given initial state (with $1/\poly(q)$ amplitude). We want this Hamiltonian to be the matrix $\bf A$ of a corresponding system of coupled oscillators. This way we could use the connection between classical and quantum systems discussed in Sec.~\ref{sec:quantumevolution}, which suggests that the evolved state of the oscillators encodes the amplitudes of $\ket \phi$. 
However, this has two problems:
\begin{enumerate}
    \item The off-diagonal entries of the Hamiltonian would not be real negative numbers, i.e., they cannot be related to spring constants;
    \item The evolution of the oscillators is induced by the $\sqrt{\bA}$ rather than ${\bf A}$, so the evolution property of the  Hamiltonian does not apply.
\end{enumerate}
The first problem is due to the fact that 
Hadamards have negative matrix entries. To fix this,
we observe that the same effect can be obtained using an operator with non-negative entries and an ancilla in the $\ket{-}:=\frac  1 {\sqrt 2} (\ket 0- \ket 1)$ state~\cite{jordan2010quantum,CGW15}.
This operator is no longer unitary, but this poses no additional problems as it is only being encoded in a Hamiltonian that, at this stage, corresponds to a system of oscillators coupled by springs.
The second problem is addressed by showing that the 
spectral properties of $\sqrt{\bA}$ 
are such that there is an appropriate
evolution time $t=\poly(q)$ after which the evolved state of the oscillators indeed encodes the output of the quantum circuit with $1/\poly(q)$ amplitude.

Formally, we consider problem instances, i.e., systems of coupled oscillators, where 
 $N=(L+1)2^{q+1}$, ${\bf M}=\one$, and ${\bf A}={\bf F}$
is
\begin{align}
\label{eq:BQPsystem}
{\bf A}=4 \one_N -\sum_{\sl=1}^{L} (\ket \sl \!\bra {\sl+1}+\ket{\sl+1}\!\bra \sl) \otimes W_{\sl} \;.
\end{align}
The matrix $\one_N$ is the $N \times N$ identity matrix.
The $2^{q+1}\times 2^{q+1}$ matrices 
 $W_{\sl}$ are real-symmetric and act on one more qubit than do the $U_{\sl}$ matrices. We define them as $W_{\sl}=U_{\sl} \otimes \one_2$ if $U_{\sl}$ is the $X$ or Toff gate acting on the space of $q$ qubits.
 When $U_{\sl}$ is a Hadamard gate, which acts on the last qubit (the $q^{\rm th}$ qubit), we replace it by the following $4 \times 4$ matrix acting on qubits $q$ and $q+1$:
 \begin{equation}
   \frac 1 {\sqrt 2}  \begin{pmatrix}   1 & 0 &  1 & 0 \cr 0 & 1 & 0 & 1 \cr   1 & 0 & 0 & 1 \cr 0 & 1 & 1 & 0  \end{pmatrix}.
 \end{equation}
This matrix acts like Hadamard on qubit $q$ when qubit $q+1$ is set to $\ket{-}$.
Thus the entries of $W_l$ are non-negative in this case as well, 
the off-diagonal entries of ${\bf A}$ are $\{0,-1/\sqrt 2,-1\}$,
and all diagonal entries are $4$. Hence, the  off-diagonal entries of the matrix ${\bf K}$ are $\{0,1/\sqrt 2,1\}$ and one can show that $\kappa_{jj}\ge 0$ for all $j\in[N]$. These spring constants can be efficiently accessed using Eq.~\eqref{eq:BQPsystem}. The gates $U_l$ are 2-sparse, implying that the $W_l$'s are also $2$-sparse, and ${\bf A}$ and ${\bf K}$
are $4$-sparse. These properties are outlined in Thm.~\ref{thm:bqpcomplete}.

For our proof, we consider the initial condition for the $N$ oscillators
where $\vec x(0)=(0,0,\ldots,0)^T$, 
$\dot{\vec x}(0)=(+1,-1,0,\ldots,0)^T$, so that the energy  is $E=1$.
Labeling each oscillator $j\in[N]$ by $(\sl,r)$, where $\sl\in [L+1]$ and $r \in[2^{q+1}]$,
the oscillator under consideration is the one with $l=L+1$ and $r=1$. In App.~\ref{app:BQPcompletenessproof} we show that
computing the kinetic energy of this oscillator
after time $t=\poly(q)=\polylog(N)$ solves the above  \BQP-complete problem. In addition, this classical system satisfies 
all of our conditions for the quantum algorithm to be efficient (i.e., the problem is in $\BQP$).

\section{Generalized coordinates}
\label{sec:generalization}

In general, when modeling a classical system using the
harmonic approximation, Newton's equation can be simply written
as $\ddot{\vec y}(t) =-{\bf A} \vec y(t)$, where $\vec y(t) \in \mathbb R^N$ encodes the generalized coordinates and ${\bf A}\succcurlyeq 0$ is of dimension $N\times N$.
To prove Thm.~\ref{thm:generalized} we use a standard approach based on quantum phase estimation~\cite{Kit95,CEMM98,Sanders2020}
to first estimate the eigenvalues of ${\bf H^{(2)}}:=-X\otimes {\bf A}$, which are $\gamma_{\eta,j}:=(-1)^{\eta} \lambda_j$, within certain precision that gives the eigenvalues of  ${\bf H}$  within precision $\cO(\epsilon/t)$. Here, $\lambda_j \ge 0$ are the eigenvalues of ${\bf A}$, and $\eta \in \{0,1\}$ determine the eigenvalues of $-X$ as $(-1)^{\eta}$, with eigenstates $\ket{0_X}=\ket -$ and $\ket{1_X}=\ket+$.
We specifically use a standard  QFT-based approach to phase estimation that eschews the need for measurement and in turn maps (within small error) each eigenvector of ${\bf H^{(2)}}$ via
\begin{align}
\nonumber
    \ket{\eta_X,\lambda_j} &\mapsto \ket{\eta_X,\lambda_j}  \Bigl( \sqrt{1-\delta_{\eta,j}}\sum_{x: |x-\gamma_{\eta,j}| \le \epsilon_{\rm PE}}b_x\ket{x} +  \\
   &   
     + \sqrt{\delta_{\eta,j}} \ket{\phi_{\eta,j}} \Bigr)\; ,
\end{align}
where $\ket{\phi_{\eta,j}}$ is some unspecified error state of the ancillas of unit norm, $\sum_x|b_x|^2=1$, $\delta_{\eta,j}\le \delta_{\rm PE}$, and $\epsilon_{\rm PE}$ and $\delta_{\rm PE}$ need to be set according to the error requirements. The states $\ket x$ encode the eigenvalue estimates of ${\bf H^{(2)}}$. From each estimate $x$, we obtain the eigenvalue estimate of ${\bf H}$ by taking its sign and computing a square root, i.e., ${\rm sign}(x)\sqrt{|x|}$. Then we apply a phase factor yielding
\begin{align}
    \ket{\eta_X,\lambda_j}& \Bigl(\sqrt{1-\delta_{\eta,j}}\sum_{x: |x-\gamma_{\eta,j}| \le \epsilon_{\rm PE}}e^{-it\,{\rm sign}(x)\sqrt{|x|}} b_x\ket{x} \nn
    & + \sqrt{\delta_{\eta,j}} \ket{\phi_{\eta,j}'}\Bigr) \; .
\end{align}
Last, we invert quantum phase estimation and other operations to (approximately) uncompute the estimated eigenvalues. The result is, for a value of the error tolerances where $\delta_{\rm PE}=\cO(\epsilon^2)$ and
\begin{align}
    \epsilon_{\rm PE}=\cO\left(\max \left( \frac{\epsilon }{t\sqrt{\|{\bf A}^{-1}\|}},\frac{\epsilon^2}{t^2}\right)\right),
\end{align}
a unitary approximating  ${e^{-it {\bf H}}}$ within error $\epsilon$. The result is implied by the complexity of quantum phase estimation that requires $\cO(\|{\bf A}\|_{\max}d \log(1/\delta_{\rm PE})/\epsilon_{\rm PE})$ queries to the oracle. See App.~\ref{app:QPE} for details.

\section{Conclusion}

We introduced an approach to simulating systems of coupled classical oscillators using a quantum computer. We showed how to map Newton's equations for these dynamics (coupled second order ordinary differential equations) to the Schr\"{o}dinger equation (a first order partial differential equation), with polylogarithmic overhead in the number of oscillators under certain assumptions. Using this mapping, we demonstrated that any classical algorithm that simulates these classical dynamics requires at least $2^{\Omega(n)}$ queries to the oracle, and that they can provide solutions to \BQP-complete problems, and would thus require exponential time to solve classically under reasonable complexity theoretic assumptions. We also generalized our approach to the simulation of other classical harmonic systems.

While providing a large quantum advantage in certain contexts, these techniques also have significant limitations. For example, the approach is only efficient for computing particularly large or global properties and when masses and spring constants can be computed in time polylogarithmic in system size. Another feature of our algorithm is that its complexity
is (almost) linear in the evolution time $t$, a scaling that might not be avoided in general~\cite{BAC07,AA17,HHK+18},  being efficient only if $t$ is also polylogarithmic in system size.
This would discourage applications where, for example, $N=\poly(t)$ (and when fast forwarding is not possible). This feature is expected to arise when simulating physical systems with geometrically-local interactions, and for initial states that are locally supported. 
In these examples, the relevant system size is determined by the ``light cone'', whose size in $D$ spatial dimensions would scale as $N \sim t^D$. Nevertheless, even for these examples our quantum algorithm would still result in significant quantum speedups (e.g., super-quadratic), suggesting a new application area for quantum computers \cite{BabbushFocus}.

As many systems from molecular vibrations, to structural mechanics, to electrical grids, to neuronal activation can be modelled within the harmonic approximation, and since interesting dynamical features can appear at relatively short times (e.g., $t$ independent of $N$),
we expect
that there exist specific applications that meet all the requirements for our quantum algorithm to be efficient. Such applications will then benefit from this speedup with appreciable real-world impact; 
identifying them is one important next step in this line of research.

Last, we note that our classical-to-quantum reduction provides yet
another way to think about quantum algorithms.
For example, once mapped to a one dimensional system, the glued-trees problem of masses and springs becomes a simple wave propagation problem.
Since waves propagate ballistically, it is 
now clear why our quantum algorithm, whose complexity is mainly dominated by evolution time but not by the number of masses, works in this case. Another example results from Grover's classical system of coupled pendulums~\cite{grover2002classical} that, after using our reduction, provides another way to solve the unstructured search problem with $\cO(t)$ queries, where $t=\cO(\sqrt N)$.

\subsection*{Acknowledgments}

The authors thank Amira Abbas, Sergio Boixo, Pedro Costa, Eddie Farhi, Bill Huggins, Marika Kieferova, Jarrod McClean, Hartmut Neven, Philipp Schleich, and Alexander Schmidhuber for helpful discussions. DWB and NW were funded to work on this project by grants from Google Quantum AI.
DWB is also supported by Australian Research Council Discovery Projects DP190102633, DP210101367, and DP220101602.  NW's work was also supported by the U.S. Department of Energy, Office of Science, National Quantum Information
Science Research Centers, Co-design Center for Quantum
Advantage under contract DE-SC0012704.


\onecolumngrid
\bibliography{ryan_references,coupledoscillators}
\section*{Appendix}
\appendix

\section{Access models and block encoding of {\bf B} and {\bf H}}
\label{app:accessmodel}

We provide more details on the access model
used for our quantum algorithm, which is similar to that used in prior quantum algorithms~\cite{BAC07,BC12}. 
The black box or oracle $\cS$
 allows us to perform the map
\begin{align}
\label{eq:blackboxneighbor}
\ket{j,\ell}\rightarrow \ket{j,a(j,\ell)} \;,
\end{align}
for any $j \in [N]:=\{1,\ldots,N\}$ and $\ell \in [d]$, where $d$ is the maximum number of nonzero entries in
any row of ${\bf K}$ (i.e., the sparsity), and $a(j,\ell)$ is the column index of the $\ell^{\rm th}$ nonzero entry in the $j^{\rm th}$ row of ${\bf K}$.
The same black box allows us to perform the maps
$\ket{j,k,z} \rightarrow \ket{j,k,z\oplus \bar \kappa_{jk}}$
and $\ket{j,z} \rightarrow \ket{j,z \oplus \bar m_j}$, for any $j,k \in [N]$, where $z$, $\bar \kappa_{jk}$, and $\bar m_j$
are assumed to be given as bit strings. (One could use three different oracles for these maps, but we combine them into one that we call $\cS$.)
Specifically, if $m_{\max} \ge m_j$ for all $j\in[N]$, then 
$\ket{\bar m_j}$ denotes a basis state determined from the bits in the binary fraction
$m_j = m_{\max} [.b_{j,1} b_{j,2} \ldots]$, where $b_{j,i} \in \{0,1\}$. Similarly, if $\kappa_{\max} \ge \kappa_{jk}$ for all $j,k \in [N]$, then $\ket{\bar \kappa_{jk}}$ denotes a basis state determined from the bits in the binary fraction 
$\kappa_{jk}=\kappa_{\max} [.c_{jk,1} c_{jk,2}\ldots]$, where
$c_{jk,i} \in \{0,1\}$. This access model simplifies some calculations and does not require knowing the important quantities in any specific units. 

Similar to Refs.~\cite{BAC07,BC12}, we assume that any number of bits can be given as the output of $\cS$, with it still being accounted for as a single oracle call. However, for computations with this output, we will only use a limited number of bits, the choice of which is governed by the parameters of the problem, particularly the approximation errors. 
In particular, we denote by $r_m$ and $r_\kappa$ the number of bits used to represent $m_j$ and $\kappa_{jk}$.
Then $m_j=m_{\max} \bar m_j/2^{r_m}$ and $\kappa_{jk}=\kappa_{\max} \bar \kappa_{jk}/2^{r_\kappa}$.
These approximations as well as the choices for $r_m$ and $r_\kappa$ are discussed below, where we show how to use $\cS$ to gain access to ${\bf B}$ or ${\bf H}$.

\subsection{Access to {\bf B} and {\bf H}}

We explain how to use $\cS$ and other gates to access 
${\bf H}$, the Hamiltonian in Eq.~\eqref{eq:hamiltonian}. This access is required by Hamiltonian simulation methods in Refs.~\cite{BCC+15,BCK15,LC17,LC19}. We do this by showing a block encoding of the relevant matrices using known methods for state preparation via inequality testing~\cite{SandersPRL18}, i.e., by constructing a unitary such that one of its blocks approximates ${\bf H}/\Lambda$, where $\Lambda$ is needed for normalization reasons and given below. We explain how to 
implement that unitary operation using elementary gates.

In our construction, we start by
providing a block encoding of ${\bf B}$ (and ${\bf B}^\dagger$). This is the $N \times M$ matrix discussed in Sec.~\ref{sec:quantumevolution}.
Because $M=N(N+1)/2$ is not a power of two in general ($N=2^n$), 
it is convenient to pad ${\bf B}$ with zeroes so that its dimension is $N \times N^2$, instead. That is, here we consider and describe the simulation of ${\bf H}$ in a larger-dimensional space, where many amplitudes of the evolved state will be zero and the others will coincide with those of $\sket{\psi(t)}$ in Eq.~\eqref{eq:encodedstate}. With this padding, the matrices ${\bf B}$ and ${\bf H}$ will be then of dimension $N \times N^2$ and $2N^2 \times 2N^2$, respectively, in the following analyses.
First we prove:

\begin{lemma}
\label{lem:blockencodingB}
Let $\epsilon' >0$, $m_{\max} \ge m_j \ge m_{\min}$ and $\kappa_{\max} \ge \kappa_{jk}$ for all $j,k \in[N]$, and $\aleph := \kappa_{\max}/m_{\min}$. Consider the $d$-sparse $N \times N^2$ matrix ${\bf B}$ 
obtained from Sec.~\ref{sec:quantumevolution} (with the padding). Then, there exists a unitary $\cU_{\bf B}$ acting on $2n+r+2$ qubits with $r=\cO(\log(1/\epsilon'))$ that provides a block encoding of ${\bf B}$ to within additive error $\epsilon'$ as follows:
\begin{align}
\left \| (\one_N \otimes \bra 0^{\otimes n} \otimes \bra {0}^{\otimes r+2} ) \cU_{\bf B} (\one_N \otimes \one_N \otimes\ket {0}^{\otimes r+2}) - \frac{1}{\sqrt{2\aleph d}} {\bf B} \right \| \le \epsilon' \;.
\end{align}
The quantum circuit that implements $\cU_{\bf B}$ makes $\cO(1)$ uses of $\cS$ and its inverse, in addition to $\cO(n + \log^2(m_{\max}/( m_{\min}\epsilon')) )$ two-qubit gates. 
\end{lemma}

Next, we use this result to provide a block encoding for the $2N^2 \times 2N^2$ matrix ${\bf H}$:

\begin{lemma}
\label{lem:blockencodingH}
Let $\epsilon' >0$, $m_j \ge m_{\min}$ and $\kappa_{\max} \ge \kappa_{jk}$ for all $j,k \in[N]$, and $\aleph := \kappa_{\max}/m_{\min}$.
Consider the $d$-sparse $2N^2 \times 2N^2$ matrix ${\bf H}$ obtained from Sec.~\ref{sec:quantumevolution} (with the padding). Then, there exists a unitary $\cU_{\bf H}$ acting on $2n+r+4$ qubits with $r=\cO(\log(1/\epsilon'))$ that provides a block encoding of ${\bf H}$ to within additive error $\epsilon'$ as follows:
\begin{align}
 \left \| \bra 0^{\otimes r+3}  \cU_{\bf H} \ket 0^{\otimes r+3} - \frac 1 {\sqrt{2 \aleph d}}{\bf H} \right \| \le \epsilon' \;.
\end{align}
The quantum circuit that implements $\cU_{\bf H}$ uses a controlled version of $\cU_{\bf B}$ and its inverse once, and additional $\cO(n)$ two-qubit gates (e.g., CNOTs and other single qubit gates).
\end{lemma}

Combining Lem.~\ref{lem:blockencodingB} and Lem.~\ref{lem:blockencodingH}, a block encoding for ${\bf H}/\Lambda$, where $\Lambda:= \sqrt{2 \aleph d}$, can be implemented within additive error $\cO(\epsilon')$ in spectral norm
using the oracles $\cS$, together with its inverse and controlled versions, $\cO(1)$
times, in addition to $\cO(n + \log^2(m_{\max}/( m_{\min}\epsilon')))$ two-qubit gates. 

\subsection{State preparation using inequality testing}

Before presenting the proofs of the lemmas, we revisit how inequality testing can be used for state preparation \cite{SandersPRL18}, since our constructions use this approach, which is more efficient than controlled rotation of an ancilla qubit as in Ref.~\cite{HHL}.
Let $\{\beta_1, \ldots, \beta_N\}$ be such that $\beta_j \ge 0$ can be expressed using $r$ bits (see below) and $\beta \ge \beta_j$ for all $j \in [N]$. Assume we have access to an oracle that, on input $\ket{j,z}$ outputs $\ket{j,z \oplus \bar \beta_j}$, where $z$ and $\bar \beta_j \in [2^r]$ are given as bit strings of size $r$ (we will fix $z=0\ldots 0$). That is, in this notation, $\sket{\bar \beta_j}$ is a  basis state, and the bits are obtained from the binary fraction $\beta_j=\beta [.b_1 \ldots b_r]$. The method 
in Ref.~\cite{SandersPRL18} can be used to perform:
\begin{align}
\label{eq:inequalitytesting}
    \ket j \ket 0 \rightarrow \ket j \left(\frac{\beta_j}{\beta} \ket{0}^{\otimes r} \ket{0} + \ket{\omega_j} \right) \; ,
\end{align}
where $\ket{\omega_j}$ is a state orthogonal to $\ket{0}^{\otimes r} \ket{0}$.
The basic steps of the method are as follows:
\begin{enumerate}
    \item Apply the oracle to compute $\bar \beta_j$, i.e.,  perform the map $\ket j \mapsto \ket {j,\bar \beta_j}$. 
    \item Prepare the equal superposition state of $r$ ancilla qubits, i.e., $\frac 1 {2^{r/2}} \sum_{x=1}^{2^r} \ket x$.
    \item Apply inequality testing to prepare 
    $\ket {j,\bar \beta_j} \frac 1 {2^{r/2}} \left(\sum_{x=1}^{\bar \beta_j} \ket x \ket 0+ \sum_{\bar \beta_j+1}^{2^r} \ket x \ket 1 \right)$.
    \item Apply Hadamards on the $r$ ancilla qubits.
    \item Apply the inverse of the oracle to reverse the computation of $\bar \beta_j$, i.e., perform the map $\ket {j,\bar \beta_j}\mapsto \ket{j,0}$.
\end{enumerate}
After these operations there will be amplitude $\bar\beta_j/2^r = \beta_j/\beta$ on $\ket{0}^{\otimes r} \ket{0}$.
This is because the amplitude can be found by taking the inner product of the state in step 3 with $\frac 1 {2^{r/2}} \sum_{x=1}^{2^r} \ket x \ket 0$.
The method requires one use of the oracle and its inverse, 
in addition to $\cO(r)$ Hadamard gates for steps 1 and 4, and $\cO(r)$ gates for the inequality test in step 3.
It is also possible to compute more complicated functions of the oracle by rearranging the inequality in step 3, which is the method we will use here.

\subsection{Proof of Lemma~\ref{lem:blockencodingB}}

    We consider a block encoding of ${\bf B}^\dagger$ for simplicity,
and later use it to provide a block encoding of ${\bf B}$ by taking the conjugate transpose. Using the padding described above, the dimension of ${\bf B}^\dagger$
is $N^2 \times N$ and its entries are of the form $\sqrt{\kappa_{jk}/m_j}$ or zero; specifically, the definition of ${\bf B}$ in Sec.~\ref{sec:quantumevolution} (obtained from Eq.~\eqref{eq:Bdaggeractionmain}) implies
\begin{align}
\label{eq:Bdaggeraction}
      {\bf B}^\dagger \ket j =   \left(\sum_{k \ge j}  \sqrt{\frac{\kappa_{jk}}{{m_j }}}\ket{j} \ket k  \right)-\left(\sum_{k<j}\sqrt{\frac{\kappa_{jk}}{{m_j }}} \ket{k}\ket {j}  \right)\;.
\end{align}
We will construct a unitary $\cU_{\bf B}^\dagger$ that is a block encoding of ${\bf B}^\dagger$ following simple steps and then explain how these can be simulated with quantum circuits.
The steps require using a work register of qubits for some computations
that we discard at the end.

\begin{enumerate}
    \item Apply a unitary that performs the map
\begin{align}
\label{eq:equalsuperposition}
    \ket 0^{\otimes n} \mapsto \frac 1 {\sqrt d} \sum_{\ell=1}^d \ket {\ell} \;.
\end{align} 
\item Apply the oracle $\cS$ for the positions of non-zero entries of ${\bf K}$ to map $\ket{\ell}$ to $\ket{a(j,\ell)}$ according to Eq.~\eqref{eq:blackboxneighbor}.
\item Apply the oracle $\cS$ two more times to compute $\bar m_j$ and $\bar \kappa_{jk}$, i.e.,  perform the map $\ket{j,0} \mapsto \ket{j,\bar m_j}$ and $\ket{j,k,0}\mapsto \ket{j,k,\bar \kappa_{jk}}$.
    \item Prepare the equal superposition state of $r$ ancilla qubits, i.e., $\frac 1 {2^{r/2}} \sum_{x=1}^{2^r} \ket x$, where $r$ is given below.
\item Using coherent arithmetic, compute the square of $x$ and multiplications needed for the inequality test
\begin{equation}\label{eq:Bblockineq}
       \frac{  \kappa_{\max} \bar\kappa_{jk}}{2^{r_\kappa}} \le \frac{x^2}{2^{2r}} \aleph  \frac{ m_{\max} \bar m_j} {2^{r_m}} \, ,
\end{equation}
where $r_\kappa$ and $r_m$ are the numbers of bits of $\bar\kappa_{jk}$ and $\bar m_j$, respectively, also determined below.
This gives the factor in the amplitude approximately $\sqrt{\kappa_{jk}/(m_j \aleph)}$. 
\item Reverse the computation of the arithmetic for Eq.~\eqref{eq:Bblockineq} and the oracles for $\bar m_j$ and $\bar \kappa_{jk}$.
This transforms the working registers back to an all-zero state.
\item Apply inequality testing (outputting the result in an ancilla qubit) and a controlled-swap operation to perform the map $\ket j \ket k \ket 0 \mapsto \ket k \ket j \ket 1$ if $k<j$ or leave the state $\ket j \ket k \ket 0$ invariant otherwise.
\item Apply $HZ$ on the ancilla qubit of the previous step to implement $\ket 0 \mapsto \frac 1 {\sqrt 2} (\ket 0 + \ket 1)$ and $\ket 1 \mapsto \frac 1 {\sqrt 2} (-\ket 0 + \ket 1)$.
\end{enumerate}

The previous sequence of unitaries define $\cU_{\bf B}^\dagger$.
To show that the method is correct, consider steps 1-6.
If we discard the working registers (used to store $\bar m_j,\bar \kappa_{jk}$ and perform the arithmetic), these steps combined implement approximately
\begin{align}\label{eq:step5}
    \ket j \ket 0^{\otimes n+r} \ket 0& \mapsto 
     \ket j \frac 1 { \sqrt{\aleph d}}  \sum_{k} \ket k  \sqrt{\frac{\kappa_{jk}}{m_j}} \ket 0^{\otimes r}\ket 0 +\ket{\omega_j} \;,
\end{align}
for some state $\ket{\omega_j}$ that is orthogonal to $\ket 0^{\otimes r}\ket 0$ on the ancilla qubits.
We want the the error in this approximation to be $\cO(\epsilon')$, and we show how to achieve this below.
Step 7 is then needed to rearrange the sum depending on whether $k<j$ or $k\ge j$, according to Eq.~\eqref{eq:Bdaggeraction}.
This step transforms Eq.~\eqref{eq:step5} to
\begin{align}
    \label{eq:step7}
      \frac 1 { \sqrt{\aleph d}}  \sum_{k \ge j}\ket j \ket k  \ket 0 \sqrt{\frac{\kappa_{jk}}{m_j}}\ket 0^{\otimes r}\ket 0  +  \frac 1 { \sqrt{\aleph d}}  \sum_{k < j} \ket k \ket j \ket 1 \sqrt{\frac{\kappa_{jk}}{m_j}}\ket 0^{\otimes r}\ket 0 + \ket{\omega'_j}   \; ,
\end{align}
where $\ket{\omega'_j}$ is still orthogonal to $\ket 0^{\otimes r}\ket 0$ on the ancilla qubits.
Applying step 8 to Eq.~\eqref{eq:step7}, and considering the part of the state where the ancillas are in $\ket 0\ket 0^{\otimes r}\ket 0$ only, we obtain
\begin{align}
    \label{eq:step8}
      \frac 1 { \sqrt{ 2 \aleph d}} \left( \sum_{k \ge j}\sqrt{\frac{\kappa_{jk}}{m_j}} \ket j \ket k  -  \sum_{k < j} \sqrt{\frac{\kappa_{jk}}{m_j}} \ket k \ket j   \right) \ket 0\ket 0^{\otimes r}\ket 0 \;.
\end{align}
This coincides with Eq.~\eqref{eq:Bdaggeraction} if we drop the normalization factor and project onto the subspace specified by $\ket 0\ket 0^{\otimes r}\ket 0$ of the ancillas. Hence, the sequence of steps defines the desired unitary $\cU_{\bf B}^\dagger$ that satisfies, for all $j \in [N]$,
\begin{align}
    \bra 0^{\otimes r+2}\cU_{\bf B}^\dagger \ket j \ket{0}^{\otimes n}\ket{0}^{\otimes r+2} \approx_{\epsilon'} \frac 1 { \sqrt{2 \aleph d}} {\bf B}^\dagger \ket j \;.
\end{align}
Equivalently, $(\one_N \otimes \bra{0}^{\otimes n}\otimes \bra 0^{\otimes r+2})\cU_{\bf B}  (\one_N \otimes \one_N \otimes \ket{0}^{\otimes r+2}) \approx_{\epsilon'}\frac 1 { \sqrt{2\aleph d}} {\bf B}$.

Note that the factor of $1/\sqrt{d}$ comes from a single sparse state preparation (with amplitudes $\sqrt{\kappa_{jk}/m_j}$).
In contrast, the block encoding based on \cite{BC12} involves a matched preparation and inverse preparation, each of which gives a factor of $1/\sqrt{d}$ for an overall factor of $1/d$ for Hamiltonian simulation.
In particular, the state preparation as in Lemma 4 of \cite{BC12} gives a factor of $1/\sqrt{d}$, then the step of the walk as in Eq.~(23) of \cite{BC12} involves a matched preparation and inverse preparation to implement the reflection.

The above steps can be  
implemented with a quantum circuit as follows. The unitary in step 1
 (Eq.~\eqref{eq:equalsuperposition}) is simply obtained from the action of $\log(d)$ Hadamard gates if $d$ is a power of 2, or otherwise can be performed with high precision by amplitude amplification (see App.~E.2 of \cite{Sanders2020}).
The gate complexity is $\mathcal{O}(\log d)=\mathcal{O}(n)$, with an amplitude for success that is very close to 1.
When allowing arbitrary qubit rotations in the gate set, as we do here, then the amplitude for success can be made exactly one, so we need no correction here.
The gate complexity in step 4 depends on the number of bits needed to represent $m_j$, $\kappa_{jk}$, and $x$ with sufficient accuracy.
The leading-order gate complexity is $\cO(r^2+rr_m)$ for computing $x^2$ and $x^2\times \bar m_j$.
We also need to multiply by $\aleph m_{\max}/\kappa_{\max}=m_{\max}/m_{\min}$.
The number of bits needed is no more than that in $x^2\times \bar m_j$, which is $\cO(r+r_m)$, giving gate complexity $\cO((r+r_m)^2)$.
The divisions by powers of 2 just involve bit shifts with no gate cost.
The gate complexity for the inequality test is linear in the number of bits, so is smaller than the multiplication cost.

It suffices to set $r=\cO(\log(1/\epsilon'))$ for overall precision $\cO(\epsilon')$, since this choice would imply that the coefficients in Eq.~\eqref{eq:step5} are given with that precision.
To choose $r_m,r_\kappa$, we note that our computations of $m_j$ and $\kappa_{jk}$ are effectively done within additive error $\delta_m=\cO(m_{\max} /2^{r_m})$ and
$\delta_\kappa=\cO(\kappa_{\max} /2^{r_\kappa})$, respectively, since $\bar m_j$ and $\bar \kappa_{jk}$ give $m_j$ and $\kappa_{jk}$ relative to $m_{\max}$ and $\kappa_{\max}$. That is, we are effectively giving the factor in the amplitude
\begin{align}
    \sqrt{\frac{\kappa_{jk}+\delta_\kappa}{(m_j+\delta_m)\aleph}}\;.
\end{align}
For additive error $\cO(\epsilon')$ in the above, we can let $\delta_m=\cO(m_j \epsilon')$, i.e., $m_j$ is computed within multiplicative error $\cO(\epsilon')$.
This choice would imply $r_m =\cO(\log(m_{\max}/(m_{\min} \epsilon'))$.
Similarly using the error propagation formula for $\kappa_{jk}$ indicates that the error is largest for small $\kappa_{jk}$, so the choice of $r_\kappa$ would depend on $\kappa_{\min}$.
However, the worst the error can be is rounding down a value by $\delta_\kappa$ to zero, which would imply an error $\cO(\sqrt{\delta_\kappa/\kappa_{\max}})$.
That implies we can choose $\delta_\kappa =\cO(\kappa_{\max} (\epsilon')^2)$ and therefore $r_\kappa=\cO(\log(1/ \epsilon'))$.

These choices of numbers of bits imply the complexity of the squaring and multiplications at most $\cO(\log^2(m_{\max}/(m_{\min} \epsilon'))$.
Last, all unitaries in steps 7 and 8 can be implemented with gate complexity $\cO(n)$ using standard techniques.
These imply the gate complexities stated in Lem.~\ref{lem:blockencodingB}.
    
\qed

\subsection{Proof of Lemma~\ref{lem:blockencodingH}}

The result of Lem.~\ref{lem:blockencodingH}
is a direct consequence of Lem.~\ref{lem:blockencodingB}. 
Recall that
\begin{align}
    {\bf H}&=-\begin{pmatrix} {\bf 0}& {\bf B}\cr {\bf B}^\dagger & {\bf 0}\end{pmatrix} \;,
\end{align}
where we use ${\bf 0}$ to denote all-zero matrices whose dimensions are clear from context.
Assuming that ${\bf B}$ is of dimension $N \times N^2$ following Lem.~\ref{lem:blockencodingB}, then ${\bf H}$ would be of dimension $(N+N^2)\times (N+N^2)$. However, because it is easier to work 
with square matrices, we will further pad ${\bf B}$ and ${\bf B}^\dagger$ with more zeroes, to make them of dimension $N^2 \times N^2$, implying that the dimension of ${\bf H}$ is now $(2N^2)\times(2N^2)$. That is, ${\bf H}$
acts 
on a space of $2n+1$ qubits. (The evolved state will have nonzero amplitude on a subspace of dimension $N+M$ only, where $M=N(N+1)/2$.) The previous padding is equivalent
to replacing ${\bf B} \ket{j,k} \rightarrow ({\bf B} \ket{j,k}) \otimes \ket 0^{\otimes n}$ and 
${\bra{j,k} }{\bf {B}}^\dagger \rightarrow ({\bra{j,k} }{\bf {B}}^\dagger) \otimes \bra 0^{\otimes n}$.

With this small modification,
Lem.~\ref{lem:blockencodingB} implies
\begin{align}
&\frac {\bf H} {\sqrt{2\aleph d}}= \nn
&-\begin{pmatrix} {\bf 0}& (\one_N \otimes \ketbra 0^{\otimes n} \otimes \bra 0^{\otimes r+2} ) \cU_{\bf B} (\one_N \otimes \one_N \otimes \ket 0^{\otimes r+2})\cr  (\one_N \otimes \one_N \otimes \bra 0^{\otimes r+2})\cU_{\bf B}^\dagger (\one_N \otimes \ketbra 0^{\otimes n} \otimes \ket 0^{\otimes r+2} )  & {\bf 0}\end{pmatrix},
\end{align}
or, equivalently, $\frac {\bf H} {\sqrt{2\aleph d}}=\bra 0^{\otimes r+2} \tilde \cU_{\bf H} \ket 0^{\otimes r+2}$,
where
\begin{align}
    \tilde \cU_{\bf H}: =- \ket 0 \! \bra 1 \otimes ((\one_N \otimes \ketbra 0^{\otimes n}\otimes  \one_2^{\otimes r+2})  \cU_{\bf B}) + H.c. .
\end{align}
Here, $H.c.$ denotes the conjugate transpose of the first term.
This operator is not yet unitary because $\ketbra 0^{\otimes n}$ is not unitary (it is a projector). However, it is simple to construct a block encoding for a projector as follows.  We bring one additional ancilla and implement a conditional unitary operation on the state of the ancilla
that is 
\begin{align}
    \ketbra + \otimes (2 \ketbra 0^{\otimes n} - \one_N) + \ketbra - \otimes \one_N \;.
\end{align}
Because $\langle 0 \ketbra + 0\rangle=\langle 0 \ketbra -  0 \rangle=1/2$, 
this gives
\begin{align}
    \bra 0 \left(  \ketbra + \otimes (2 \ketbra 0^{\otimes n} - \one_N) + \ketbra - \otimes \one_N \right) \ket 0 = \ketbra 0^{\otimes n} \;.
\end{align}
Applying $\bra 0 \cdots \ket 0$ to this operator gives the block encoding
\begin{align}
    \bra 0 \left( \ketbra + \otimes (2 \ketbra 0^{\otimes n} - \one_N) + \ketbra - \otimes \one_N \right) \ket 0 = \ketbra 0^{\otimes n} \;,
\end{align}
which implements the desired projector.
We write $\cU_{\rm cond}$ for this unitary when acting on the system of $2n+1$ qubits (the space associated with ${\bf H}$) plus $r+3$ ancilla qubits. Hence, our block encoding for ${\bf H}$ is
\begin{align}
  \frac 1 {\sqrt{2\aleph d}}  {\bf H} = \bra 0^{\otimes r+3}  \cU_{\bf H} \ket 0^{\otimes r+3} \;,
\end{align}
where
\begin{align}
    \cU_{\bf H} : = -\cU_{\rm cond} (\ket 0 \! \bra 1 \otimes \one_2 \otimes \cU_{\bf B}) + H.c..
\end{align}
Here $\one_2$ acts on the extra qubit used for block encoding the projector.

Simulating $\cU_{\bf H}$ with a quantum circuit requires applying a controlled version of $\cU_{\bf B}$ and $\cU_{\bf B}^\dagger$ once, in addition to a simple $X$ gate on the controlled qubit.
Simulating $\cU_{\rm cond}$ can be done with gate complexity $\cO(n)$, since it requires applying a conditional phase on $\ket{0}^{\otimes n}$.

\qed

\section{Oracle lower bound and proof of Lemma~\ref{lem:kineticgluedtrees}}
\label{app:oracleproof}

Consider the glued-trees oscillator network in Fig.~\ref{fig:gluedtrees}. 
Let any node represent a unit mass (i.e., $m_j=m_{\max}=m_{\min}=1$ for all $j\in[N]$) and each edge represent a spring of constant 1 (i.e., $\kappa_{jk}=\kappa_{\max}=1$ if $(j,k)$ is an edge and $\kappa_{jk}=0$ otherwise). The ${\rm ENTRANCE}$ and ${\rm EXIT}$ masses are additionally connected to a wall each with a spring of constant also $1$. We write $\vec x(t) \in \mathbb R^N$ and $\dot {\vec x}(t) \in \mathbb R^N$ for the positions and velocities the masses, where we assume their motion is constrained to one spatial dimension, such as the ``horizontal'' direction. 
Although we do not know the actual names of the masses a priori, our labels are such that $j=1$ refers to the ${\rm ENTRANCE}$, $j=2,3$ refer to the masses in the second column, and so on, until $j=N=2^{n+1}-2$ represents the ${\rm EXIT}$ mass. The energy of the system is $E=T(t)+U(t)$, where 
\begin{align}
T(t)&=\frac 1 2 \sum_{j=1}^N (\dot x_j(t))^2 \; , \\
U(t)& = \frac 1 2 \sum_{j,k>j} \kappa_{jk}(x_j(t)-x_k(t))^2 + \frac 1 2 \kappa_{11} (x_1(t))^2 + \frac 1 2 \kappa_{NN} (x_N(t))^2\;,
\end{align}
are the kinetic and potential energies, respectively.
Newton's equation gives then a set of $N$ coupled second-order differential equations:
\begin{align}
\label{eq:newtongluedtrees}
\ddot{\vec x}(t)= -{\bf A} {\vec x}(t) \;,
\end{align}
where ${\bf A}$ is the $N \times N$-dimensional symmetric matrix
\begin{align}
{\bf A} =\begin{pmatrix}
    3 & -1 & -1 & 0 & \cdots & \cdots & \cdots & \cdots & 0 \cr
-1 & 3 & 0 & -1 & \cdots & \cdots & \cdots & \cdots & 0 \cr
-1 & 0 & 3 & 0 & \cdots & \cdots & \cdots & \cdots & 0\cr
0 & -1 & 0 & 3 & \cdots & \cdots & \cdots & \cdots & 0\cr
\vdots & \vdots & \vdots & \vdots & \ddots & \vdots & \vdots & \vdots & \vdots \cr
 \vdots & \vdots & \vdots & \vdots &  \vdots &  3 & 0 &- 1 & 0 \cr
 \vdots & \vdots & \vdots & \vdots & \vdots & 0 & 3 & 0 & -1 \cr
 \vdots & \vdots & \vdots & \vdots & \vdots & -1 & 0 & 3 & -1 \cr
0 & 0 & 0 & 0 & \cdots & 0 & -1 & -1 & 3
\end{pmatrix} \;.
\end{align}
Note that ${\bf A}= 3 I_N - A$, where $A$ is the adjacency matrix of the graph constructed from the two binary trees randomly glued, if we disregard the edges that connect
the roots to their respective walls, and $I_N$
is the $N \times N$ identity matrix. Also, ${\bf A}$
is positive semi-definite because it is symmetric and diagonally dominant. This property also follows from the potential being
$U(t)=\frac 1 2 \vec x(t)^T {\bf A} \vec x(t) \ge 0$
for all $\vec x(t) \in \mathbb R^N$, which implies ${\bf A}\succcurlyeq 0$. In fact, ${\bf A}$ is positive definite: The matrix ${\bf A}'={\bf A}-\ketbra 1 - \ketbra N$ represents the above network of oscillators where the masses at the roots are not connected to any wall. Hence, ${\bf A}'$ contains a non-degenerate eigenvector of eigenvalue zero corresponding to the translations of the system, i.e., the eigenvector $\ket{ u_1}:=\frac 1 {\sqrt N} \sum_j \ket j$. The only way for ${\bf A}$ to contain an eigenvalue zero is if this eigenvector $\ket{ u_1}$ was also an eigenvector of $\ketbra 1 + \ketbra N$, which is not the case. Nevertheless,
while ${\bf A}\succ 0$, its smallest eigenvalue is exponentially small in $n$ since the expectation $\bra{u_1} {{\bf {A}}} \ket{ u_1} =2/N$ is an upper bound on the smallest eigenvalue.

We would like to show the following property. If $ {\vec x}(0)=(0,0,\ldots,0)^T$ and $\dot {\vec x}(0)=(1,0,\ldots,0)^T$,
then $\dot{\vec x}(t)$ is such that the magnitude of its $N^{\rm th}$ entry, corresponding to ${\rm EXIT}$, is at least polynomially small in $n$ for a time $t$ that is at most polynomial in $n$ (i.e., the kinetic energy of the $N^{\rm th}$ oscillator is $\Omega(1/\poly(n))$). 
For these initial conditions, the solution to Eq.~\eqref{eq:newtongluedtrees} implies
\begin{align}
    \dot {\vec x}(t)=\cos \left(t \sqrt{\bf A} \right) \dot {\vec x}(0).
\end{align}
The matrix ${\bf A}$, and hence $\sqrt{\bf A}$,
possesses a symmetry that allows one to simulate the dynamics
of the $N$ oscillators by considering that of $2n$ oscillators in one spatial dimension, instead. (A similar idea was used in Ref.~\cite{childs2003exponential} to prove that the quantum algorithm solves the problem efficiently.) To show this,
we define $2n$ real components:
\begin{align}
    z_l(t) = \frac 1 {\sqrt{N_l}}\sum_{j \in l^{\rm th} \ {\rm column}} x_j(t) \;,
\end{align}
where $N_l$ is the number of masses in the $l^{\rm th}$ column and $l \in [2n]$; that is $N_l=2^{l-1}$ if $l \le n$ and $N_l=2^{2n-l}$ if $l \ge n+1$.
Then, for the masses in any column that is not the ${\rm ENTRANCE}$ ($l=1)$, ${\rm EXIT}$ ($l=2n)$, or the randomly glued ones where $l=n$ and $l=n+1$, simple manipulations of Eq.~\eqref{eq:newtongluedtrees} give
\begin{align}
    \ddot  z_{l}(t) = \sqrt{2} z_{l-1}(t)-3 z_{l}(t) +\sqrt{2} z_{l+1}(t) \;.
\end{align}
For the masses in the remaining columns, Eq.~\eqref{eq:newtongluedtrees} gives
\begin{align}
  \ddot  z_1(t)  &= -3   z_1(t) + \sqrt{2} z_2(t)  \\
   \ddot z_{n}(t) &= \sqrt{2} z_{n-1}(t)-3 z_n(t)  +2z_{n+1}(t) \\
   \ddot z_{n+1}(t)&=2z_{n}(t) -3 z_{n+1}(t) + \sqrt{2} z_{n+2}(t)  \\
   \ddot  z_{2n}(t) &= \sqrt{2} z_{2n-1}(t) -3   z_{2n}(t)\;.
\end{align}
Then, if $\vec z (t):=(z_1(t),\ldots,z_{2n}(t))^T \in \mathbb R^{2n}$, Eq.~\eqref{eq:newtongluedtrees} implies
\begin{align}
\label{eq:1Dmotion}
    \ddot{\vec z}(t) = - \tilde {\bf A} \vec z (t) \;,
\end{align}
where  $\tilde {\bf A}$ is a $(2n) \times (2n)$ real-symmetric and tridiagonal matrix that, in bra-ket notation, can be written as
\begin{align}
\nonumber
\tilde {\bf A}   &=  3 \sum_{l=1}^{2n} \ketbra l -\sqrt 2 \left(\sum_{l=1}^{n-1} \left(\ket {l}\! \bra{l+1}+\ket{l+1}\!\bra l  \right) + \sum_{l=n+1}^{2n-1} \left(\ket {l} \!\bra{l+1}+\ket{l+1}\!\bra l  \right)\right) \\ &  \quad  -2 \left(\ket{n}\!\bra{n+1}+\ket{n+1}\!\bra{n} \right).
\end{align}

\begin{figure}[b]
    \centering
    \begin{tikzpicture}[scale=1.7]
\foreach \i/\label in {1/$z_1$,2/$z_2$,3/$z_3$,4/$z_{n-1}$,5/$z_n$,6/$z_{n+1}$,7/$z_{n+2}$,8/$z_{2n-2}$,9/$z_{2n-1}$,10/$z_{2n}$}
\node[circle,draw,fill=black,inner sep=2pt,label=below:\label] (v\i) at (\i-1,0) {};
\foreach \i in {1,2,3,4,5,6,7,8,9,10}
\draw [-] (v1)--node[above]{$\sqrt{2}$}(v2);
\draw [-] (v2)--node[above]{$\sqrt{2}$}(v3);
\draw [-, dotted] (v3)--node[above]{}(v4);
\draw [-] (v5)--node[above]{$\sqrt{2}$}(v4);
\draw [-] (v5)--node[above]{$2$}(v6);
\draw [-, dotted] (v7)--node[above]{}(v6);
\draw [-] (v7)--node[above]{$\sqrt{2}$}(v8);
\draw [-] (v8)--node[above]{$\sqrt{2}$}(v9);
\draw [-] (v9)--node[above]{$\sqrt{2}$}(v10);
\end{tikzpicture}
    \caption{The $2n \times 2n$ matrix $\tilde{A} = 3I_{2n}- \tilde {\bf A}$ viewed as the adjacency matrix of a weighted graph.}
    \label{fig:line}
\end{figure}
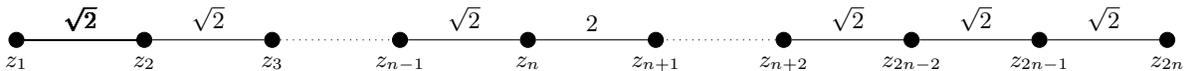

Hence, we have effectively reduced the dimension  of the system from $N$ to $2n$ when symmetries are considered.
The matrix $\tilde {\bf A}$ is $3I_{2n}-\tilde A$, where $\tilde A$ is the $2n \times 2n$ adjacency matrix of the graph constructed from the glued binary trees in the new coordinate system defined above (see Fig.~\ref{fig:line}).
The matrix $\tilde {\bf A}$ is also positive, with its smallest eigenvalue being exponentially small in $n$, and any two eigenvalues separated by a gap that is at least $\Delta'=\Omega(1/n^3)$.
In particular, if $\vec z(0)=(0,0,\ldots,0)^T$, the solution to Eq.~\eqref{eq:1Dmotion} implies
\begin{align}
   \dot{ \vec z}(t) =\cos(t \sqrt{\tilde {\bf A}}) \dot{ \vec z}(0) \;.
\end{align}

We are interested in the magnitude of $\dot z_{2n}(t) \equiv \dot x_N(t)$ or, more specifically, the kinetic energy of the $N^{\rm th}$ oscillator.
Similar to Ref.~\cite{childs2003exponential}, rather than considering a fixed $t$, 
we will consider an average over $t$ as follows: 
\begin{align}
\label{eq:pexitgluedtrees}
P_{\rm EXIT}(T):= \frac 1 {T} \int_0^T dt |\dot x_N(t)|^2 = \frac 1 {T} \int_0^T dt |\dot z_{2n}(t)|^2 \;.
\end{align}
Here, $T > 0$ is set below and $P_{\rm EXIT}(T)$ is the average kinetic energy of the $N^{\rm th}$ oscillator
from time $t=0$ to $t=T$ renormalized by the energy $E=1/2$ for the above initial conditions; for our algorithm, $P_{\rm EXIT}(T)$
coincides with the average probability of projecting $\sket{\psi(t)}$ into a basis state 
that correspond to the ${\rm EXIT}$ mass.
Thus, the goal reduces to showing that there exists $T=\cO({\poly}(n))$ such that $P_{\rm EXIT}(T)=\Omega(1/{\poly}(n))$. If we prove this, then Eq.~\eqref{eq:pexitgluedtrees} automatically implies the existence of $t=\cO(\poly(n))$ such that $K_N(t)=\Omega(1/\poly(n))$, which is our main goal. 
Finding such $t$ can be done classically by simulating
Eq.~\eqref{eq:1Dmotion} in time polynomial in $n$.

We can write
\begin{align}
    \dot z_{2n}(t) 
    & = \bra{2n}\frac 1 2 \left( e^{i t \sqrt{\tilde \bA}} + e^{-i t \sqrt{\tilde \bA}}\right) \ket{1} \;. 
\end{align}
Let $\{\ket {\lambda _1},\ldots,\ket {\lambda _{2n}}\}$ be the $2n$ eigenvectors of the adjacency matrix $\tilde A$ of eigenvalues $\lambda_1 < \lambda_2<\ldots <\lambda_{2n}$.
These are also eigenvectors of $\sqrt{\tilde {\bf A}}$
of eigenvalues $\gamma_l:=\sqrt{3-\lambda_l}>0$. 
This implies
\begin{align}
P_{\rm EXIT}(\infty)&:= 
\lim_{T\rightarrow \infty} P_{\rm EXIT}(T) \nn
&= \lim_{T\rightarrow \infty} \frac 1 {T} \left(\int_0^T dt  \bra{2n}    \frac 1 2\left( e^{i t \sqrt{\tilde \bA}} + e^{-i t \sqrt{\tilde \bA}}\right)   \ketbra 1  \frac 1 {2}  \left( e^{i t \sqrt{\tilde \bA}} + e^{-i t \sqrt{\tilde \bA}}\right)      \ket{2n}\right) \nn
& =   \sum_{l,l'=1}^{2n}
 \langle 2n  \ketbra{\lambda_l} 1 \rangle \!\bra 1 \lambda_{l'} \rangle \!\bra{\lambda_{l'}} 2n \rangle \lim_{T\rightarrow \infty} \frac 1{T}\int_0^T dt \cos(t \gamma_l)  \cos(t \gamma_{l'}) \nn
 & = \frac 1 2 \sum_{l=1}^{2n}
 \langle 2n  \ketbra{\lambda_l} 1 \rangle \! \bra 1 \lambda_{l} \rangle \!\bra{\lambda_{l}} 2n \rangle \nn
 & = \frac 1 2 \sum_{l=1}^{2n} |\!\bra {\lambda_l}1\rangle|^2 \;
 |\!\bra {\lambda_l}2n\rangle|^2 \nn
 & = \frac 1 2 \sum_{l=1}^{2n} |\!\bra {\lambda_l}1\rangle|^4 \nn
 & \ge \frac 1 2 \frac 1 {2n} \sum_{l=1}^{2n}|\!\bra {\lambda_l}1\rangle|^2 \nn
 \label{eq:pexitbound}
 & =  \frac 1 {4n} \;.
\end{align}
To obtain this we used $\lim_{T \rightarrow \infty}\frac 1 T \int_0^T dt \; e^{i t(\gamma_l-\gamma_{l'})}=\delta_{l,l'}$, $\lim_{T \rightarrow \infty}\frac 1 T \int_0^T dt \; e^{i t(\gamma_l+\gamma_{l'})}=0$ since the eigenvalues are positive,  $|\!\bra {\lambda_l}1\rangle|=|\!\bra {\lambda_l}2n\rangle|$ for all $l \in[2n]$ due to a reflection symmetry of the network (i.e., the adjacency matrix of the graph in the line is invariant under the transformation $l \leftrightarrow 2n-l+1$), and also  
the Cauchy Schwarz inequality for the last line as there are $2n$ different eigenvectors. This would already prove
the desired result, but we are interested in finite times $T<\infty$ and, more precisely, showing a similar bound for $T=\cO(\poly(n))$.

Using again the reflection symmetry of the network, for  $T<\infty$  the average success probability can be written as
\begin{align}
\nonumber
&  P_{\rm EXIT}(T)=\\
  &=\sum_{l,l'}
 \langle 2n  \ketbra{\lambda_l} 1 \rangle \!\bra 1 \lambda_{l'} \rangle \!\bra{\lambda_{l'}} 2n \rangle  \frac 1{T}\int_0^T dt \cos(t \gamma_l)  \cos(t \gamma_{l'}) \nn
 &=  \sum_{l,l'} |\!\bra 1 \lambda_{l} \rangle|^2
 |\!\bra 1 \lambda_{l'} \rangle|^2   \frac 1{T}\int_0^T dt \cos(t \gamma_l)  \cos(t \gamma_{l'})\nn
& = \frac 1 2 \sum_{l} |\!\bra 1 \lambda_{l} \rangle|^4 \frac 1{T}\int_0^T dt (1 + \cos(2t \gamma_l))+\sum_{l \ne l'} |\!\bra 1 \lambda_{l} \rangle|^2
 |\!\bra 1 \lambda_{l'} \rangle|^2   \frac 1{T}\int_0^T dt \cos(t \gamma_l)  \cos(t \gamma_{l'}) \nn
 & = \frac 1 2 \sum_{l} |\!\bra 1 \lambda_{l} \rangle|^4 + \frac 1 2 \sum_{l} |\!\bra 1 \lambda_{l} \rangle|^4
 \frac 1{T} \int_0^T dt \cos(2t \gamma_l) + \sum_{l \ne l'} |\!\bra 1 \lambda_{l} \rangle|^2
 |\!\bra 1 \lambda_{l'} \rangle|^2   \frac 1{T}\int_0^T dt \cos(t \gamma_l)  \cos(t \gamma_{l'}) \nn
 & = P_{\rm EXIT}(\infty) + \frac 1 2 \sum_{l} |\!\bra 1 \lambda_{l} \rangle|^4
 \frac 1{T} \int_0^T dt \cos(2t \gamma_l) + \sum_{l \ne l'} |\!\bra 1 \lambda_{l} \rangle|^2
 |\!\bra 1 \lambda_{l'} \rangle|^2   \frac 1{T}\int_0^T dt \cos(t \gamma_l)  \cos(t \gamma_{l'})\;,
\end{align}
where we used the identity $\cos^2(\alpha)=\frac 1 2 (1+\cos(2\alpha))$.
Then, the correction to $P_{\rm EXIT}(\infty)$ for $T<\infty$ is such that
\begin{align}
\nonumber
    |P_{\rm EXIT}(T) - P_{\rm EXIT}(\infty)| =& \left| \sum_{l \ne l'} |\langle 1  \ket{\lambda_l}|^2 | \! \bra 1 \lambda_{l'}\rangle|^2  \frac 1{T}\int_0^T dt \cos(t \gamma_l)  \cos(t \gamma_{l'}) + \right . \\
    &\left. \frac 1 2\sum_l |\!\bra 1 \lambda_{l} \rangle|^4 \frac 1 {T} \int_0^T dt \;  \cos(2t \gamma_l)\right| \;. 
\end{align}
Note that, for $l \ne l'$,
\begin{align}
 \left| \frac 1{T}\int_0^T dt \cos(t \gamma_l)  \cos(t \gamma_{l'})   \right| &= 
 \frac 1 {T}\left| \frac 1{2} \left(  \frac{\sin(T(\gamma_l + \gamma_{l'}))}{\gamma_l + \gamma_{l'}} + \frac{\sin(T(\gamma_l - \gamma_{l'}))}{\gamma_l - \gamma_{l'}}  \right)\right| \\
& \le \frac 1 {T \Delta}\;,
\end{align}
where $\Delta$ is the minimum spectral gap between any pair of eigenvalues $\gamma_l$, that is, $\Delta: =\min_{l}|\gamma_{l+1} - \gamma_{l}|$. (Also $|\gamma_l+\gamma_{l'}|> \Delta$ if $l \ne l'$.)
Also, if we order the eigenvalues so that $\gamma_{l+1} >\gamma_l$ for all $l=[2n-1]$, for $l \ge 2$ we have $\gamma_l \ge \Delta$ and
\begin{align}
   \left|\frac 1 {T} \int_0^T dt  \cos(2t \gamma_l) \right| = \frac 1 {T} \left| \frac {\sin(2 t \gamma_l)} {2\gamma_l} \right| &\le \frac 1{2 T \Delta}.
\end{align}
For $l=1$, the eigenvalue $\gamma_l$ is exponentially small in $n$ and the corresponding average can be $\cO(1)$. Combining these equations and using the reflection symmetry  of the network we obtain
\begin{align}
    |P_{\rm EXIT}(T) - P_{\rm EXIT}(\infty)| & \le \frac 1 {T \Delta} \sum_{l \ne l'} |\langle  \lambda_l\ket 1|^2 |\langle  \lambda_{l'}\ket 1|^2 + \frac 1 {4T \Delta} \sum_{l \ge 2} |\langle  \lambda_{l}\ket 1|^4 + \frac 1 2|\langle  \lambda_{1}\ket 1|^4 \nn
    & \le \frac 1 {T \Delta} \sum_{l,l'} |\langle  \lambda_l\ket 1|^2 |\langle  \lambda_{l'}\ket 1|^2 +
    \frac 1 2 |\!\bra 1 \lambda_{1} \rangle|^4 \nn
    & \le \frac 1 {T \Delta} + \frac 1 2 |\!\bra 1 \lambda_{1} \rangle|^4 \;,
\end{align}
where we also used $\sum_l |\!\bra 1 \lambda_{l} \rangle|^2=1$. We need to show that both these terms are small.

In Ref.~\cite{childs2003exponential} it has been shown that the eigenvalues $\lambda_l$'s are separated by spectral gaps bounded by $\Delta'=\Omega(1/n^3)$. The same follows for the $\gamma_l$'s. More precisely, suppose that the two closest eigenvalues of $\tilde{\bf A}$ are $3-\lambda \le 6$ and $3-(\lambda + \Delta ' )$.
Then, the two closest eigenvalues of $\sqrt{\tilde{\bf A}}$ are separated as
\begin{align}
\Delta &= \sqrt{3-\lambda} - \sqrt{3-(\lambda + \Delta ')}\nn
&= \sqrt{3-\lambda} \left(1-\sqrt{1 - \frac{\Delta'}{3-\lambda}}\right) \nn
&\ge \sqrt{3-\lambda} \left(1-  \left( 1- \frac{\Delta'}{3-\lambda}\right)\right) \nn
& \ge \frac{\Delta'}{\sqrt{3-\lambda}} \nn
& \ge \frac{\Delta'}{\sqrt{6}} \;.
\end{align}
Then, since $\Delta'= \Omega(1/n^3)$, we have $\Delta = \Omega(1/n^3)$. This implies 
that we can choose $T=\cO(n^4)$  so that, for example,
\begin{align}
    \frac 1 { T \Delta} \le \frac 1 {8n} \;.
\end{align}
It is also possible to show that $|\!\bra 1 \lambda_{1} \rangle|^4$ is exponentially small in $n$. 
Consider the vector
\begin{align}
 \ket{v_1}:=   \sum_{l=1}^n \frac 1 {\sqrt {2^{2+n-l}}}\ket l +  \sum_{l=n+1}^{2n} \frac 1 {\sqrt {2^{1+l-n}}}\ket l \;.
\end{align}
The length of this vector is $\sqrt{\sum_{l=1}^n (1/2^{1+n-l})}=\sqrt{1-1/2^n}$ being exponentially close to 1. In addition,
\begin{align}
\nonumber
  &\tilde{\bf A}  \ket{ v_1} = \\
  \nonumber
  &3 \ket{ v_1} -\sqrt 2 \left(\sum_{l=1}^{n-1}\frac 1 {\sqrt {2^{2+n-l}}} \ket{l+1}+ \sum_{l=2}^{n} \frac 1 {\sqrt {2^{2+n-l}}} \ket{l-1} +\sum_{l=n+1}^{2n-1} \frac 1 {\sqrt {2^{1+l-n}}}\ket{l+1} +\sum_{l=n+2}^{2n}\frac 1 {\sqrt {2^{1+l-n}}}\ket{l-1} \right) \\& \quad \quad -2 \frac 1 2 (\ket{n+1}+\ket{n}) \nn
  \nonumber
  & = 3 \ket{ v_1}-\sqrt 2 \left(\sum_{l=2}^{n}\frac 1 {\sqrt {2^{3+n-l}}} \ket{l}+ \sum_{l=1}^{n-1} \frac 1 {\sqrt {2^{1+n-l}}} \ket{l} +\sum_{l=n+2}^{2n} \frac 1 {\sqrt {2^{l-n}}}\ket{l} +\sum_{l=n+1}^{2n-1}\frac 1 {\sqrt {2^{2+l-n}}}\ket{l} \right)  \\& \quad \quad -2 \frac 1 2 (\ket{n+1}+\ket{n}) \nn
  \nonumber
  & = 3 \ket{ v_1} -3 \sum_{l=2}^{n-1} \frac 1 {\sqrt {2^{2+n-l}}} \ket{l} - \frac 1 2 \ket n-\frac 1 {\sqrt{2^{n-1}}} \ket 1 - 3 \sum_{l=n+2}^{2n-1} \frac 1 {\sqrt{2^{1+l-n}}}\ket l -\frac 1{\sqrt{2^{n-1}}} \ket{2n} -\frac 1 2 \ket{n+1} \\& \quad \quad -2 \frac 1 2 (\ket{n+1}+\ket{n}) \nn
  & = 3  \ket{ v_1} - 3  \ket{ v_1} + \frac 1{2\sqrt{2^{n-1}}}\ket 1 +  \frac 1{2\sqrt{2^{n-1}}}\ket {2n} \nn
  & = \frac 1{2\sqrt{2^{n-1}}}\ket 1 +  \frac 1{2\sqrt{2^{n-1}}}\ket {2n}.
\end{align}
This implies, $\| \tilde{\bf A} (\ket{ v_1}/\|\ket{v_1}\|)\|=\cO(1/\sqrt{2^n})$
and that $\ket{ v_1}/\|\ket{v_1}\|$ is exponentially close to the normalized eigenvector of lowest eigenvalue as the spectral gaps are at least $\Delta =\Omega(1/n^3)$; that is $\|(\ket{ v_1}/\|\ket{v_1}\|)-\ket{ \lambda_1}\|=\cO(1/\exp(n))$. Since the first entry of  $\ket{ v_1}/\|\ket{v_1}\|$ is $\cO(1/\exp(n))$,
these results imply $|\!\bra 1 \ket{\lambda_1}\!|=\cO(1/\exp(n))$.
Combining these bounds, and for the choice of $T$ above, we obtain
\begin{align}
|P_{\rm EXIT}(T) -P_{\rm EXIT}(\infty)| \le \frac 1 {8n} + \cO(1/\exp(n)) \;.
\end{align}
Together with Eq.~\eqref{eq:pexitbound}, this gives
\begin{align}
P_{\rm EXIT}(T) \ge \frac 1 {8n} - \cO(1/\exp(n)) \;,
\end{align}
which is the desired result: we can choose $T=\cO(n^4)$
so that the average probability $P_{\rm EXIT}(T)$ is $\Omega(1/n)$, implying the existence of $t \in [0,T]$ with the desired property $|\dot x_N(t)|=\Omega(1/\poly(n))$.
\qed

We note that numerical solutions show a stronger result, where $t \propto n$ suffices. See Fig.~\ref{fig:numsim} for an example.

\begin{figure}\centering
\includegraphics[scale=.6]{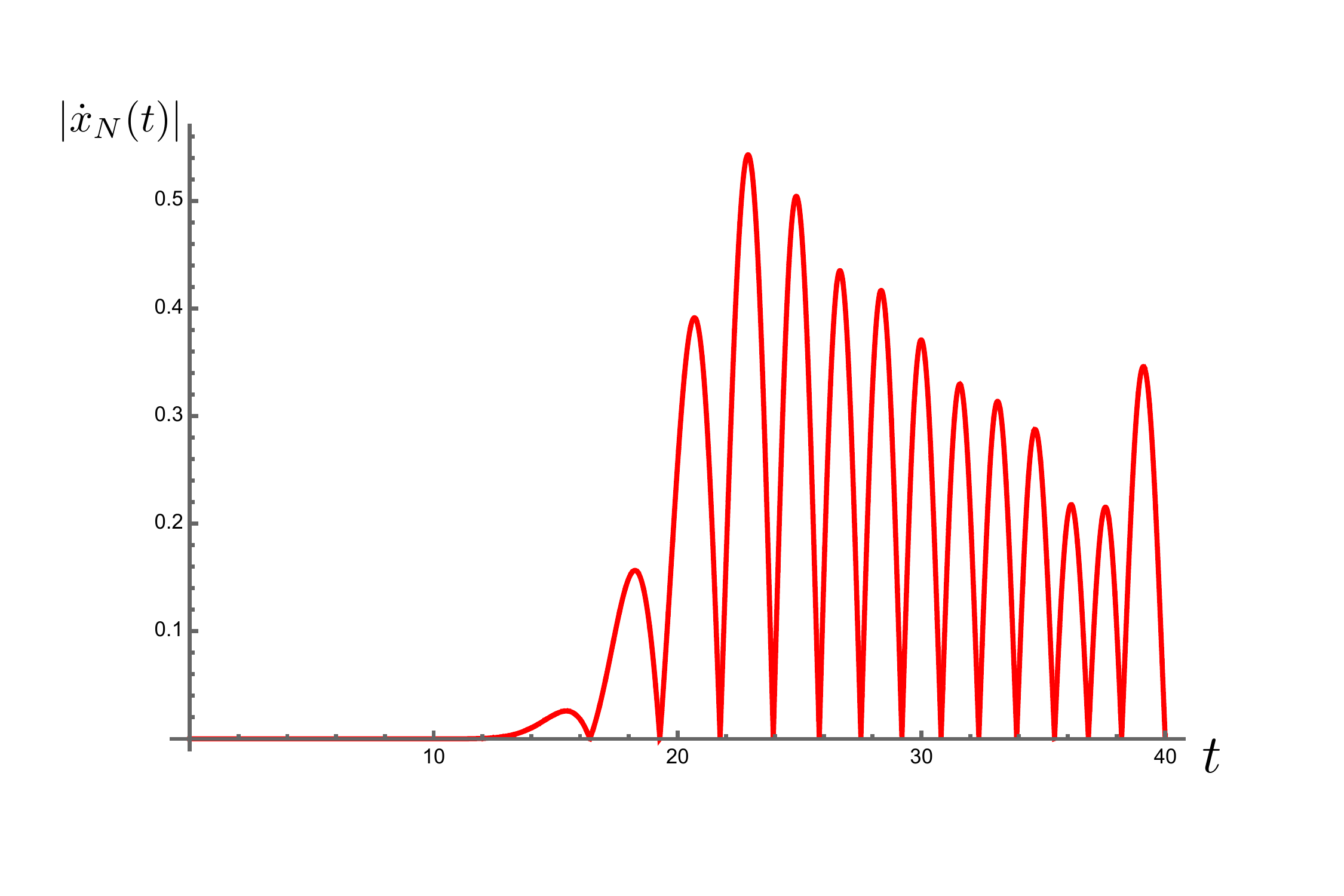}\vspace{-3em}
\caption{Numerical simulation of the network in Fig.~\ref{fig:gluedtrees} for $n=20$ ($N=2^{21}-2$). Initially $\vec x(0)=(0,0,\ldots,0)^T$ and $\dot{\vec x}(0)=(1,0,\ldots,0)^T$. At time $t \approx 2n$, $|\dot x_N(t)|$ becomes significant. Similar behavior is observed for larger values of $n$.\label{fig:numsim}}
\end{figure}

\section{BQP-completeness and proof of Thm.~\ref{thm:bqpcomplete}}
\label{app:BQPcompletenessproof}

In this section our goal is to show Thm.~\ref{thm:bqpcomplete}, which establishes that Problem~\ref{prob:nonoracular} is \BQP-complete. A problem is \BQP-complete if it is in \BQP\ and if it is \BQP-hard, which means that every problem in \BQP\ can be reduced to it by classical polynomial-time reductions. Since we have already established that Problem~\ref{prob:nonoracular} is in \BQP\ (a consequence of Thm.~\ref{thm:main}), we only need to show it is \BQP-hard.

To show our problem is \BQP-hard, we start the reduction from the following \BQP-complete problem: 
\begin{problem}\label{prob:bqpcomplete}
Given a quantum circuit on $q$ qubits with $L=\poly(q)$ gates over the gate set of Hadamard, Pauli X, and Toffoli acting on the initial state $\ket{0}^{\otimes q}$, decide if the output state has overlap at least $1-1/\exp(\sqrt q)$ with $\ket{0}^{\otimes q}$ or has overlap at most $1/\exp(\sqrt q)$ with $\ket{0}^{\otimes q}$, promised that one of these is the case.   
\end{problem}

This problem is easily seen to be in \BQP. It is \BQP-hard by reduction from the standard \BQP-complete problem, which has a similar input, but the goal is to decide a different property of the output state: Upon measuring the first qubit of the output state, we have to decide if it is $1$ with probability at least $2/3$ or $1$ with probability at most $1/3$, promised that one of these is the case. By known results, we can assume our gate set is real and contains only Hadamard and Toffoli, since this is a universal gate set for real quantum computation, which is computationally as powerful as complex quantum computation~\cite{shi2002both,aharonov2003simple}. We also add the single-qubit Pauli $X$ gate so that we can create the state $\ket 1$ from $\ket 0$, as we want the computation to start with all qubits in state $\ket 0$. Also, the $X$ and Toffoli gates generate swap gates, and we can then assume that all Hadamards act on the same qubit. This property will be useful to simplify the presentation of the proof, but it is not necessary, and one could in principle allow Hadamards to act on any qubit.

Given such a circuit, we can amplify the constants $2/3$ and $1/3$ to $1-1/\exp(q)$ and $1/\exp(q)$ by running $\cO(q)$ copies of the circuit in parallel and taking the majority vote of the first output qubits of all circuits. Now we have a new circuit on $\cO(q^2)$ qubits whose first qubit is almost certainly (with probability $1-1/\exp(q)$) $\ket{1}$ when the answer is yes and almost certainly $\ket{0}$ when the answer is no. We then use an additional qubit and copy this answer to that qubit. Since the qubit we are copying is exponentially close to being $\ket 0$ or $\ket 1$, let us assume it is one of these, and this will only introduce an error of $1/\exp(q)$. We then run the circuit's inverse on the remaining qubits to restore them to the all zero state. The resulting output state is now exponentially close to all zeros if the additional qubit was also zero, which happens when the input was a no instance. If it was a yes instance, the additional qubit we added will be exponentially close to $\ket{1}$, and hence the overall state will have almost no amplitude on the all zeros state.

Thus deciding if the output state of a $\cO(q^2)$ qubit circuit of size $\poly(q)$ is $1-1/\exp(q)$ close to the all zeros state or has at most $1/\exp(q)$ overlap on the all zeros state is \BQP-complete. This yields the stated result by renaming $q$ to $q^2$. (Our proof below does not require the closeness to be exponentially small, and it would suffice to have inverse polynomial closeness, as long as this polynomial was smaller than all the other polynomials appearing in the problem.)

\subsection{From a circuit to a network of oscillators}

We now show that our quantum algorithm for simulating coupled classical oscillators allows us to solve the \BQP-complete problem above. We begin with the given quantum circuit on $q$ qubits with $L=\poly(q)$ gates. If the gates are $U_1$ to $U_L$, the output of this circuit is $U_L \ldots U_1 \ket 0^{\otimes q}$. Our goal is to decide if this state is essentially $\ket 0^{\otimes q}$ or has essentially no overlap with $\ket 0^{\otimes q}$.
The unitaries $U_1$ to $U_L$ are either the single-qubit gates Hadamard $H$ or Pauli $X$ (tensored with identity on all other qubits) or the three-qubit Toffoli gate ${\rm Toff}$ (tensored with identity on all other qubits). These gates are
\begin{align} 
H =
\begin{pmatrix}
\frac 1 {\sqrt 2} & \frac 1 {\sqrt 2}  \cr \frac 1 {\sqrt 2} & -\frac 1 {\sqrt 2}
\end{pmatrix},
\qquad 
X =
\begin{pmatrix}
0 & 1  \cr 1 & 0
\end{pmatrix},
\qquad \mathrm{and} \qquad
{\rm Toff} =
\begin{pmatrix}
1 & 0 & 0 & 0 & 0 & 0 & 0 & 0 \cr 0 & 1 & 0 & 0 & 0 & 0 & 0 & 0 \cr  0 & 0 & 1 &0 & 0 & 0 & 0 & 0  \cr 0 & 0 & 0  & 1 & 0 & 0 & 0 & 0 \cr 0 & 0 & 0 & 0 & 1  & 0 & 0 & 0 \cr 0 & 0 & 0 & 0 & 0& 1   & 0 & 0 \cr 0 & 0 & 0 & 0 & 0& 0 & 0& 1   \cr 0 & 0 & 0 & 0 & 0& 0 & 1 & 0  
\end{pmatrix}
\end{align}
in the corresponding one-qubit and three-qubit subspaces.
It will be very useful in our context that these are real as the resulting Hamiltonian will have real entries, which we require because the $\kappa_{jk}$ are real. 

From this circuit we will create a network of $N$ oscillators, where $N=(L+1) 2^{q+1}$. 
It is convenient to label each oscillator $j \in [N]$ using two indexes $j \rightarrow (\sl,r)$, where $\sl \in [L+1]$, $r \in [2^{q+1}]$, and $j=(\sl-1)2^{q+1}+r$. 
We will call the oscillator with $\sl=L+1$ and $r=1$ the output oscillator. In our construction, the $N \times N$ matrix of spring constants ${\bf K}$ that describes the couplings between oscillators, which depends on the $U_{\sl}$'s,
 is 4-sparse and each entry is bounded as $\kappa_{jk}\le 4$. We will have all masses be $m_j=1$. Note that ${\bf A}=\sqrt{\bM}^{-1} {\bf F}\sqrt{\bM}^{-1}={\bf F}$ and $\vec y(t)=\sqrt{\bM} \vec x(t)=\vec x(t)$ in this case. This simplifies our analysis.

In the following, we show that for our oscillator network, there exists a time $t=\polylog (N)$ such that, for the initial conditions where $\vec x(0)=(0,\ldots,0)^T$, $\dot x_1(0)=1$, $\dot x_2(0)=-1$, and $\dot x_j(0)=0$ for all $j>2$,
the kinetic energy of the output oscillator is either exponentially close to $0$ (when the original circuit had almost no overlap with $\ket 0^{\otimes q}$), or $1/\polylog (N)$ (when the original circuit's output is very close to $\ket 0^{\otimes q}$). Since our algorithm allows us to estimate the kinetic energy of a given oscillator to additive $1/\polylog (N)$ precision, we can distinguish these two cases with complexity $\cO(\polylog (N))$. 
The initial quantum state $\sket{\psi(0)}$ is simply $\ket -$ tensored with the all zeros state, and hence easy to prepare.
This proves that Problem~\ref{problem:kineticestimation}  is \BQP-complete even with these constraints on the initial conditions, $t$, $\epsilon$, $\bM$, and ${\bf K}$, which is the desired result.

We now provide the details of this construction.
We use the bra-ket notation for specifying vectors and matrices, which is convenient for relating properties of the classical system with properties of quantum states. 
We will specify our oscillator network using the $\bf A$ matrix, instead of the $\bf K$ matrix. This real matrix  is PSD and has non-negative entries on the diagonal and non-positive entries on the off-diagonal. The standard Feynman-Kitaev \cite{Feynman1986a,Kitaev2002b} circuit-to-Hamiltonian construction gives us a Hamiltonian of the form
\begin{align}
\label{eq:F-KHamiltonian}
\sum_{\sl=1}^{L} (\ket {\sl}\! \bra{\sl+1}  +\ket{\sl+1}\!\bra {\sl} ) \otimes W_{\sl}\, ,
\end{align}
where $W_{\sl}$ is the $\sl^{\rm th}$ gate in the circuit, i.e., $U_{\sl}$. We cannot let $\bf A$ equal this Hamiltonian, since the matrix is not PSD and there are off-diagonal entries of both signs. The first issue is easily fixed by adding a multiple of the identity, but it will require some work to fix the second issue. We will adjoin an additional qubit to work in a slightly larger space. Formally, we consider the following Hilbert space:
\begin{equation}
    \mathcal{H} = \mathcal{H}_{\rm clock} \otimes \mathcal{H}_{\rm comp} \otimes \mathcal{H}_2\,  .
\end{equation}
Here, $\mathcal{H}_{\rm clock}$ is the $(L+1)$-dimensional Hilbert space that is used to describe the state of the clock (i.e., corresponding to the first register in Eq.~\eqref{eq:F-KHamiltonian}), which is used to track the progress of a simulated quantum computation on a state in $\mathcal{H}_{\rm comp} \otimes \mathcal{H}_2$ where $\mathcal{H}_{\rm comp}$ is $2^{q}$-dimensional and $\mathcal{H}_2$ is two-dimensional. The purpose of $\mathcal{H}_{\rm comp}$ is to store the state of the given quantum circuit at a time given by the clock.  The purpose of $\mathcal{H}_2$ is to address the issue on the off-diagonal entries raised above. The off-diagonal entries of ${\bf A}$ need to have the same sign, but our gate set includes a gate, $H$, with entries of both signs.  We address this by providing a resource state $\ket{-} =\frac 1 {\sqrt{2}}(\ket{0} - \ket{1})$ in the last register which we use to effectively create negative signs in a subspace although the operators will only have positive entries. Similar constructions that create Hamiltonians with non-negative entries have been used within the context of Hamiltonian complexity; see Refs.~\cite{jordan2010quantum,CGW15} for example.

We will map our circuit to a system of oscillators that satisfy the assumptions made in Problem~\ref{problem:oscillators} by choosing the couplings  to implement an ``encoded'' sequence of gates $\{W_{\sl}\}_{\sl}$ as
\begin{align}
\label{eq:barAdefBQP}
{\bf A}= 4 \one_{N} -\sum_{\sl=1}^{L} (\ket {\sl} \bra{\sl+1}  +\ket{\sl+1}\bra {\sl} ) \otimes W_{\sl},
\end{align}
where each $W_{\sl}$ corresponds to an encoded version of the $\sl^{\rm th}$ gate in the original circuit, which is either a Hadamard, $X$, or a Toffoli, and $\one_{N}$ is the identity matrix on the system of dimension $N$. Notice that $W_{\sl}$ lives in $\mathcal{H}_{\rm comp} \otimes \mathcal{H}_2$, whereas the original unitaries $U_{\sl}$ lived in $\mathcal{H}_{\rm comp}$. So we encode each gate into a larger space with 1 additional qubit.

We want all the entries in the encoded versions of Hadamard, $X$, and Toffoli to be non-negative, which will make the off-diagonal entries non-positive in ${\bf A}$ due to Eq.~\eqref{eq:barAdefBQP}.
Let's start with the Hadamard matrix. Without loss of generality, because we can swap qubits, we assume the Hadamard always acts on qubit $q$, the last qubit of $\mathcal{H}_{\rm comp}$. 
We describe our encoded Hadamard gate via the following positive-valued real-symmetric block matrix acting on the last qubit of $\mathcal{H}_{\rm comp}$ and $\mathcal{H}_{2}$:
\begin{equation}
H_\mathrm{enc} = \frac{1}{\sqrt{2}}\begin{pmatrix} \one_2 & \one_2 \\ \one_2 & X \end{pmatrix}.
\end{equation}
This can be seen to act as a logical Hadamard when acting on a state of the form $\ket{\psi} \ket{-}$ because of the following block-encoding result:
\begin{align}
(\one_{2} \otimes \bra{-})H_\mathrm{enc}(\one_{2} \otimes \ket{-}) &=
     \frac{1}{\sqrt{2}}\left((\one_{2} \otimes \bra{-})\begin{pmatrix} \one_2 & \one_2 \\ \one_2 & X \end{pmatrix}(\one_{2} \otimes \ket{-}) \right) \nn
     &= \frac{1}{\sqrt{2}} \begin{pmatrix} 1 & 1 \\ 1 & -1 \end{pmatrix} \nn
     \label{eq:blockEncode}
     &=H.
\end{align}
Note here that unlike a conventional block encoding, the encoded Hadamard gate $H_\mathrm{enc}$ is not a unitary operation. However it block encodes a unitary operation within the logical subspace where the Hilbert space $\mathcal{H}_2$ contains the state $\ket{-}$, which is sufficient for our purposes as it represents the Hadamard gate as a real-symmetric matrix with positive entries. 

The Toffoli and $X$ gates have non-negative entries and so we define the encoded Toffoli and $X$ gates on $\mathcal{H}_{\rm comp} \otimes \mathcal{H}_2$ as simply $U_l \otimes \one_2$. The action of this matrix on the first register ($\mathcal{H}_{\rm comp}$) is a Toffoli gate or $X$ gate, independent of the second register ($\mathcal{H}_2$).
Finally, as in the standard circuit-to-Hamiltonian construction, the initial state that we use is $\ket{\sl=1}\otimes\ket{0}^{\otimes q}$ on the first two spaces, and $\ket -$ on the last space, as discussed.
As mentioned above, the $W_{\sl}$'s, which are not necessarily unitary on the entire Hilbert space (but act as unitaries on the subspace where the last register is $\ket -$), have non-negative entries $\in \{0,1/\sqrt 2,1\}$ and are of sparsity at most 2.

The matrix ${\bf F}={\bf A}$ contains all the desired properties
to describe $N$ coupled oscillators. Its off-diagonal entries are $a_{jk}\in \{0,-1,-1/\sqrt 2 \}$, for $j,k \in[N]$, $j \ne k$. This agrees with the assumptions that the off-diagonal elements in $\mathbf{F}$ are negative given in Sec.~\ref{sec:quantumevolution}.  All diagonal entries are $a_{jj}=4$, which similarly agrees with the requirement in Sec.~\ref{sec:quantumevolution} that the diagonal elements of $\mathbf{F}$ are positive. The spring constants of the system are such that $\kappa_{jk}=-a_{jk}$ for $j \ne k$ and then $\kappa_{jk}\in \{0,1/\sqrt 2,1\}$.
Also, $\kappa_{jj}=a_{jj}-\sum_{k \ne j}\kappa_{jk}=4-\sum_{k \ne j}a_{jk}$. The matrix ${\bf K}$ has a non-zero diagonal entry and it can have at most $3$ nonzero off-diagonal entries.
For example, if $W_{\sl} = H_\mathrm{enc}$ and $W_{\sl+1} = {\rm Toff}$ or $W_{\sl+1} = X$, then the corresponding row will have at most $3$ non-zero matrix elements off the diagonal with coefficients $1/\sqrt{2}$ for two and $1$ for the remainder.  Maximizing over all such possibilities  $\sum_{k \ne j}a_{jk} \le 1+2/\sqrt 2$ in our construction, we obtain $\kappa_{jj}>0$. 
(Since the Hadamard gates act on the same qubit, we do not have two consecutive Hadamards in our circuit.)
The matrix of spring constants ${\bf K}$
is then $4$-sparse and each entry satisfies $\kappa_{jk}\le 4$.
Given $j\in[N]$, it is possible to compute all neighbors $a(j,\ell)$ of $j$, where $\ell \in [4]$, and the corresponding non-zero spring constants, in polynomial time from the polynomial-sized representation in Eq.~\eqref{eq:barAdefBQP}. This gives an efficient circuit to query the matrix ${\bf K}$ as in Sec.~\ref{sec:algorithm}.
Now that we have established that a query to the elements of the oscillator can be efficiently simulated we will now turn to showing that the classical dynamics under an exponentially large $\mathbf{A}$ can implement an arbitrary quantum computation.  

\subsection{Solving the original problem using the dynamics of oscillators}

As discussed above, the initial state we want to use is  $\ket{\sl=1} \otimes \ket 0^{\otimes q} \otimes \ket -$. As a vector, this is a vector with the first entry equal to $+1/\sqrt 2$, the second entry equal to $-1/\sqrt 2$, and the remaining entries zero. Using our problem's encoding in \eqref{eq:encodedstate}, this corresponds to the initial state where $\dot y_1(0)=1$,  $\dot y_2(0)=-1$, $\dot y_j(0)=0$ for all $j >2$, $\vec{x}(0)=0$, and $E=1$. This corresponds to all the oscillators starting at their rest positions, with the first oscillator having velocity $+1$, the second one having velocity $-1$, and the remaining oscillators being stationary.

When the initial conditions are such that $\vec x(0)=(0,\ldots,0)^T$, the solution to Newton's equations is even easier to describe, and Eq.~\eqref{eq:solution} implies
\begin{align}
\dot{\vec x}(t)&=\dot{\vec y}(t) 
= \Re{e^{i \sqrt{\bA} t} \dot{\vec y}(0)}
= \cos(\sqrt{\bA} t) \dot{\vec y}(0) \;.
\end{align}
In the rest of this section we want to understand the behavior of $\dot{\vec y}(t)$ for our initial conditions.

We start by defining an $N \times N$ matrix ${\bf A}'$, which is obtained from
${\bf A}$ in Eq.~\eqref{eq:barAdefBQP} by replacing $W_{\sl} \rightarrow U_{\sl} \otimes \one_2$ when $U_{\sl}$ is a Hadamard gate (i.e., ${\bf A}'$ is allowed to have off-diagonal entries of both signs).  We do this for simplicity to remove the dependence on the $\ket{-}$ state in the ancilla.  Specifically let $U_{\sl} \in \{H,X,{\rm Toff}\}$ and
\begin{align}
\label{eq:A'}
{\bf A}' = 4 \one_N -\sum_{\sl=1}^{L} (\ket {\sl} \bra{\sl+1}  +\ket{\sl+1}\bra {\sl} ) \otimes U_{\sl} \otimes \one_2 \;.
\end{align}
The block-encoding property of Eq.~\eqref{eq:blockEncode} implies
 \begin{align}
\dot{\vec y}(t)&=\cos(\sqrt{\bA}t) \dot{\vec y}(0) \nn
& =\cos(\sqrt{\bA}t)  \ket{\sl=1} \otimes \ket 0^{\otimes q} \otimes \ket -
\nn
\label{eq:poft}
&=  \cos(\sqrt{\bA'}t) \ket{\sl=1} \otimes \ket 0^{\otimes q}   \otimes \ket - \;.
\end{align}
That is, $W_{\sl}$
and $U_{\sl} \otimes \one_2$ act identically on the subspace
specified by $\ket-$. 
As a corollary, the velocities of the oscillators
are such that the last register remains in $\ket -$.
The matrix ${\bf A}'$ is also $4$-sparse
and satisfies ${\bf A}' \ge 0$, implying that $\sqrt{\bA'}$ is a well defined  $N \times N$ Hermitian matrix that is row-computable.

We define the $N \times N$ select unitary ($U_0:=\one_{2^q}$)
\begin{align}
S:= \sum_{\sl=0}^{L} \ketbra {\sl+1} \otimes U_{\sl} \ldots U_0 \otimes \one_2 \;,
\end{align}
which implements the unitary in the circuit up to $U_{\sl}$ when the first register is $\ket {\sl+1}$, i.e., a vector in $\mathbb R^{L+1}$ with entry 1 in position $l+1$ and zeroes elsewhere. Note that $S$ and $S^\dagger$ leave any vector represented as $\ket{\sl=1} \otimes \ket \psi$, $\ket \psi \in \mathbb C^{2^{q+1}}$, invariant. We also define the $N \times N$ Hermitian matrix
\begin{align}
{\bf X}:= 4 \one_N - \sum_{\sl=1}^{L} (\ket {\sl}\! \bra{\sl+1} + \ket{\sl+1}\!\bra \sl) \otimes \one_{2^{q+1}} \;.
\end{align}
Then, by inspection of Eq.~\eqref{eq:A'}, we obtain from the fact that each $U_\sl$ in our gate set is Hermitian and unitary
\begin{align}
{\bf A}'=S {\bf X} S^\dagger \;.
\end{align}
Since $U_0=\one_{2^q}$, for our initial conditions $\dot{\vec y}(0)=\ket{\sl=1} \otimes \ket 0^{\otimes q}  \otimes \ket -$, 
\begin{align}
\dot{\vec y}(t) &=   \cos(\sqrt{\bA'}t)\ket{\sl=1} \otimes \ket 0^{\otimes q}  \otimes \ket - \nn
& = S \cos(\sqrt{{\bf X}}t) S^\dagger \ket{\sl=1} \otimes \ket 0^{\otimes q}  \otimes \ket -  \nn
& = S \cos(\sqrt{{\bf X}}t) \ket{\sl=1} \otimes \ket 0^{\otimes q}  \otimes \ket - \;.
\end{align}
Since ${\bf X}$ acts trivially on $\mathcal{H}_{\rm comp} \otimes \mathcal{H}_2$,
it is useful to define 
\begin{align}
{\bf X}':= 4 \one_{L+1} -\sum_{\sl=1}^L (\ket \sl \! \bra{\sl+1} +\ket{\sl+1}\!\bra \sl) \;,
\end{align}
implying
\begin{align}
\dot{\vec y}(t) &= S \cos(\sqrt{{\bf X}}t) \ket{\sl=1} \otimes \ket 0^{\otimes q}  \otimes \ket - \nn
& = S \left(  \cos(\sqrt{{\bf X'}}t) \ket{\sl=1} \right) \otimes \ket 0^{\otimes q}  \otimes \ket - \;.
\end{align}
This compact expression for the vector of velocities
is useful to show the desired result.

In general, we can write
\begin{align}
\cos(\sqrt{{\bf X'}}t) \ket{\sl=1} = \sum_{\sl=1}^{L+1} \alpha_\sl(t) \ket \sl \;,
\end{align}
with $\alpha_\sl (t)\in \mathbb R$, because the Taylor expansion of the cosine function has only even powers of $\sqrt{\bf X'}$. Then
\begin{align}
\dot{\vec y}(t) &=  S \left(  \cos(\sqrt{{\bf X'}}t) \ket{\sl=1} \right) \otimes \ket 0^{\otimes q}  \otimes \ket - \nn
& = \sum_{\sl=1}^{L+1} \alpha_\sl(t) \ket \sl \otimes (U_{\sl-1} \ldots U_0 \ket 0^{\otimes q} ) \otimes \ket - \;.
\end{align}
Hence, when the first register in the tensor product is $\sl=L+1$, the second register is $U_L \ldots U_1 \ket 0^{\otimes q}$, and recall that our goal was to decide if this state was close to or far from $\ket{0}^{\otimes q}$.

So if we measure the output state, in the case where the circuit's output was close to $\ket{0}^{\otimes q}$, we will see $\ket{L+1}\ket{0}^{\otimes q}$ with probability $(1-\frac{1}{\exp(\sqrt q)})|\alpha_{L+1}(t)|^2 \geq \frac{1}{2}|\alpha_{L+1}(t)|^2$, and if the circuit's output had $1/\exp(\sqrt q)$ overlap with $\ket{0}^{\otimes q}$, then we'll see $\ket{L+1}\ket{0}^{\otimes q}$ with at most $1/\exp(\sqrt q)$ probability. 

So all we have to show is that there is exists a $t=\polylog (N)$ such that $|\alpha_{L+1}(t)|^2=\Omega(1/\polylog (N))=\Omega(1/\poly(q))$, which will allow us to distinguish the two cases by using our algorithm $\cO(\polylog (N))=\poly(q)$ times.

\subsection{Establishing inverse polynomial overlap}

It is not too hard to establish inverse polynomial overlap $|\alpha_{L+1}(t)|^2 = 
\Omega(1/\polylog (N))$ for $t=\polylog (N)$. (In fact, it is possible to get perfect overlap as we discuss in the next section.)
To prove this result, we use some  known results on the spectral properties of ${\bf X}'$. 
Its eigenvectors are
\begin{align}
\ket{\phi_\sl}=\sqrt{\frac 2 {L+2}} \sum_{\sl'=1}^{L+1} \sin \left(\frac{\pi \sl \sl'} {L+2} \right) \ket{\sl'} \, ,
\end{align}
and the eigenvalues are
\begin{align}
\gamma_\sl = 4- 2 \cos \left( \frac{\pi \sl}{L+2} \right) \, ,
\end{align}
where $\sl \in[L+1]$. Note that $6 > \gamma_\sl > 2$.
The $\ket{\phi_\sl}$'s are also the eigenvectors of $\sqrt{{\bf X'}}$, whose eigenvalues are $\gamma'_\sl:=\sqrt{\gamma_\sl}$ and $\sqrt 6 > \gamma'_l>\sqrt 2$. Let $\Delta_\sl=\gamma_{\sl+1}-\gamma_{\sl}$ be the spectral gaps of ${\bf X}'$ and note that $\Delta_\sl >0$ and $\Delta_\sl/\gamma_\sl>0$. The spectral gaps of $\sqrt{{\bf X'}}$ satisfy (for $l \in [L]$)
\begin{align}
    \gamma'_{\sl+1} -\gamma'_\sl &\ge \frac{\pi}{L+2} \min_{x\in(\sl\pi/(L+2),(\sl+1)\pi/(L+2))} \frac{d}{dx} \sqrt{4- 2 \cos (x)} \nn
    &= \frac{\pi}{L+2} \min_{x\in(\sl\pi/(L+2),(\sl+1)\pi/(L+2))} \frac{\sin(x)}{\sqrt{4- 2 \cos (x) }} \nn
    &\ge \frac{\pi}{L+2} \min_{x\in(\sl\pi/(L+2),(\sl+1)\pi/(L+2))} \frac{\sin(x)}{\sqrt{2}}.
\end{align}
This derivative is zero at $x=0$ and $\pi$, which is outside the region of values of $x$.
The derivative will therefore take its smallest (nonzero) values at the nearest allowed values of $x$:
\begin{equation}
    x=\frac{\pi}{L+2},\  \textrm{and} \ x=\frac{\pi(L+1)}{L+2}.
\end{equation}
We have ($L\ge 1$)
\begin{equation}
    \sin\left(\frac{\pi}{L+2}\right)=\sin\left(\frac{\pi(L+1)}{L+2}\right) \ge \frac 1 {\sqrt 2}\frac{\pi}{L+2}.
\end{equation}
Thus we have 
\begin{equation}
    \gamma'_{\sl+1} -\gamma'_\sl \ge \frac 1{ 2}\left(\frac{\pi}{L+2}\right)^2 .\label{eq:gapX'}
\end{equation}
Hence, the smallest spectral gap of $\sqrt{{\bf X'}}$, $\Delta'=\min_\sl(\gamma'_{\sl+1} -\gamma'_\sl)$, satisfies
$\Delta' = \Omega(1/L^2)$.

The amplitude $\alpha_{L+1}(t)$ that we are interested in can be written as
\begin{align}
\alpha_{L+1}(t)& =\bra {L+1}\cos( \sqrt{{\bf X'}}t) \ket 1 \nn
& =\sqrt{\frac 2 {L+2}}\bra {L+1}\cos( \sqrt{{\bf X'}}t) \sum_{\sl=1}^{L+1} \sin \left( \frac{\pi \sl}{L+2}\right) \ket{\phi_\sl} \nn
& = \frac 2 {L+2} \sum_{\sl=1}^{L+1} \cos(\gamma_\sl' t)\sin \left( \frac{\pi \sl(L+1)}{L+2}\right)\sin \left( \frac{\pi \sl}{L+2}\right) \nn
& = \frac 2 {L+2} \sum_{\sl=1}^{L+1} (-1)^{\sl-1} \sin^2 \left( \frac{\pi \sl}{L+2}\right)\cos(\gamma_\sl' t) \;.
\end{align}
Recall that the probability of measuring $\ket{L+1}$
after evolving with $\sqrt{{\bf X'}}$ for time $t$, when the initial state is $\ket 1$,
is 
\begin{align}\label{eq:alpha}
|\alpha_{L+1}(t)|^2 &= \left(\frac 2 {L+2}\right)^2\sum_{\sl,\sl'=1}^{L+1} (-1)^{\sl+\sl'} \sin^2 \left( \frac{\pi \sl}{L+2}\right)  \sin^2 \left( \frac{\pi \sl'}{L+2}\right) \nn
&\quad \times  \frac 1 4 \left( e^{i t (\gamma'_\sl+\gamma'_{\sl'})} + e^{i t (\gamma'_\sl-\gamma'_{\sl'})}+ e^{i t (-\gamma'_\sl+\gamma'_{\sl'})} + e^{i t (-\gamma'_\sl-\gamma'_{\sl'})}\right) \, ,
\end{align}
where we used that $\cos(\theta) = \frac{1}{2}(e^{i\theta} + e^{-i\theta})$.

If $\sl \ne \sl'$, Eq.~\eqref{eq:gapX'} implies $|\gamma'_\sl - \gamma'_{\sl'}| \ge \pi^2/(2(L+2)^2) \ge \Delta'$, and $|\gamma'_\sl + \gamma'_{\sl'}|\ge 2\sqrt 2$ for all $\sl,\sl'$.
Then, for any $\epsilon>0$, there exists a probability distribution $f(t)$, where $t \in\{0,1,\ldots,T\}$ and $T=\cO((L+2)^2 \log(1/\epsilon))$, such that $f(t)\ge 0$, $\sum_{t=0}^T f(t)=1$, and
\begin{align}
\left|\sum_{t=0}^T f(t) e^{i t (\gamma'_\sl - \gamma'_{\sl'})}\right|& \le \epsilon\, , \qquad \forall \; \sl \ne \sl'\, , \\
\left|\sum_{t=0}^T f(t) e^{i t (\gamma'_\sl + \gamma'_{\sl'})}\right| &\le \epsilon\, , \qquad \forall \; \sl , \sl'\, .
\end{align}
In other words, the absolute value of the Fourier transform of $f(t)$ for frequencies $\omega$ such that  $|\omega|\ge \Delta'$ is upper bounded by $\epsilon$.
One choice for $f(t)$ is the probability distribution obtained by taking $m=\cO(\log(1/\epsilon))$ samples $\{t_1,\ldots,t_m\}$ from a uniform distribution where $t_i \in \{0,1,\ldots,T'\}$, $T'=\cO(L^2)$, and outputting $t=\sum_i t_i$. Other choices can be found in Ref.~\cite{BKS09}.
This also implies
\begin{align}
\left|\sum_{t=0}^T f(t)  \sum_{\sl \ne \sl'} (-1)^{\sl+\sl'}  \sin^2 \left( \frac{\pi \sl}{L+2}\right)  \sin^2 \left( \frac{\pi \sl'}{L+2}\right)  e^{i t (\gamma'_\sl -\gamma'_{\sl'})} \right|  
& \le
 \epsilon   \sum_{\sl \ne \sl'} \sin^2 \left( \frac{\pi \sl}{L+2}\right)  \sin^2 \left( \frac{\pi \sl'}{L+2}\right)\nn
  & = \epsilon \left(\frac{(L+2)^2} 4 - \frac{3(L+2)} 8 \right) \nn
  \label{eq:randomizationerror}
 & \le \epsilon \frac{(L+2)^2} 4\;,
\end{align}
and
\begin{align}
\left|\sum_{t=0}^T f(t)  \sum_{\sl , \sl'} (-1)^{\sl+\sl'}  \sin^2 \left( \frac{\pi \sl}{L+2}\right)  \sin^2 \left( \frac{\pi \sl'}{L+2}\right)  e^{i t (\gamma'_\sl +\gamma'_{\sl'})} \right|  
& \le
 \epsilon   \sum_{\sl, \sl'} \sin^2 \left( \frac{\pi \sl}{L+2}\right)  \sin^2 \left( \frac{\pi \sl'}{L+2}\right)\nn
 \label{eq:randomizationerror2}
  & = \epsilon \frac{(L+2)^2} 4   \;.
\end{align}
Hence, when we analyze $\sum_{t=0}^T f(t) |\alpha_{L+1}(t)|^2$,
the terms that dominate the sum are those that correspond to $l=l'$ in Eq.~\eqref{eq:alpha} and where the phases are $e^{\pm it(\gamma'_l-\gamma'_{l'})}$.
That is, the above equations give
\begin{align}
\left| \sum_{t=0}^T f(t) |\alpha_{L+1}(t)|^2 -\frac 1 2 \sum_{t=0}^T f(t) \left(\frac 2 {L+2}\right)^2 \sum_{\sl=1}^{L+1}  \sin^4 \left( \frac{\pi \sl}{L+2}\right) \right| &\le  2 \left(\frac 2 {L+2}\right)^2  \epsilon \frac{(L+2)^2} 4 \nn
& \le 2\epsilon \;.
\end{align}
The second term on the left-hand side can be computed and is $\frac 3 {4(L+2)}$, which gives
\begin{align}
\label{eq:averageprobability}
\left| \sum_{t=0}^T f(t) |\alpha_{L+1}(t)|^2 -   \frac 3 {4(L+2)}  \right| \le 2 \epsilon .
\end{align}
Let us choose $\epsilon = 1/(4(L+2))$.
This implies that the average $\sum_{t=0}^T f(t)|\alpha_{L+1}(t)|^2$ is $\Omega(1/L)$ and recall that $T=\widetilde \cO(L^2)$. Then, there must exist a $t=\widetilde \cO(L^2)$ such that 
$|\alpha_{L+1}(t)|^2 = \Omega(1/L)$.

Now we are done, since we can run the algorithm for all $t=1,\ldots,T$, since there are only $\poly(L)=\polylog(N)$ values. Alternatively, we can find the desired $t$ by simulating $\sqrt{{\bf X'}}$ for all times $t=1,2,\ldots,T$ on a classical computer at cost that is also polynomial in $L$ or $\polylog(N)$. Note that although our proof does not pin down a specific $t$ for which this works, numerical simulations show that a fixed time $t=\cO(L)$ suffices for any $L$.

\subsection{Bonus: Establishing perfect overlap}

While the argument above completes the proof of \BQP-completeness, we observe that it is also possible to obtain $|\alpha_{L+1}(t)|=1$ in the proof above, which is equivalent to perfect transmission of a disturbance down a spin chain, which was solved in Refs.\ \cite{Vaia,Newton,Racah}.
The principle is to adjust the weights of the operator ${\bf X}'$ as
\begin{equation}
    {\bf X}' = \sum_{\sl=0}^{L} b_\sl \ket{\sl+1}\bra{\sl+1}  - \sum_{\sl=1}^{L} \sqrt{u_\sl} \left(\ket{\sl}\bra{\sl+1}+\ket{\sl+1}\bra{\sl}\right)\, ,
\end{equation}
where we have adjusted the first sum to be consistent with the notation in Ref.~\cite{Racah}.
The weights $b_\sl$ and $u_\sl$ are adjusted so that the eigenvalues are proportional to the squares of a sequence of integers.
Provided the matrix is persymmetric (so $b_\sl=b_{L-\sl}$ and $u_\sl=u_{L+1-\sl}$), then $e^{i\sqrt{{\bf X}'}t}\ket{1}$ gives exactly $\ket{L+1}$ for some $t$.
In particular, in Ref.\ \cite{Vaia}, ${\bf X}'$ has the eigenvalues $2k^2$ for $k=0$ to $L$, which gives the perfect transfer for $t=\pi/\sqrt{2}$.

Here we cannot use that exact result, because 
that would imply certain values for the diagonal entries of ${\bf X}'$ (and hence ${\bf A}$) for which the Hamiltonian would not correspond to a system of coupled oscillators. For our case,
it will suffice if the diagonal entries are
at least 
$2\sqrt{2}$ times the off-diagonal entries.
To obtain the larger on-diagonal entries we can use the analysis in terms of para-Racah polynominals in \cite{Racah}.
For odd $L=2j+1$, one takes (from Eq.~(2.9) of \cite{Racah}, and using $\alpha=1/2$)
\begin{align}
    b_\sl &= \begin{cases}
    \tfrac 12[a(a+j)+c(c+j)+\sl(L-\sl)] & {\rm if~} \sl\ne j,j+1 \, , \\
    a^2+\tfrac 12 j (1+a-c)(1+a+c+j)-\tfrac 12(a-c)(1+j)(a+c+j) & {\rm if~} \sl=j,j+1 \, ,
    \end{cases} \\
    u_\sl &= \begin{cases}
    \frac{\sl(L+1-\sl)(L-\sl+a+c)(\sl-1+a+c)((\sl-j-1)^2-(a-c)^2)}{4(L-2\sl)(L-2\sl+2)} & {\rm if~} \sl\ne j+1 \, , \\
    \tfrac 14 (a-c)^2 (1+j)^2 (a+c+j)^2 & {\rm if~} \sl=j+1 \, .
    \end{cases}
\end{align}
We have flipped the sign of the matrix in Ref.\ \cite{Racah} to be consistent with our usage here.
The eigenvalues as per Eq.~(3.11) of Ref.\ \cite{Racah} are (again flipping the sign from that work)
\begin{align}
    \lambda_{2s} &= (s+a)^2, \quad s=0,\ldots,j \, , \nn
    \lambda_{2s+1} &= (s+c)^2, \quad s=0,\ldots,j \, .
\end{align}
We then get the appropriate set of eigenvalues if $a-c=1/2$.
The matrix is persymmetric if $\alpha=1/2$.
We will also use $L_+=L+1$ in the notation to simplify the form of the expressions.
Taking $a=L_+/2+1/4$ and $c=L_+/2+3/4$, we obtain
\begin{align}
    4 b_\sl &= 5L_+^2/2-1/4-2(\sl-L/2)^2\, , \\
    16 u_\sl &= \sl(2L_+-\sl)(L_+^2-\sl^2) \, .
\end{align}
Exactly the same result is obtained for odd $L=2j$ using the expressions in Eq.~(4.4) of Ref.\ \cite{Racah}.
In either case it is found that the eigenvalues are $\tfrac 14(L+k+1/2)^2$ for $k=0$ to $L$.
Then one can evolve directly to $\sl=L+1$ for $t=2\pi$, so that $|\alpha_{L+1}(2\pi)|=1$.
Alternatively, one can divide the ${\bf X}'$ given here by $L^2$ so the coefficients are $\mathcal{O}(1)$, and evolve for time $2\pi L$.

For the relative values of $b_\sl$ and $u_\sl$, we need to show
\begin{equation}\label{eq:buineq}
    b_\sl \ge \sqrt 2 (\sqrt{u_\sl}+\sqrt{u_{\sl+1}}) \, ,
\end{equation}
for $\sl=0$ to $L$.
Here we use the convention that $u_0=0$, which is given by the formula above.
For $\sl=0$ we want $b_0\ge \sqrt{2u_1}$, so Eq.~\eqref{eq:buineq} can still be used with the convention $u_0=0$.
Note that both $b_\sl$ and $u_\sl$ take their maximum values in the centre, but $b_\sl$ is nonzero at the boundaries whereas $u_\sl$ is close to zero.
This means that $b_\sl/\sqrt{u_\sl}$ is smallest for $\sl=L/2$, and there the ratio is approximately $10/3$, which is significantly larger than $2\sqrt{2}$.

To show this inequality more rigorously, we can use the concavity of $u_\sl$ as a function of $\sl$ and the concavity of the square root function to give
\begin{equation}
    \sqrt{8 u_{\sl+1/2}} \ge \sqrt 2 (\sqrt{u_\sl}+\sqrt{u_{\sl+1}}) \, ,
\end{equation}
where $u_{\sl+1/2}$ is using the same function of $\sl$ (even though it is only meaningful to give coefficients for integer $\sl$).
We find that
\begin{equation}
   b_{\sl}^2-8u_{\sl+1/2} = 1 + 16( L_+^4 - L_+^2 )[1 + (1 - \beta) \beta] + 
16 L_+^4 \left(3 - \beta \{9 - \beta [5 + 4 (2 - \beta) \beta]\}\right) \, ,
\end{equation}
with $\beta = (\sl+1/2)/{L_+}$.
Because we consider values of $\sl\in[0,L]$, we need only consider $\beta\in[0,1]$.
The expression $\left(3 - \beta \{9 - \beta [5 + 4 (2 - \beta) \beta]\}\right)$ is positive; it has its minimum of $1/2$ at $\beta=1/2$.
The expression $L_+^4 - L_+^2$ is non-negative, and $1 + (1 - \beta) \beta$ is positive for $\beta\in[0,1]$.
As a result
\begin{align}
    b_{\sl}^2&\ge 8u_{\sl+1/2} \nn
  \implies \quad  b_{\sl} & \ge \sqrt{8u_{\sl+1/2}} \ge \sqrt 2 (\sqrt{u_\sl}+\sqrt{u_{\sl+1}}) \, ,
\end{align}
as required.

\section{Phase estimation approach and proof of Thm.~\ref{thm:generalized}}
\label{app:QPE}

We provide proof of Thm.~\ref{thm:generalized}, which provides a method for simulating the dynamics a broader family of classical systems under the harmonic approximation than the previous method described in Thm.~\ref{problem:oscillators}, at the price of reduced scaling with the desired error tolerance $\epsilon$. 

\begin{proof}[Proof of Thm.~\ref{thm:generalized}]
We describe the simulation of ${\bf H}=-X \otimes \sqrt{\bA}$, where $X$ is the single-qubit Pauli bit flip operator and $\sqrt{\bA}$ is the principal square root of ${\bf A}$, using a standard quantum phase estimation (QPE) approach. 
In this approach, we run QPE
with a unitary that is a walk operator built from a block encoding of $\mathbf{H^{(2)}}:=-X \otimes \mathbf{A}$, i.e., a unitary that contains $\mathbf{H^{(2)}}/\lam$ in a block, where $\lam>0$ is due to normalization reasons. This QPE provides estimates of the eigenvalues of the walk operator, which can be converted into estimates of the eigenvalues of $\mathbf{H^{(2)}}$
and ultimately of ${\bf H}$. Once these estimates are obtained, 
simulation of ${\bf H}$ is simply applying a phase that is a product of the eigenvalue estimate and $t$. Last, we reverse QPE and other calculations to uncompute the estimated eigenvalues.

In the case where we have oracle access to ${\bf K}$, then ${\bf A}={\bf B}{\bf B}^\dagger$ and we can obtain a block encoding of ${\bf A}$ -- and hence  of $\mathbf{H^{(2)}}$ -- from the block encodings of ${\bf B}$ and ${\bf B}^\dagger$. These were given in App.~\ref{app:accessmodel}.
This approach would imply $\lam=2\aleph d$.
In the case of generalized coordinates described in Sec.~\ref{sec:generalization}, where we assume oracle access to $d$-sparse ${\bf A}$ but not ${\bf K}$, then $\mathbf{H^{(2)}}$ can be block encoded using standard methods with $\Lambda=\cO(\|{\bf H^{(2)}}\|_{\rm max} d)=\mathcal{O}(\|{\bf A}\|_{\rm max} d$)~\cite{LC19}.  
Let $\ket{\lambda_j}$ denote the eigenvectors of ${\bf A}$
of eigenvalue $\lambda_j \ge 0$. Then, we can write 
$\ket {\eta_X} \ket{\lambda_j}=\ket {\eta_X,  \lambda_j}$ for the eigenvectors of 
$\mathbf{H^{(2)}}$, where $\eta \in \{0,1\}$, and $\ket {0_X}=\ket -$ and $\ket {1_X}=\ket +$ are the eigenvectors of $X$. The corresponding eigenvalues of $\mathbf{H^{(2)}}$ are $\gamma_{\eta,j}:=(-1)^{\eta} \lambda_j$.
Our first step is to estimate the eigenvalue $\gamma_{\eta,j}$ within fixed error, and then propagate that error into the maximum error that can be observed in $(-1)^{\eta} \sqrt{\lambda_j}$, which is the corresponding eigenvalue of ${\bf H}$,
using an arithmetic circuit on the outputs of the QPE routine.
Phase estimation can be used to provide a confidence interval $S_{\eta,j}$ for the eigenvalue estimates of $\mathbf{H^{(2)}}$, which we call $x$.
That is, for all $x\in S_{\eta,j}$ that are estimates of $\gamma_{\eta,j}$, and for a given $\epsilon_{\rm PE}$, we define $\epsilon_{\eta,j}(x) := x-\gamma_{\eta,j}$ and $S_{\eta,j}:=\{x : \epsilon_{\eta,j}(x) \in [-\epsilon_{\rm PE},\epsilon_{\rm PE}]\}$.
To describe the contribution to estimates outside the confidence interval we use a state $\ket{\phi_{\eta,j}}$ of unit norm and an amplitude $\sqrt{\delta_{\eta,j}}$, where $\delta_{\eta,j}>0$ depends on $\eta$ and $j$, and $1-{\delta_{\eta,j}}$ is the confidence level when the input state is $\ket {\eta_X,\lambda_j}$. 
Eventually, we will set $\delta_{\eta,j} \le \delta_{\rm PE}$ for all $\eta$ and $j$, where $\delta_{\rm PE}>0$ is determined below.
Then, when implementing QPE
on input eigenstate $\ket {\eta_X,\lambda_j}$
and using a unitary that provides eigenvalue estimates of $\mathbf{H^{(2)}}$ (e.g., a walk operator),
the state is approximately transformed as
\begin{equation}
    \ket {\eta_X,\lambda_j} \mapsto \ket {\eta_X,\lambda_j}  \left(\sqrt{1-\delta_{\eta,j}}  \sum_{x\in S_{\eta,j}} b_{x} \ket{x} + \sqrt{\delta_{\eta,j}} \ket{\phi_{\eta,j}}\right) \, ,
\end{equation}
for some unit vector $\vec b$ that gives the probability distribution $|b_x|^2$ for phase estimates within the confidence interval (the states $\ket x$ can be basis states that encode the estimates of $\gamma_{\eta,j}$). 
When performing QPE using Kaiser windows \cite{Sanders2020},
the number of invocations of the walk operator (i.e., the query complexity) built from the block encoding of $\mathbf{H^{(2)}}$ is then~\cite{Kaiser}
\begin{align}
\label{eq:QPEcomplexity}
    \cO( \lam \log(1/\delta_{\rm PE})/\epsilon_{\rm PE}) = \cO(\|{\mathbf{A}}\|_{\max}d \log(1/\delta_{\rm PE})/\epsilon_{\rm PE}) \;.
\end{align}

The walk operator we use in QPE combines
the block encoding of $\mathbf{H^{(2)}}$ with other two-qubit gates, including a reflection on some ancilla qubits \cite{BerryNPJ18}.  That is, each use of the walk operator
requires $\cO(1)$ uses of the oracle to access ${\bf A}$.
The eigenvalues of the walk operator are then $\mp e^{\pm i\, \arcsin(\lambda_j/\lam)}$.
The actual eigenvalues of $\mathbf{H^{(2)}}$, which are $\gamma_{\eta,j}$, can be estimated by QPE using the walk operator, then taking the sine of the result to give an estimate of $\gamma_{\eta,j}/\lam$. Note that $|\gamma_{\eta,j}|/\lam \le 1$ for all $\eta \in \{0,1\}$ and $j \in[N]$. Aiming for $\epsilon_{\rm PE}/\Lambda$ uncertainty in the phase estimation of $\arcsin(\lambda_j/\Lambda)$ yields the following estimate in $\gamma_{\eta,j}/\lam$:
\begin{equation}
    \left|\gamma_{\eta,j}/\lam - \sin(\arcsin(\gamma_{\eta,j}/\lam) + \epsilon_{\rm PE}/\lam)\right| \le \epsilon_{\rm PE}/\lam \, .
\end{equation}
Then, in QPE we are first obtaining an estimate $x=\gamma_{\eta,j}+\epsilon_{\eta,j}(x)$, 
where the error $\epsilon_{\eta,j}(x)$ satisfies $|\epsilon_{\eta,j}(x)|\le \epsilon_{\rm PE}$ within a $(1-\delta_{\rm PE})$-confidence level. We want to convert this estimate into an estimate of an eigenvalue of ${\bf H}$, i.e., an estimate of $(-1)^{\eta} \sqrt{\lambda_j}={\rm sign}( \gamma_{\eta,j})\sqrt{|\gamma_{\eta,j}|}$, which is obtained by computing ${\rm sign}(x) \sqrt{|x|}$.
This gives the error
\begin{align}
\nonumber
    \Delta_{\eta,j}(x) &:= {\rm sign}(x) \sqrt{|x|}-
    {\rm sign}(\gamma_{\eta,j})\sqrt{|\gamma_{\eta,j}|} \\
    &={\rm sign}(\gamma_{\eta,j}+\epsilon_{\eta,j}(x))\sqrt{|\gamma_{\eta,j}+\epsilon_{\eta,j}(x)|} - {\rm sign}(\gamma_{\eta,j})\sqrt{|\gamma_{\eta,j}|} \, ,
\end{align}
which we seek to bound.
In the case where $\gamma_{\eta,j}\ge 0$ (i.e., $\eta=0$ and $\gamma_{\eta,j}=\lambda_j$), we have
\begin{equation}
| \Delta_{\eta,j}(x)|=    \left| {\rm sign}(\lambda_j+\epsilon_{\eta,j}(x))\sqrt{| \lambda_j+\epsilon_{\eta,j}(x)|} - \sqrt{\lambda_j} \right| \, ,
\end{equation}
and in the case where $\gamma_{\eta,j}\le 0$ (i.e., $\eta=1$ and $\gamma_{\eta,j}=-\lambda_j$), we have
\begin{equation}
  | \Delta_{\eta,j}(x)|=  \left| {\rm sign}(-\lambda_j+\epsilon_{\eta,j}(x))\sqrt{|-\lambda_j+\epsilon_{\eta,j}(x)|} +\sqrt{\lambda_j} \right| \, .
\end{equation}
If we replace $\epsilon_{\eta,j}(x)$ with $-\epsilon_{\eta,j}(x)$ this becomes
\begin{equation}
    \left| {\rm sign}(\lambda_j+\epsilon_{\eta,j}(x))\sqrt{|\lambda_j+\epsilon_{\eta,j}(x)|} -\sqrt{\lambda_j} \right| \, ,
\end{equation}
and so bounding the expression for the case $\eta=0$ will be sufficient to account for all cases.
To find the bound for this case we will show that for any $a\in \mathbb{R}$
\begin{equation}
    \left| {\rm sign}(1+a)\sqrt{|1+a|} - 1 \right| \le \sqrt{2} \min(|a|,\sqrt{|a|})\, .
\end{equation}
There are four cases to consider to prove this expression.
\begin{enumerate}
    \item For $a\ge 1$ we can prove the inequality by using
    \begin{align}
    1+a \le 1+a+2\sqrt{a} \; &\implies \; \sqrt{1+a} \le 1+\sqrt{a} \nn
    & \implies  \sqrt{1+a} -1 \le \sqrt{a} \nn
    & \implies {\rm sign}(1+a)\sqrt{|1+a|} - 1 \le \sqrt{|a|}
    \, .
\end{align}
Since $\sqrt{|a|}\le |a|$ for $a\ge 1$, this proves the inequality.
\item For $a\in[0,1]$, we can use
\begin{align}
\sqrt{ 1 +a} \le 1+a \; &\implies \; 
    \sqrt{ |1 +a|} - 1 \le |a| \nn
    & \implies {\rm sign}(1+a)\sqrt{|1+a|} - 1 \le |a|
    \, .
\end{align}
Since $|a| \le \sqrt{|a|}$ for $a\in[0,1]$, this proves the inequality.
\item For $a\in [-1,0]$,
\begin{align}
    \sqrt{ |1 +a|} \ge |1+a| \; &\implies \; 
    \sqrt{ |1 +a|} - 1 \ge a \nn
    & \implies \; 1-\sqrt{ |1 +a|} \le -a \nn
    & \implies \left|  {\rm sign}(1+a)\sqrt{|1+a|} - 1 \right| \le |a| \, .
\end{align}
Again for $a\in [-1,0]$ we can use $|a| \le \sqrt{|a|}$.
This time we need to take the absolute value because $\sqrt{ |1 +a|}\le 1$.
\item For $a\le -1$ we can use
\begin{align}
0 \le (|a|-2)^2 \; & \implies \; 4(|a|-1) \le a^2 \nn
& \implies \; 2\sqrt{|a|-1} \le |a| \nn
& \implies \; (|a|-1) + 2\sqrt{|a|-1} + 1 \le 2|a| \nn
& \implies \; \sqrt{|a|-1} + 1 \le \sqrt{2|a|} 
\nn
& \implies \; \left| {\rm sign}(1+a)\sqrt{|1+a|} - 1 \right| \le \sqrt{2|a|} 
\, .
\end{align}
That also implies that it is upper bounded by $\sqrt{2}|a|$.
\end{enumerate}
The result gives us the bound on the error in the signed square root as follows:
\begin{equation}
    \Delta_{\eta,j}:=\max_{x:|x-\gamma_{\eta,j}|\le \epsilon_{\rm PE}} |\Delta_{\eta,j}(x)|\le \min\left( \epsilon_{\rm PE}\sqrt{2/\lambda_j} , \sqrt{2\epsilon_{\rm PE}} \right) \, .
\end{equation}
Hence, from the estimate $x$ of $\gamma_{\eta,j}$ within error at most $\epsilon_{\rm PE}$, we can compute an estimate of $(-1)^{\eta} \sqrt{\lambda_j}$, the eigenvalue of ${\bf H}$, within error at most $\Delta_{\eta,j}$ as above.

To implement $e^{i t X \otimes \sqrt{\bA}}$, we apply a phase factor proportional to the estimated eigenvalue from QPE. 
Then the state after applying the phase factor is transformed as
\begin{align}
\nonumber
    \ket{\eta_X,\lambda_j} &\mapsto \ket{\eta_X,\lambda_j} \left(\sqrt{1-\delta_{\eta,j}}  \sum_{x\in S_{\eta,j}} b_{x} e^{-i t {\rm sign}(x)\sqrt{|x|} }\ket{x} + \sqrt{\delta_{\eta,j}} \sket{\phi_{\eta,j}'}\right)\\
    & =e^{-i t (-1)^{\eta} \sqrt{\lambda_j}}\ket{\eta_X,\lambda_j} \left(\sqrt{1-\delta_{\eta,j}}  \sum_{x\in S_{\eta,j}} b_{x} e^{-i t \Delta_{\eta,j}(x) }\ket{x} + \sqrt{\delta_{\eta,j}} \sket{\phi_{\eta,j}'}\right) \nonumber\\
    &= e^{-i t (-1)^{\eta} \sqrt{\lambda_j}}\ket{\eta_X,\lambda_j}  \Biggr[\left(\sqrt{1-\delta_{\eta,j}}  \sum_{x\in S_{\eta,j}} b_{x} \ket{x} + \sqrt{\delta_{\eta,j}} \ket{\phi_{\eta,j}}\right) \nonumber\\
    &\qquad\qquad + \left(\sqrt{1-\delta_{\eta,j}}  \sum_{x\in S_{\eta,j}} b_{x} (e^{-it \Delta_{\eta,j}(x) }-1)\ket{x} + \sqrt{\delta_{\eta,j}} (\sket{\phi_{\eta,j}'}-\ket{\phi_{\eta,j}})\right) \Biggr] \, .
\end{align}
The states $\sket{\phi_{\eta,j}}$ and $\sket{\phi_{\eta,j}'}$ are states of the ancillas of unit norm, supported in the space orthogonal to $\{\ket x : x \in S_{\eta,j}\}$.
After the phase was applied, we need to invert QPE.
The inverse phase estimation and projection onto the zero state of the ancillas that were used to estimate the eigenvalues corresponds to applying
\begin{equation}
    \sbra{\tilde \lambda_{\eta,j}} = \sum_{x\in S_{\eta,j}} b_x^* \bra{x} \, ,
\end{equation}
to give
\begin{align}
&  e^{-i t {(-1)^{\eta}} \sqrt{\lambda_j}}\ket{\eta_X,\lambda_j}  \Biggr[\sqrt{1-\delta_{\eta,j}}   + \sqrt{\delta_{\eta,j}} \langle\tilde \lambda_{\eta,j} \ket{\phi_{\eta,j}} \nonumber\\
   &\qquad\qquad + 
    \sqrt{1-\delta_{\eta,j}}  \sum_{x\in S_{\eta,j}} |b_{x}|^2 (e^{-it\Delta_{\eta,j}(x) }-1) + \sqrt{\delta_{\eta,j}} \left(\langle \tilde \lambda_{\eta,j}\ket{\phi_{\eta,j}'}-\langle {\tilde \lambda}_{\eta,j} \ket{\phi_{\eta,j}}\right) \Biggr] \, .
\end{align}
The error can be bounded from the Euclidean distance between the above state and the correct state $e^{-i t {(-1)^{\eta}} \sqrt{\lambda_j}}\ket{\eta_X,\lambda_j}$. Using $\delta_{\eta,j} \le \delta_{\rm PE}$,
this distance is upper bounded by
\begin{equation}
    1-\sqrt{1-\delta_{\rm PE}} + 3\sqrt{\delta_{\rm PE}} + \sqrt{1-\delta_{\rm PE}}\langle e^{-it \Delta_{\eta,j}(x)}-1 \rangle \, ,
\end{equation}
where the expectation value  is on the probability distribution given by $|b_x|^2$.
This can be further upper bounded by
\begin{equation}
    \max_{\eta,j}  \Delta_{\eta,j}t +4\sqrt{\delta_{\rm PE}}\;.
\end{equation}
To appropriately bound the error by $\epsilon$ as required by Problem~\ref{problem:harmonicapprox}, we can take
\begin{equation}
    4\sqrt{\delta_{\rm PE}} = \epsilon/2 \, ,
\end{equation}
so $\delta_{\rm PE} = \epsilon^2/64$.
Because measurement of the phase with this confidence level has a factor $\cO(\log(1/\delta_{\rm PE}))$ in the complexity in Eq.~\eqref{eq:QPEcomplexity}, it corresponds to a factor of $\mathcal{O}(\log(1/\epsilon))$.
Then for 
\begin{align}
\nonumber
    \max_{\eta,j} \Delta_{\eta,j} & \le \max_{\eta,j} \min\left( \epsilon_{\rm PE}\sqrt{2/ \lambda_j} , \sqrt{2\epsilon_{\rm PE}} \right) \\
    & = \min\left( \epsilon_{\rm PE}\sqrt{2/\min_j \lambda_j} , \sqrt{2\epsilon_{\rm PE}} \right),
\end{align}
to satisfy $\max_{\eta,j} \Delta_{\eta,j} t \le \epsilon/2$, we can take
\begin{equation}
    \epsilon_{\rm PE} = \max \left( \frac{\epsilon\sqrt{\min_j\lambda_j}}{2t\sqrt{2}} , \frac{\epsilon^2}{8t^2} \right) .
\end{equation}
Using these expressions in Eq.~\eqref{eq:QPEcomplexity} gives us an overall query complexity
\begin{equation}
    \mathcal{O} \left(  \|{\mathbf{A}}\|_{\max} d \log(1/\epsilon) \min \left( \frac{t\sqrt{\|\mathbf{A}^{-1}\|}}{\epsilon} , \frac{t^2}{\epsilon^2} \right)\right) 
\end{equation}
for the phase estimation approach. This is the stated result.

There are further elementary two-qubit gates arising from three main areas.
\begin{enumerate}
    \item Computing ${\rm sign}(x)\sqrt{|x|}$, the estimate of the eigenvalue of ${\bf H}$, multiplying by $t$, and then applying the phase factor.
    The dominant complexity is that from computing the square root, which scales using Newton iteration and textbook multiplication  as $\widetilde{\cO}(\log^2(1/\epsilon))$.  The phase rotation then requires a linear number of controlled rotation gates based on the target, which in turn requires at most $\cO(\log(1/\epsilon))$ two-qubit gates to implement.  Thus the former cost dominates.
    \item The gates for the implementation of the block encoding will be logarithmic in $N$ and the allowable error of the block encoding.
    Because the allowable error in the block encoding needs to be $\epsilon$ divided by the number of block encodings in the phase estimation, there will be logarithmic factors in many of the parameters here.
    \item The gate complexity of preparing the control states states for phase estimation with optimal confidence intervals is at most linear in the dimension of this control register \cite{linearprep}, which is the same as the number of oracle calls.
    This is the least significant contribution to the complexity and gives no logarithmic factors.
\end{enumerate}
In quoting the complexity we give the logarithmic factor in $N$ coming from implementing the block encoding, but for simplicity use $\widetilde{\mathcal{O}}$ and do not explicitly give the logarithmic factors in other parameters.
\end{proof}

\section{Initial state preparation}
\label{app:initialstate}

Our algorithm for solving Problem~\ref{problem:oscillators} accepts as input a circuit $\cW$ that prepares an initial state of the form in Eq.~\eqref{eq:encodedstate} for $t=0$.  
In this appendix we explain how we might create such a circuit given only the ability to separately create superpositions over the initial positions and initial velocities. Our construction is related to that in Lemma~\ref{lem:blockencodingB} of App.~\ref{app:accessmodel}.

As in Sec.~\ref{sec:algorithm}, we consider a setting where the $N\times N$ matrices ${\bf M}$ and ${\bf K}$ can be queried through the use of a unitary $\cS$ that gives the positions and non-zero entries of these matrices.

\begin{lemma}
\label{lemma:stateprep}
Let ${\bf K}$ be the $N \times N$ symmetric and $d$-sparse matrix of spring constants, $\kappa_{\max} \ge \kappa_{jk} \ge 0$,  ${\bf M}\succ 0$ be the $N \times N$ diagonal matrix of masses, and $m_{\max} \ge m_j>0$, where $\kappa_{\max}$ and $m_{\max}$ are known. 
Assume we are given access to  ${\bf K}$ and ${\bf M}$ through an oracle $\cS$, and access to a unitary $\cU$ that performs the map
\begin{align*}
     \alpha\ket 0 \mapsto  \ket 0  \sket{\dot{\vec x} (0)}, \qquad
     \beta\ket 1   \mapsto \ket 1 \sket{{\vec x} (0)} \, ,
\end{align*}
where 
\begin{align*}
 \sket{\dot{\vec x}(0)}& = \sum_{j=1}^N \dot x_j(0) \ket j, 
\qquad\sket{\vec x(0)} = \sum_{j=1}^N x_j(0) \ket j \, ,
\end{align*}
are normalized states that encode the initial states of the oscillators in their amplitudes, and $\alpha$ and $\beta$ are known norms of the vectors $\dot{\vec{x}}(0)$ and $\vec x(0)$, respectively.
Then, there exists a quantum algorithm that prepares a state that is $\epsilon$-close in Euclidian norm to 
\begin{align}
\label{eq:appinitialstate}
\sket{\psi(0)}=\frac 1 {\sqrt{2E}} \begin{pmatrix}
    \sqrt{\bM} \dot{\vec x}(0) \cr i \vec \mu(0)
\end{pmatrix} \;,
\end{align}
where $E>0$ is a known constant (the energy of the classical oscillators) and $\vec \mu(0) \in \mathbb R^M$ for $M=N(N+1)/2$ is a vector whose entries are $\sqrt{\kappa_{jj}}x_j(0)$ or $\sqrt{\kappa_{jk}}(x_j(0)-x_k(0))$, $k>j$. 
The quantum circuit makes $Q_{\rm ini}=\cO(\sqrt{E_{\max}d/E})$ uses of $\cU$, $\cS$, and its inverses, in addition to
\begin{equation}
    G_{\rm ini}=\cO(\sqrt{E_{\max}d/E} \; \polylog (N E_{\rm max}/(E\epsilon)))
\end{equation}
 two-qubit gates.
 Here $E_{\max}=\frac {m_{\max}} 2 \sum_j (\dot x_j(0))^2 + \frac {\kappa_{\max}} 2 \sum_j (x_j(0))^2$ is the energy of a system of $N$ uncoupled oscillators where all masses are $m_{\max}$ and all individual spring constants are $\kappa_{\max}$, and for the same initial conditions.
\end{lemma}

\begin{proof}
For ease of implementation, we can encode the state on $2n+1$ qubits, where $n=\log(N)$.
In this space, the initial state is represented as
\begin{align}
\label{eq:initialstateapp}
    \sket{\psi(0)}=\frac 1 {\sqrt{2E}} \left[ \sum_j \sqrt{m_j}\dot x_j(0) \ket j \ket 0 + i \sum_j \sqrt{\kappa_{jj}} x_j(0) \ket j \ket j +i \sum_{j<k} (\sqrt{\kappa_{jk}}(x_{j}(0) -x_k(0))\ket j \ket k\right]\;.
\end{align}
First we rotate a qubit and apply the state preparation oracle as
\begin{align}
    \ket 0 &\mapsto \frac 1 {\sqrt{m_{\max}\alpha^2 + 2\kappa_{\max} d \, \beta^2 }} (\sqrt{m_{\max}} \alpha \ket 0 + i \sqrt{2\kappa_{\max} d} \beta \ket {1}) \nn
    \label{eq:stateprepfirststep}
    & \mapsto \frac 1 {\sqrt{m_{\max}\alpha^2 + 2\kappa_{\max} d \, \beta^2 }} (\sqrt{m_{\max}} \alpha \ket 0 \sket{\dot{\vec x}(0)}+ i \sqrt{2\kappa_{\max} d} \beta \ket {1} \sket{{\vec x}(0)}) \, .
\end{align}
The way we have described the oracle for $\dot{\vec{x}}(0)$ and $\vec x(0)$ allows for the case where $\alpha$ or $\beta$ may be zero, in which case $\cU$ can be arbitrary on that subspace because it has no effect on the state above.
Then our goal is to apply $\sqrt{\bM}$ to the $\sket{\dot{\vec x}(0)}$ portion, and ${\bf B}^\dagger \sqrt{\bM}$ to the $\sket{{\vec x}(0)}$ portion.

The most challenging part is for applying ${\bf B}^\dagger \sqrt{\bM}$.
To do this, we essentially apply the same construction in Lemma~\ref{lem:blockencodingB} in App.~\ref{app:accessmodel}, but for block encoding 
${\bf B}^\dagger \sqrt{\bM}$ rather than ${\bf B}^\dagger$.
First consider the operation where we prepare a superposition over $d$ values in an ancilla to give
\begin{equation}
    \sket{{\vec x}(0)} \frac 1{\sqrt{d}} \sum_{\ell=1}^{d} \ket{\ell}.
\end{equation}
This preparation is simple if $d$ is a power of 2, and otherwise can be performed using $\mathcal{O}(\log d)=\mathcal{O}(n)$ gates via amplitude amplification (see App.~E.2 of \cite{Sanders2020}).
We can then apply the oracle for the positions of non-zero entries of ${\bf K}$ to map $k$ to the non-zero entry in row $j$, and obtain
\begin{equation}
    \frac 1{\sqrt{d}}  \sum_{j=1}^N x_j(0) \ket{j}  \sum_{k\in[N]: \kappa_{jk}\ne 0} \ket{k}.
\end{equation}
There may be less than $d$ values of $k$ for each $j$ such that $\kappa_{jk}\ne 0$, but the oracle may be chosen to give dummy values of $k$ to pad it out to $d$.
These will later be eliminated because $\kappa_{jk}$ is actually zero for those values of $j,k$.

In the usual way, we can use the oracle for the matrix entries of $\kappa_{jk}$ to output its value, and perform an inequality test with an ancilla in an equal superposition to apply a factor corresponding to $\sqrt{\kappa_{jk}/\kappa_{\max}}$ (which is no greater than 1) \cite{SandersPRL18}.
This will give us, with an amplitude corresponding to the amplitude for success,
\begin{equation}
    \frac 1{\sqrt{\kappa_{\max} d}}  \sum_{j,k} \sqrt{\kappa_{jk}}\, x_j(0) \ket{j} \ket{k}.
\end{equation}
We then perform an inequality test, and perform a controlled swap based on the result of the inequality test.
This gives the state
\begin{equation}
    \frac 1{\sqrt{\kappa_{\max} d}} \left[ \sum_{j=1}^N \sqrt{\kappa_{jj}}\, x_j(0) \ket{j}  \ket{j}\ket{0}
    +\sum_{k>j} \sqrt{\kappa_{jk}}\, x_j(0) \ket{j} \ket{k}\ket{0}
    +\sum_{k<j} \sqrt{\kappa_{jk}}\, x_j(0) \ket{j} \ket{k}\ket{1} \right] ,
\end{equation}
where the third register is the qubit flagging the result of the inequality test.
Then we perform a $Z$ gate on that qubit and relabel the third sum to give
\begin{equation}
    \frac 1{\sqrt{\kappa_{\max} d}} \left[ \sum_{j=1}^N \sqrt{\kappa_{jj}}\, x_j(0) \ket{j}  \ket{j}\ket{0}
    +\sum_{j<k} \sqrt{\kappa_{jk}}\, \ket{j} \ket{k}[x_j(0)\ket{0}-x_k(0)\ket{1}] \right]  .
\end{equation}
Projecting onto the $\ket{+}$ state on the ancilla qubit then gives
\begin{equation}
    \frac 1{\sqrt{2\kappa_{\max} d}} \left[ \sum_{j=1}^N \sqrt{\kappa_{jj}}\, x_j(0) \ket{j}  \ket{j}
    +\sum_{j<k} \sqrt{\kappa_{jk}}\, [x_j(0)-x_k(0)] \ket{j} \ket{k} \right] ,
\end{equation}
where the amplitude is indicating the amplitude for success.

To apply $\sqrt{\bM}$ to the $\sket{\dot{\vec x}(0)}$ portion, we can simply perform the same inequality testing procedure from \cite{SandersPRL18} to apply a factor of $\sqrt{m_j/m_{\max}}$, and give
\begin{equation}
    \sket{\dot{\vec x}(0)} \mapsto \frac 1 {\sqrt{m_{\max}}} \sum_{j=1}^N \sqrt{m_j} \dot x_j(0) \ket j  .
\end{equation}
Combined, the preparation on these two parts of the state gives
\begin{align}
    & \frac 1 {\sqrt{m_{\max}\alpha^2 + 2\kappa_{\max} d \, \beta^2}} (\sqrt{m_{\max}} \alpha \ket 0 \sket{\dot{\vec x}(0)}+ i 2\sqrt{\kappa_{\max} d} \beta \ket {1} \sket{{\vec x}(0)}) \nn
    & \mapsto \frac 1 {\sqrt{m_{\max}\alpha^2 + 2\kappa_{\max} d \, \beta^2}} \left[ \sum_{j=1}^N \sqrt{m_j}\,  \dot x_j(0) \ket 0 \ket j \ket 0
\right. \nn & \quad \left. +i \sum_{j=1}^N \sqrt{\kappa_{jj}} \, x_j(0) \ket 1 \ket{j} \ket{j}  + i\sum_{j<k} \sqrt{\kappa_{jk}}\,  (x_j(0) - x_k(0)) \ket 1 \ket{j}   \ket{k}\right] , \label{eq:subnormsta}
\end{align}
which is the correct state (Eq.~\eqref{eq:initialstateapp}), but subnormalised indicating that it is not produced deterministically and we need to use amplitude amplification.

The number of amplitude amplification rounds scales as the inverse of the amplitude, so
the complexity in terms of $\alpha,\beta,E$ is
\begin{align}
    \cO\left( \sqrt{\frac{m_{\max}\alpha^2+2 \kappa_{\max} d\, \beta^2}{2 E}}\right)\, .
\end{align}
We assume we know the constants for this Lemma to simplify the amplitude amplification. It is also possible to perform amplitude amplification when the amplitude is unknown but bounded \cite{fixedpoint}.
Note that
\begin{align}
    \frac 1 2 m_{\max} \alpha^2 &= \frac 1 2 \sum_{j=1}^N m_{\max} (\dot x_j(0))^2  \\
    & =: K_{\max}\;,
\end{align}
which is the kinetic energy of the system of oscillators at $t=0$ if all masses were $m_j=m_{\max}$. Also,
\begin{align}
    \frac 1 2 \kappa_{\max} \beta^2 &= \frac 1 2 \sum_{j=1}^N \kappa_{\max} (x_j(0))^2 \\
    & =: U_{\max}
\end{align}
is the potential energy of a system of $N$ uncoupled oscillators at $t=0$ if all spring constants were $\kappa_{jj}=\kappa_{\max}$ and $\kappa_{jk}=0$ if $j \ne k$. If $E_{\rm max}=K_{\max}+U_{\max}=\frac {m_{\max}} 2 \|\dot{\vec{x}}(0))\|^2 +\frac {\kappa_{\max}} 2 \|\vec{x}(0)\|^2$ is the total energy of such a system, it satisfies
\begin{align}
     m_{\max} \alpha^2 + 2 \kappa_{\max} d\, \beta^2 \le 4 E_{\max} d\, .
\end{align}
The number of amplitude amplification rounds is then $\cO(\sqrt{E_{\max}d/E })$.
Each amplitude amplification round makes two uses of $\cS$, $\cS^\dagger$, and also one use of $\cU$ and $\cU^\dagger$.
Hence, the overall query complexity of our algorithm that prepares $\sket{\psi(0)}$ is
\begin{align}
  Q_{\rm ini}= \cO\left( \sqrt{\frac{E_{\max}d}{E} }\right) \;.
\end{align}

To find the accuracy of the block encoding we need to account for the error due to applying the factors of $\sqrt{\kappa_{jk}/\kappa_{\max}}$ and $\sqrt{m_j/m_{\max}}$ by inequality testing.
If we give the maximum errors in these factors as $\delta$ (corresponding to $\log(1/\delta)$ bits in the inequality testing), then the error in Eq.~\eqref{eq:subnormsta} is upper bounded as
\begin{align}
& \hspace{1.25em} \frac 1 {\sqrt{m_{\max}\alpha^2 + 2\kappa_{\max} d \, \beta^2 }} \ 
    \Biggl\| \sum_{j} \delta \, \sqrt{{m_{\max}}}\dot x_j(0) \ket 0 \ket j \ket 0
+ \sum_{j} \delta \, \sqrt{\kappa_{\max}} x_j(0) \ket 1 \ket{j} \ket{j} \nn  
& \hspace{13em}
+ \sum_{\substack{j<k\\ \kappa_{jk}\ne 0}} \delta \, \sqrt{\kappa_{\max}}(x_j(0) - x_k(0)) \ket 1 \ket{j} \ket{k}\Biggr\| \nn
& \le \frac 1 {\sqrt{m_{\max}\alpha^2 + 2\kappa_{\max} d \, \beta^2 }}\left[ m_{\max}\left\| \sum_{j=1}^N \delta \, \dot x_j(0) \ket j \right\|^2
+ \kappa_{\max} \left\| \sum_{\substack{j\le k\\ \kappa_{jk}\ne 0}} \delta \, x_j(0) \ket{j}   \ket{k}
-\sum_{\substack{j> k\\ \kappa_{jk}\ne 0}} \delta \, x_j(0) \ket{k} \ket{j}
\right\|^2 \right]^{1/2} \nn
&\le \frac \delta {\sqrt{m_{\max}\alpha^2 + 2\kappa_{\max} d \, \beta^2 }} \left[ m_{\max} \left\| \dot{\vec{x}}_j(0) \right\|^2
+ \kappa_{\max} \left\| \sqrt{2}\sum_{j,k;\kappa_{jk}\ne 0} x_j(0) \ket{j}   \ket{k}\right\|^2 \right]^{1/2}\nn
&\le \frac \delta {\sqrt{m_{\max}\alpha^2 + 2\kappa_{\max} d \, \beta^2 }} \left[ m_{\max}\left\| \dot{\vec{x}}(0) \right\|^2
+ 2\kappa_{\max} d\left\| \vec{x}(0)\right\|^2 \right]^{1/2} \nn
 &= \delta \; .
\end{align}
The inequality on the third line is obtained by noting that we can move from the state on the third line to that on the second line by the procedure described above, with an inequality test between $j$ and $k$, a controlled swap and phase, then projection onto $\ket{+}$ on the ancilla.
Thus, $\delta$ corresponds to the error \emph{before} amplitude amplification, and it is amplified by a factor of $\cO(\sqrt{E_{\max}d/E})$.

If we are aiming for error $\epsilon$ in the state preparation, we should therefore take
\begin{equation}
    \delta = \cO\left(\epsilon \sqrt{\frac{E}{E_{\max}d} } \right) .
\end{equation}
Because the state preparation by inequality testing uses squares, we have a gate complexity that is the square of $\log(1/\delta)$ for each step of amplitude amplification, and therefore the gate complexity from this source is
\begin{equation}
    \cO\left( \sqrt{\frac{E_{\max}d}{E} } \log^2 \left(\frac{E_{\max}d}{E \epsilon} \right)\right).
\end{equation}
Moreover, there are $\cO(\log(N))$ gates needed to implement an inequality test and controlled swap for each step of amplitude amplification.
Using $d\le N$, we can give the overall gate complexity as 
\begin{equation}
    \cO\left( \sqrt{\frac{E_{\max}d}{E} } \log^2 \left(\frac{E_{\max}N}{E \epsilon} \right)\right),
\end{equation}
which are the stated results. 
\end{proof}

The factor $E_{\rm max}/E$ in the complexity can become large when the masses $m_j$ or spring constants $\kappa_{jk}$ lack uniformity. In this case, the state $\sket{\psi(0)}$ is more complicated, and one has to ``pay extra'' for its preparation.
In addition, the query complexity in Lem.~\ref{lemma:stateprep} matches a lower bound for some instances. For example, consider the case where $x_j(0)=0$ for all $j \in[N]$, which implies $\beta^2=0$, and $\dot x_j(0)=1/\sqrt N$ for all $j \in [N]$, which implies $\alpha^2=1$. Let the masses be $m_j=1$, for some unknown $j$, and $m_{j'}=0$ otherwise. Then, $\sket{\psi(0)}$ is simply a state $\ket j$ in a computational basis, and our state preparation algorithm
would output the unknown $j$ with high probability, as in the unstructured search problem~\cite{Gro96}. A lower bound on the query complexity for this case is $\Omega(\sqrt N)$~\cite{bennett1997strengths}. Since $E_{\max}=\frac 1 2 m_{\max} \alpha^2=\frac 1 2$, $E=\frac 1 2 (\dot x_j(0))^2=\frac 1 {2N}$, and $d=1$ for this instance, the query complexity of our approach is  $\cO(\sqrt N)$, matching the lower bound.

\section{Other encodings}\label{app:otherencodings}

In Sec.~\ref{sec:quantumevolution}, we choose a specific encoding for the position and velocities as a quantum state $\ket \psi$ in Eq.~\eqref{eq:quantumstate}, which satisfies Eq.~\eqref{eq:schrodinger}. 
Let $\ket{\psi(t)} \in \mathbb C^{N+M}$ be written as
\begin{equation}
    \ket{\psi(t)}=\begin{pmatrix}
       \vec \nu(t)\\ i\vec \mu(t)
    \end{pmatrix},
\end{equation}
where $\vec \nu(t) \in \mathbb C^N$ and $i\vec \mu(t) \in \mathbb C^{M}$. We made the choice of $\vec \nu(t)= \dot{\vec y}(t)$ and $\vec \mu(t)={\bf B}^\dagger \vec y(t)$
to present our main results. Nevertheless, other choices also work, and we now discuss another solution (and problem). This provides us with a new encoding that is different from the one we used in Problem~\ref{problem:oscillators}, and is closely related to the encoding used in Ref.~\cite{Costa19}.

Let $\vec \nu(t)={\bf P} \vec y(t)$, where ${\bf P}$ projects out 
the components of $\vec y(t)$ corresponding to the null space of ${\bf A}$ (or ${\bf B}^\dagger$). 
Let $\vec \mu(t)=-{\bf B}^+ {\bf P}\dot{\vec y}(t)$,
where ${\bf B}^+$ is the Moore-Penrose pseudo-inverse of ${\bf B}$. This is an $M\times N$ matrix that
 satisfies ${\bf B}{\bf B}^+{\bf B}={\bf B}$ and
${\bf B}^+={\bf B}^\dagger ({\bf B}{\bf B}^\dagger)^+=
{\bf B}^\dagger {\bf A}^+$.
We claim that this choice also satisfies Eq.~\eqref{eq:schrodinger}. Specifically, we want it to satisfy
\begin{align}
\label{eq:schrodinger2}
\begin{pmatrix} \dot{ \vec \nu} (t) \cr  i\dot{\vec \mu} (t)\end{pmatrix} 
=-i {\bf H} \begin{pmatrix} \vec \nu (t) \cr  i{\vec \mu} (t)\end{pmatrix} 
= \begin{pmatrix}  -{\bf B} \vec \mu(t) \cr i {\bf B}^\dagger \vec \nu(t)\end{pmatrix}.
\end{align}

We then have
\begin{align}
\dot{\vec \nu}(t)
= {\bf P}\dot{\vec y}(t) 
= {\bf A}{\bf A}^+{\bf P}\dot{\vec y}(t) 
={\bf B} {\bf B}^+ {\bf P}\dot{\vec y}(t) 
= -  {\bf B} \vec \mu(t) \;,
\end{align}
which establishes the first component of Eq.~\eqref{eq:schrodinger2}. Note that $ {\bf A}{\bf A}^+{\bf P}={\bf P}$, since ${\bf A}^+$
does not act on eigenvectors of ${\bf A}$ of eigenvalue zero.
Using Newton's equations (Eq.~\eqref{eq:matrixnewton}), we have
\begin{align}
i \dot{\vec \mu}(t) 
= -i {\bf B}^+{\bf P}\ddot{\vec y}(t) 
 = i {\bf B}^+{\bf P} {\bf A}{\vec y}(t) 
 = i {\bf B}^+ {\bf A} {\bf P} {\vec y}(t) 
 = i {\bf B}^\dagger {\bf A}^+ {\bf A}  {\bf P} {\vec y}(t) 
 = i {\bf B}^\dagger {\bf P} {\vec y}(t) 
 = i {\bf B}^\dagger \vec \nu(t) \;,
\end{align}
which establishes the second component in Eq.~\eqref{eq:schrodinger2}.

Hence, evolution under ${\bf H}$ also allows us to solve the following variation of Problem~\ref{problem:oscillators} and~\ref{problem:harmonicapprox}:
\begin{problem}
\label{problem:oscillators2}
Let ${\bf A} \succcurlyeq 0$ be an $N \times N$ real-symmetric,  PSD, and $d$-sparse matrix, and ${\bf B}$ an $N\times M$  matrix that satisfies ${\bf B}{\bf B}^\dagger={\bf A}$.
Define the normalized state
\begin{align}
\label{eq:encodedstate3}
\ket{\psi(t)} := \frac 1 {\sqrt {2F}} \begin{pmatrix} {\bf P}  {\vec y}(t) \cr - i {\bf B}^+ {\bf P} \dot{\vec y}(t) \end{pmatrix} \;,
\end{align}
where $F>0$ is a constant, 
 ${\bf P}$ is the projector onto the subspace orthogonal to the null space of ${\bf A}$, and ${\bf B}^+$ is the Moore-Penrose pseudo-inverse of ${\bf B}$.
Assume we are given oracle access to ${\bf A}$ and to a unitary $\cW$ that prepares 
the initial state $\sket{\psi(0)}$.
Given $t$ and $\epsilon$, the goal is to output a state that is $\epsilon$-close to $\ket{\psi(t)}$ in Euclidean norm.
\end{problem}

Note that the normalizing constant is
\begin{align}
F= \frac 1 2 \vec y(t)^T {\bf P}\vec y(t) + \frac 1 2  \dot{\vec y}(t)^T {\bf P} ({\bf B}^+)^\dagger {\bf B}^+ {\bf P}\dot{\vec y}(t) \;,
\end{align}
which satisfies
\begin{align}
2\dot F&=\dot{\vec y}(t)^T {\bf P}\vec y(t) +{\vec y}(t)^T {\bf P}\dot{\vec y}(t) + \ddot{\vec y}(t)^T {\bf P} ({\bf B}^+)^\dagger {\bf B}^+ {\bf P} \dot{\vec y}(t)+ \dot{\vec y}(t)^T {\bf P} ({\bf B}^+)^\dagger {\bf B}^+ {\bf P}\ddot{\vec y}(t) \nn
& = \dot{\vec y}(t)^T {\bf P}\vec y(t) +{\vec y}(t)^T {\bf P}\dot{\vec y}(t) - {\vec y}(t)^T {\bf A} {\bf P} ({\bf B}^+)^\dagger {\bf B}^+ {\bf P} \dot{\vec y}(t)- \dot{\vec y}(t)^T {\bf P} ({\bf B}^+)^\dagger {\bf B}^+ {\bf P} {\bf A}  {\vec y}(t)\nn
& = \dot{\vec y}(t)^T {\bf P}\vec y(t) +{\vec y}(t)^T {\bf P}\dot{\vec y}(t) - {\vec y}(t)^T {\bf P} {\bf A} ({\bf B}^+)^\dagger {\bf B}^+ {\bf P} \dot{\vec y}(t) -
\dot{\vec y}(t)^T {\bf P} ({\bf B}^+)^\dagger {\bf B}^+ {\bf A} {\bf P}{\vec y}(t) \nn
& = \dot{\vec y}(t)^T {\bf P}\vec y(t) +{\vec y}(t)^T {\bf P}\dot{\vec y}(t) - {\vec y}(t)^T {\bf P} {\bf A} ({\bf A}^+)^\dagger {\bf B}{\bf B}^+ {\bf P} \dot{\vec y}(t) -
\dot{\vec y}(t)^T {\bf P} ({\bf B}^+)^\dagger {\bf B}^\dagger {\bf A}^+{\bf A} {\bf P}{\vec y}(t)\nn
& = \dot{\vec y}(t)^T {\bf P}\vec y(t) +{\vec y}(t)^T {\bf P}\dot{\vec y}(t) - {\vec y}(t)^T {\bf P} {\bf B}{\bf B}^+ {\bf P} \dot{\vec y}(t) -
\dot{\vec y}(t)^T {\bf P} ({\bf B}^+)^\dagger {\bf B}^\dagger  {\bf P}{\vec y}(t)\nn
& = \dot{\vec y}(t)^T {\bf P}\vec y(t) +{\vec y}(t)^T {\bf P}\dot{\vec y}(t) - {\vec y}(t)^T {\bf P}  \dot{\vec y}(t) -
\dot{\vec y}(t)^T {\bf P} {\vec y}(t)\nn
& = 0\;,
\end{align}
where we used ${\bf B}{\bf B}^+ {\bf P} ={\bf A} {\bf A}^+ {\bf P}={\bf P}$.
Then, $F$ is independent of $t$ and refers to another conserved quantity different from $E$.
We can use our quantum algorithm to prepare
$\ket{\psi(t)}$ in Eq.~\eqref{eq:encodedstate3} by simulating ${\bf H}$ in Eq.~\eqref{eq:hamiltonian}. The complexity of Hamiltonian simulation is the same as that in Thm.~\ref{thm:main}.

In the encoding used for Problem~\ref{problem:oscillators2}, some amplitudes
are proportional to ${\bf P}\vec y(t)$, and we have more direct access to the displacements of the oscillators
than that using the encoding in Problem~\ref{problem:oscillators}. However, 
preparing $\sket{\psi(0)}$ in this case
is expected to be more costly as it involves the action of ${\bf B}^+$. This is the reason why the complexity in the quantum algorithm of 
Ref.~\cite{Costa19} for simulating the wave equation can be dominated by that of the initial state preparation. 
That case is a special instance of Problem~\ref{problem:oscillators2} where ${\bf A}$
corresponds to the (discretized) Laplacian.

\section{Related work and optimality of our approach}\label{app:relatedwork}

Our quantum algorithm 
provides the solution to a second-order differential equation, i.e., $\ddot {\vec y}(t)=-{\bf A} {\vec y}(t)$, encoded in the amplitudes of a quantum state $\sket{\psi(t)}$ defined in Problems~\ref{problem:oscillators} and~\ref{problem:harmonicapprox}. We do this via the classical-to-quantum reduction given in Sec.~\ref{sec:quantumevolution}, which results in a first-order differential equation, corresponding to a Schr\"odinger equation with a time-independent Hamiltonian ${\bf H}$ that depends on ${\bf B}$, where ${\bf A}={\bf B}{\bf B}^\dagger$.
Prior works have also considered quantum algorithms for solving first-order differential equations
using a variety of approaches. A prominent example is Ref.~\cite{Ber14} and related results (cf.~\cite{childs2021high,Krovi2023}) that use the quantum algorithm for linear systems of equations~\cite{HHL,CKS17,CostaPRXQ22}. Here we argue that those approaches, while possibly applicable to our problem, generally yield inefficient quantum algorithms for Problems~\ref{problem:oscillators},~\ref{problem:kineticestimation},~\ref{problem:harmonicapprox}, and~\ref{problem:potentialestimation} due to issues arising from the encoding. Indeed, by formulating 
the problem as a Hamiltonian simulation problem, 
and using optimal methods for the latter, we argue that our approach is the optimal one for solving these problems. Among other useful features, our approach does not necessitate a clock register to encode the solution at all times as in Ref.~\cite{Ber14}. 
More details follow.

One standard approach to formulate
the second-order differential equation as a first-order one in our case would be to consider
\begin{align}
\label{eq:firstordera}
   \dot {\vec v}(t) =\begin{pmatrix} \bf 0 & -{\bf A}  \cr \one_N & {\bf 0}\end{pmatrix} {\vec v}(t) 
   \;,
\end{align}
where
\begin{align}
    {\vec v}(t) := \begin{pmatrix}
        \dot {\vec y}(t) \cr  {\vec y}(t)
    \end{pmatrix}
\end{align}
is a vector in $\mathbb R^{2N}$ that encodes the coordinates of all oscillators. (We assume that the choice of units is set from the beginning so that calculations are done with real numbers.)
Another standard approach would be to consider
\begin{align}
\label{eq:firstorderb}
   \dot {\vec w}(t) =\begin{pmatrix} \bf 0 & -\one_N \cr {\bf A} & {\bf 0}\end{pmatrix} {\vec w}(t) 
   \;,
\end{align}
where
\begin{align}
    {\vec w}(t) := \begin{pmatrix}
        \dot {\vec y}(t) \cr {\bf A}{\vec y}(t)
    \end{pmatrix}
\end{align}
is also a vector in $\mathbb R^{2N}$.
Equations~\eqref{eq:firstordera} and \eqref{eq:firstorderb} are also homogeneous 
first-order differential equations, whose solutions can be obtained by applying an exponential operator to $\vec v(0)$ or $\vec w(0)$. 
Note that $\frac {d}{dt}\|{\vec v}(t)\|\ne 0$ and $\frac {d}{dt}\|{\vec w}(t)\|\ne 0$
in general, so we cannot directly apply Hamiltonian simulation methods using these encodings, as the evolution of these vectors is not unitary.

To solve differential equations that do not conserve the norm,
Refs.~\cite{Ber14,childs2021high,Krovi2023}
propose a range of quantum algorithms that, for this application, would output (normalized) quantum states
\begin{align}
    \ket{\vec v(t)}&:=\frac 1 {\sqrt{\| \dot {\vec y}(t)\|^2 + \| {\vec y}(t)\|^2}}\begin{pmatrix}
        \dot {\vec y}(t) \cr  {\vec y}(t)
    \end{pmatrix} \; ,
\end{align}
or 
\begin{align}
    \ket{\vec w(t)}&:=\frac 1 {\sqrt{\| \dot {\vec y}(t)\|^2 + \|{\bf A} {\vec y}(t)\|^2}}\begin{pmatrix}
        \dot {\vec y}(t) \cr  {\bf A}{\vec y}(t)
    \end{pmatrix} \; .
\end{align}
In Refs.~\cite{Ber14,childs2021high,Krovi2023} this is done by approximating the solutions to the differential equations as solutions to systems of linear equations. The complexity of this approach depends on several parameters, in particular the condition number of the matrix to be inverted. 
In Ref.~\cite{An2022b} a solution is found by giving the solution to the differential equation in exponential form, and then using the LCU approach to approximate the exponential; a related approach that approximates the exponential operator is given in Ref.~\cite{CS16}. 
That approach is not applicable here, because it requires the matrix to be normal.
Here the matrices -- the 2$\times$2 block matrices including ${\bf A}$ -- are not normal, so the approach cannot be used.
Nevertheless, even disregarding
such difficulties, we show that this way of encoding the coordinates of the oscillators as in Eqs.~\eqref{eq:firstordera} or~\eqref{eq:firstorderb} is readily problematic.

To observe this, 
it suffices to consider the evolution of a single normal mode of frequency $\omega_k$, i.e., the eigenvector of ${\bf A}$ of eigenvalue $(\omega_k)^2$. We obtain
\begin{align}
    \vec y(t) = \vec{a}_k \cos(t \omega_k + \phi_k)\;,
\end{align}
where $\phi_k \in \mathbb R$ is the initial phase and $\vec a_k \in \mathbb R^N$ is the eigenvector of ${\bf A}$. 
Computing the time average we obtain
\begin{align}
    \langle \| \vec y(t)\|^2 \rangle 
     &= \frac 1 2
    \| \vec a_k\|^2
\end{align}
and
\begin{align}
    \langle \| \dot {\vec y}(t)\|^2 \rangle 
    &= \frac 1 2
    |\omega_k|^2\| \vec a_k\|^2 \nn
    & =   |\omega_k|^2  \langle \| \vec y(t)\|^2 \rangle \;.
    \end{align}
    (The latter is also the time average of the  kinetic energy.)
In a case where ${\bf A}$ gives rise to normal modes of low frequency, we have $|\omega_k| \ll 1$ (in the corresponding units), and then
\begin{align}
  \langle \| \dot{\vec y}(t)\|^2 \rangle \ll  \langle \| \vec y(t)\|^2 \rangle \;.
\end{align}
The implication is that the support 
of $\ket{\vec v(t)}$ is, on average, mostly concentrated 
on the subspace spanned by basis vectors that encode $\vec y(t)$. These do not 
encode the terms appearing in the kinetic energy or $\dot{\vec y}(t)$.
Hence, the complexity of solving Problem~\ref{problem:kineticestimation} using this encoding,
where $\ket{\vec v(t)}$ is prepared rather than $\ket{\psi(t)}$,
is at least linear in $1/|\omega_k|$, which is the factor needed to increase the desired amplitudes to a constant.  
This can become unbounded if $|\omega_k| \rightarrow 0$. 

In the second case we note that the support of $\ket{\vec w(t)}$
is, on average, mostly concentrated on the subspace spanned by basis vectors that encode $\dot{\vec y}(t)$ instead. That is, the case $\langle \|{\bf A}\vec y (t) \|^2\rangle \ll \langle \| \dot{\vec y }(t) \|^2\rangle $ can occur when there are normal modes of low frequency. To solve Problem~\ref{problem:oscillators} or Problem~\ref{problem:harmonicapprox}, and especially Problem~\ref{problem:potentialestimation}, the second component of the state must be proportional to $ {\bf B}^\dagger \vec y(t)$ or $\sqrt{\bf A} \vec y(t)$, which requires the application of the pseudo-inverse of ${\bf B}$ or ${\sqrt{\bf A}}$ to $ \ket{\vec w(t)}$. The condition number of these pseudo-inverses is also polynomial in $1/|\omega_k|$; for example $\|(\vec a_k)^T {\bf B}\|=|\omega_k| \ll 1$ for low frequencies. This implies that the complexity of solving 
Problem~\ref{problem:oscillators}, Problem~\ref{problem:harmonicapprox}, or Problem~\ref{problem:potentialestimation} using this encoding, 
where $\ket{\vec w(t)}$ is prepared rather than $\ket{\psi(t)}$, is at least linear in $1/|\omega_k|$, which is the complexity of the most efficient methods to implement the pseudo-inverse~\cite{CKS17}.
As in the previous case, this can also become unbounded
if $|\omega_k| \rightarrow 0$.

Similar observations follow directly from the results for the complexity in Refs.~\cite{Ber14,childs2021high}.
There the complexity is polynomial in the condition number of the matrix that diagonalises the matrix appearing in the differential equations, which is here that with the blocks $-{\bf A}$ and $\one_N$ in Eq.~\eqref{eq:firstordera}, or the negative of these for $\vec w(t)$ in Eq.~\eqref{eq:firstorderb}.
The matrix that diagonalises this matrix has condition number
\begin{equation}
    \max(|\omega_k|,|\omega_k|^{-1}) \, .
\end{equation}
Thus it will become unbounded when $|\omega_k| \rightarrow 0$, as described above.
Technically, we cannot use the result as given in Ref.~\cite{Ber14} directly, because the stability condition it uses requires the matrices appearing in the differential equations to not have eigenvalues exactly on the imaginary axis, as is the case here.

An alternative approach to avoid dependence on the condition number of the diagonalising matrix is that using the log-norm \cite{Krovi2023}.
There we would need to subtract the identity times the log-norm from the original matrix.
Here, the log-norm is
\begin{equation}
    \max |1-\omega_k^2|/2 \, .
\end{equation}
This means it is always positive, and so subtracting a multiple of the identity would result in a norm that exponentially decreases.
Hence the approach of Ref.~\cite{Krovi2023} would yield a complexity that is exponential in time.

At a high level, our encoding is motivated by the conservation of energy 
in time, and places equal emphasis on those terms that encode the kinetic energy ($\dot{\vec y}(t)$) and the potential energy (${\bf B}^\dagger{\vec y}(t)$) of the oscillators. The other two encodings  
do not have this feature: $\ket{\vec v(t)}$ underrepresents the kinetic energy terms and $\ket{\vec w(t)}$ underrepresents the potential energy terms, bringing the issues
discussed above.

Last, we provide a comment on the relation 
between our results and
a closely related result in Ref.~\cite{Costa19}
for simulating the wave equation. 
In Ref.~\cite{Costa19} it is shown that the wave equation, a second-order and homogeneous differential equation, can also be mapped to a Schr\"odinger equation. Their construction is related to ours in that it uses a factorization of the (discrete) Laplacian as ${ L}={\bf B}{\bf B}^\dagger$, and indeed it is well-known that the wave equation is one example of a classical system of coupled oscillators where all masses and springs are uniform, and where the couplings between oscillators are geometrically local (e.g., on the square grid). However, the encoding used in Ref.~\cite{Costa19} is different to ours; for example, $N$ amplitudes are reserved to encode the intensity of the wave, which corresponds to $\vec y(t)$ in our case. (The component of $\vec y(t)$ on the eigenvector of eigenvalue 0 of $L$ is projected out.) 
The other amplitudes are proportional to ${\bf B}^+ \dot{\vec y}(t)$, where ${\bf B}^+$ is the pseudo-inverse of ${\bf B}$. 
That is, the state prepared by the algorithm 
of Ref.~\cite{Costa19} is of the form
\begin{align}
    \ket{\phi(t)} \propto \begin{pmatrix} \vec y(t) \cr
        {\bf B}^+ \dot{\vec y}(t)
    \end{pmatrix} \;,
\end{align}
where the constant of proportionality is also time-independent. This is essentially the same encoding as the one discussed in Appendix~\ref{app:otherencodings}, since ${\bf P}\vec y(t)=\vec y(t)$ by assumption. The length of this vector is preserved in time since the evolution is unitary. Nevertheless,
one implication of using this encoding to solve Problem~\ref{problem:oscillators} or Problem~\ref{problem:kineticestimation} for this example is the need to invert ${\bf B}^+$ initially, to obtain amplitudes proportional to ${\bf B}^+\dot{\vec y}(0)$. The condition number can be large, i.e., it is $\cO(\poly(N))$ for the example of the wave equation, and hence the initial state preparation step ($t=0$)
can be inefficient.
Indeed, Ref.~\cite{Costa19} manages to give evidence of a polynomial quantum speedup only.
In addition, when considering the wave equation, the system is geometrically local and implies that $N$ is polynomial in the evolution time $t$ (i.e, $t$ is exponential in $n$). In our problem, however, we can treat a significantly larger set of instances: we do not impose uniform masses, uniform spring constants, or even geometrically-local interactions, and we can allow for times $t=\cO(\poly(n))$. This generality together with our improved encoding are key to obtaining our claimed exponential quantum speedups for the problems defined.

\end{document}